\documentclass{birkjour}

\usepackage[utf8]{inputenc}
\usepackage[english]{babel}
\usepackage{csquotes}
\usepackage{bbm}
\usepackage{slashed}
\usepackage{mathtools}
\usepackage{lmodern}
\mathtoolsset{showonlyrefs}
\usepackage{enumitem}
\usepackage{hyphenat}
\usepackage{stmaryrd}
\usepackage{xpatch}
\usepackage[giveninits=true,sorting=nyt,backend=biber,maxalphanames=3,style=alphabetic,citestyle=alphabetic]{biblatex}
\DeclareNameAlias{author}{last-first}
\DeclareNameAlias{editor}{last-first}

\renewbibmacro{in:}{%
  \ifentrytype{article}{}{\printtext{In\intitlepunct}}}

\renewbibmacro*{volume+number+eid}{%
  {\bf\printfield{volume}}%
  \printfield{number}%
  \setunit{\addcomma\space}%
  \printfield{eid}}
\DeclareFieldFormat[article]{number}{\mkbibparens{#1}}
\DeclareFieldFormat{pages}{#1}

\xpatchbibdriver{article}
  {\newunit
   \usebibmacro{note+pages}}
  {}{}{}

\renewbibmacro*{journal+issuetitle}{%
  \usebibmacro{journal}%
  \setunit*{\addspace}%
  \iffieldundef{series}
    {}
    {\newunit
     \printfield{series}%
     \setunit{\addspace}}%
  \usebibmacro{volume+number+eid}%
  \newunit
  \usebibmacro{note+pages}%
  \setunit{\addspace}%
  \usebibmacro{issue+date}%
  \setunit{\addcolon\space}%
  \usebibmacro{issue}%
  \newunit}  

\newbibmacro*{byeditor:in}{%
  \ifnameundef{editor}
    {}
    {\printnames[editorin]{editor}%
     \addspace\bibsentence%
     \mkbibparens{eds.}%
     \clearname{editor}%
     \printunit{\addcomma\space}}}

\xpatchbibdriver{incollection}
  {\usebibmacro{in:}%
   \usebibmacro{maintitle+booktitle}%
   \newunit\newblock
   \usebibmacro{byeditor+others}}
  {\usebibmacro{in:}%
   \usebibmacro{byeditor:in}%
   \setunit{\labelnamepunct}\newblock
   \usebibmacro{maintitle+booktitle}%
   \newunit\newblock
   \usebibmacro{byeditor}}
  {}{}

\xpatchbibdriver{incollection}
  {\newunit\newblock
   \usebibmacro{chapter+pages}}
  {}
  {}{\typeout{failed to patch bibmacro incollection to remove page macro}}

\xpatchbibdriver{incollection}
  {\usebibmacro{series+number}}
  {\usebibmacro{series+number}
   \newunit\newblock
   \usebibmacro{chapter+pages}}
  {}{\typeout{failed to patch bibmacro incollection to re-add page macro}}  

\DeclareFieldFormat[incollection]{year}{\mkbibparens{#1}}

\PassOptionsToPackage{hyphens}{url}
\usepackage{url}
\usepackage{color,hyperref}
\usepackage[hyphenbreaks]{breakurl}
\definecolor{darkblue}{rgb}{0.0,0.0,0.5}
\hypersetup{
  linkcolor  = darkblue,
  citecolor  = darkblue,
  urlcolor   = darkblue,
  colorlinks = true,
}

\makeatletter
\AtBeginDocument{
  \hypersetup{
    pdftitle = {\@title},
    pdfauthor = {\@author}
  }
}
\makeatother

\usepackage{smartref}
\addtoreflist{section}

\newcommand{\N}{\mathbb{N}}

\newcommand{\C}{\mathbb{C}}
\newcommand{\Z}{\mathbb{Z}}
\newcommand{\R}{\mathbb{R}}

\newcommand{\ri}{\mathrm{i}}
\newcommand{\rd}{\mathrm{d}}

\newcommand{\Span}{\mathrm{span}}

\newcommand{\ret}{\mathrm{ret}}
\newcommand{\adv}{\mathrm{adv}}

\newcommand{\rEG}{\mathrm{EG}}
\newcommand{\rint}{\mathrm{int}}
\newcommand{\const}{\mathrm{const}}

\DeclareMathOperator{\sgn}{sgn}
\DeclareMathOperator{\Texp}{Texp}

\DeclareMathOperator{\T}{T}

\DeclareMathOperator{\aT}{\overline{T}}
\DeclareMathOperator{\Ret}{R}

\DeclareMathOperator{\Adv}{A}

\DeclareMathOperator{\Dif}{D}
\DeclareMathOperator{\Gre}{G}
\DeclareMathOperator{\Wig}{W}

\DeclareMathOperator{\sd}{sd}

\DeclareMathOperator{\ext}{e}
\DeclareMathOperator{\der}{d}
\DeclareMathOperator{\CC}{\mathbf{c}}
\newcommand{\cL}{\mathcal{L}}
\newcommand{\cS}{\mathcal{S}}
\newcommand{\cD}{\mathcal{D}}
\newcommand{\Fa}{\mathcal{F}}
\newcommand{\Fh}{\mathcal{F}^\mathrm{hom}}
\newcommand{\id}{\mathbbm{1}}

\newcommand{\supp}{\mathrm{supp}}
\newcommand{\zerop}{0}
\newcommand{\normord}[1]{:\mathrel{#1}:}
\newcommand{\mP}[1]{\frac{\rd^4 #1}{(2\pi)^4}}
\newcommand{\mH}[2]{\rd\mu_{#1}(#2)}
\newcommand{\F}[1]{\tilde{#1}}
\newcommand{\FF}[1]{\widetilde{#1}}

\makeatletter
\newcommand{\leqnomode}{\tagsleft@true}
\newcommand{\reqnomode}{\tagsleft@false}
\makeatother

\newtheorem{thm}{Theorem}
\newtheorem{dfn}[thm]{Definition}
\newtheorem{asm}{Assumption}
\newtheorem{lem}[thm]{Lemma}

\newtheorem{rem}[thm]{Remark}

\numberwithin{equation}{section}
\numberwithin{thm}{section}
\numberwithin{asm}{section}

\usepackage{titletoc}

\begin{document}

\title[Weak Adiabatic Limit in QFTs with Massless Particles]{Weak Adiabatic Limit in Quantum Field \\ Theories with Massless Particles}

\author[P. Duch]{Pawe{\l} Duch}
\address{Institute of Physics\\ Jagiellonian University\\
  \L ojasiewicza 11\\ 30-348 Krak\'ow\\ Poland}
\email{pawel.duch@uj.edu.pl}

\begin{abstract}
We construct the Wightman and Green functions in a~large class of models of perturbative QFT in the four-dimensional Minkowski space in the Epstein--Glaser framework. To this end we prove the existence of the weak adiabatic limit, generalizing the results due to Blanchard and Seneor. Our proof is valid under the assumption that the time-ordered products satisfy certain normalization condition. We show that this normalization condition may be imposed in all models with interaction vertices of canonical dimension 4 as well as in all models with interaction vertices of canonical dimension 3 provided each of them contains at least one massive field. Moreover, we prove that it is compatible with all the standard normalization conditions which are usually imposed on the time-ordered products. The result applies, for example, to quantum electrodynamics and non-abelian Yang--Mills theories.
\end{abstract}

\maketitle

\titlecontents{section}
  [0.1em]{}{\hspace{-1em}}{}{\titlerule*[0.5pc]{.}\contentspage}
\titlecontents{subsection}
 [0.1em]{}{\hspace{-1em}}{}{\titlerule*[0.5pc]{.}\contentspage}
\renewcommand{\contentsname}{{\begin{flushleft}{\bf\large Contents}\end{flushleft}}}
\tableofcontents

\section{Introduction}

Relativistically covariant perturbative quantum field theory (QFT) in the Minkowski space is one of the most successful modern physical theories. Its models, especially the Standard Model, have been repeatedly tested in experiments and provide a~very accurate description of almost all phenomena in particle physics. The fundamental objects of interest in QFT in the flat spacetime are the vacuum expectation values (VEVs) of the products and the time-ordered products of the interacting fields called the Wightman and Green functions, respectively. Conceptually, the Wightman functions are of more fundamental character. They are used for examining the basic physical structure of a~theory. Their properties: Poincar{\'e} covariance, spectrum condition, hermiticity, locality, positivity and clustering, reflect the basic assumptions of relativistic QFT. Moreover, by the Wightman reconstruction theorem~\cite{streater2000pct} their knowledge allows to obtain the operator formulation of the theory. On the other hand, the Green functions are of more practical importance and it is usually easier to obtain them in the perturbation theory. Furthermore, by the LSZ reduction formula~\cite{lehmann1955formulierung} they are directly related to the S-matrix elements. Because of this they have been studied in the physical literature more extensively than the Wightman functions.

The perturbative construction of the Green functions in a large class of models was given by Lowenstein~\cite{lowenstein1976convergence}, who generalized the BPHZ method~\cite{bogoliubov1957multiplication,hepp1966proof,zimmermann1969convergence}, and by Breitenlohner and Maison~\cite{breitenlohner1977dimensionally1,breitenlohner1977dimensionally2}, who used the dimensional regularization~\cite{bollini1972dimensional,veltman1972regularization}. Their results apply to all models with interaction vertices of infrared dimension equal at least to $4$. The infrared dimension is the additive number calculated by assigning dimension $2$ to massive fields, dimension $\frac{3}{2}$ to massless Dirac fields and dimension $1$ to massless boson fields and each derivative. Since the infrared dimension is not lower than the canonical dimension, the Green functions exist in particular in all models with interaction vertices of the canonical dimension equal to~$4$. Let us recall that the canonical dimension is the additive number which is used in the classification of models of QFT into super-renormalizable, renormalizable and non-renormalizable models. It is calculated by assigning dimension $1$ to each derivative as well as scalar and vector fields\footnote{In the case of vector fields we always use the Gupta--Bleuler formalism and the Feynman gauge.} and dimension $\frac{3}{2}$ to Dirac fields.

The Wightman functions have not been defined with the use of the momentum space methods. However, both the Wightman and Green functions can be constructed in the Epstein--Glaser (EG) approach~\cite{epstein1973role}. The basic objects in this approach are the time-ordered products of free fields, which satisfy a set of axioms formulated in position representation. The EG method is based on the observation that the time-ordered product of $n+1$ local fields at points $x_1,\ldots,x_{n+1}$ is determined uniquely by the time-ordered products of at most $n$ fields for all test functions with support away from the main diagonal $x_1=\ldots=x_{n+1}$. The solution of the ultraviolet (UV) problem in the EG approach amounts to finding an extension of the above distribution to the space of all test functions which is compatible with the axioms. In order to define an interacting model one specifies its interaction vertices $\mathcal{L}_1,\ldots,\mathcal{L}_{\mathrm{q}}$, which are polynomials in free fields and their derivatives. To each interaction vertex $\mathcal{L}_l$ we associate a~coupling constant $e_l$ and a switching function $g_l\in\mathcal{S}(\R^4)$. To avoid the infrared (IR) problem in the definition of the scattering matrix and the interacting fields one multiplies every interaction term by a~switching function belonging to the Schwartz class. In order to define the physical scattering matrix and the physical interacting fields one has to remove this IR regularization by taking the so-called adiabatic limit. 

Let $\mathbf{g}$ denote the list of the switching functions $g_1,\ldots,g_\mathrm{q}$. The IR-regularized scattering matrix is given by
\begin{equation}
 S(\mathbf{g}) = \Texp\left(\ri  \int\rd^4 x \, \sum_{l=1}^{\mathrm{q}}\,g_l(x)\mathcal{L}_l(x)  \right).
\end{equation}
The IR-regularized interacting fields $C_\rint(\mathbf{g};x)$ and their time-ordered products $\T(C_{1,\rint}(\mathbf{g};x_1),\ldots, C_{m,\rint}(\mathbf{g};x_m))$ are defined using the Bogoliubov formulas (cf. Section \ref{sec:W_G_IR}). The physical scattering matrix $S$ is obtained in the following way
\begin{equation}\label{eq:intro_strong}
 S\Psi =  \lim_{\epsilon\searrow 0} S(\mathbf{g}_\epsilon)\Psi,
\end{equation}
where $\Psi$ is any vector from some suitably chosen dense subspace of the Fock space. The physical Wightman and Green functions are defined as the limits of the vacuum expectation values of the products or the time-ordered products of the interacting fields with IR regularization
\begin{equation}\label{eq:intro_W_G}
\begin{gathered}
\Wig(C_1(x_1),\ldots,C_m(x_m)):=
 \lim_{\epsilon\searrow 0}\, (\Omega|C_{1,\rint}(\mathbf{g}_\epsilon;x_1)\ldots C_{m,\rint}(\mathbf{g}_\epsilon;x_m)\Omega),
 \\[4pt]
\Gre(C_1(x_1),\ldots,C_m(x_m)):=
 \lim_{\epsilon\searrow 0}\, (\Omega|\!\T(C_{1,\rint}(\mathbf{g}_\epsilon;x_1),\ldots,C_{m,\rint}(\mathbf{g}_\epsilon;x_m))\Omega).
\end{gathered}
\end{equation}
\\
By definition $\mathbf{g}_\epsilon=(g_{1,\epsilon},\ldots,g_{\mathrm{q},\epsilon})$ and $g_{l,\epsilon}(x):=g_l(\epsilon x)$, where $g_l$ is a~Schwartz function such that $g_l(0)=1$ for all $l\in\{1,\ldots,\mathrm{q}\}$. Note that in the limit $\epsilon\searrow0$ the interaction is adiabatically turned on and off. We say that the strong adiabatic limit exists if in each order of the perturbation theory the limit \eqref{eq:intro_strong} exists for any $\Psi$ from the chosen domain and its value is independent of the choice of the Schwartz functions $g_1,\ldots,g_\mathrm{q}$. Similarly, the weak adiabatic limit exists if in each order of the perturbation theory the limits \eqref{eq:intro_W_G} exist in the sense of Schwartz distributions and their values are independent of the choice of the Schwartz functions $g_1,\ldots,g_\mathrm{q}$.

The existence of the weak adiabatic limit in purely massive theories was shown by Epstein and Glaser in~\cite{epstein1973role}. This result was subsequently extended to the case of quantum electrodynamics (QED) and the massless $\varphi^4$ theory by Blanchard and Seneor~\cite{blanchard1975green}. The proof of the existence of the weak adiabatic limit in a~more general class of models has not been given in the literature.  It is the aim of the present paper to fill this gap. For purely massive theories the existence of the strong adiabatic limit was shown by Epstein and Glaser in~\cite{epstein1976adiabatic}. Because of the infrared problem this limit does not exist in most theories with massless particles, e.g., in QED. Adiabatic limit of the inclusive cross-sections in low orders of the perturbation theory of QED was considered in~\cite{scharf2014} (see also~\cite{dutsch1993infrared,dutsch1993vertex}). Finally, the existence of the expectation values of the products of the interacting fields in thermal states has been recently proved by Fredenhagen and Lindner in the EG framework with the use of the time-slice axiom and the time-averaged Hamiltonian~\cite{fredenhagen2014construction} (see also~\cite{lindner2013perturbative,drago2015generalised}). Let us also note that the method of Epstein and Glaser has been reformulated and generalized in a number of ways. It is the basis of QFT in curved spacetime which is used to construct interacting quantum fields propagating on an arbitrary globally hyperbolic spacetime~\cite{brunetti2000microlocal,hollands2001local,hollands2002existence}. The ideas of Epstein and Glaser combined with the formalism of deformation quantization and functional approach led also to the foundation of the perturbative algebraic quantum filed theory --- the formulation of perturbative QFT in terms of abstract local algebras~\cite{fredenhagen2015perturbative,rejzner2016perturbative}.

The main result of the paper is the proof of the existence of the weak adiabatic limit in the four-dimensional Minkowski space in models with the interaction vertices $\mathcal{L}_1,\ldots,\mathcal{L}_\mathrm{q}$ satisfying one of the following conditions: 
\begin{enumerate}[leftmargin=*,label={(\arabic*)}]
\item $\forall_{l\in\{1,\ldots,\mathrm{q}\}}$ $\dim(\mathcal{L}_l)=4$ or 
\item $\forall_{l\in\{1,\ldots,\mathrm{q}\}}$ $\dim(\mathcal{L}_l)=3$ and $\mathcal{L}_l$ contains at least one massive field, 
\end{enumerate}
where $\dim(B)$ is the canonical dimension of the polynomial $B$ (the case of models with the interaction vertices of mixed dimensions is briefly discussed in Appendix \ref{sec:general}). It is a generalization of the results due to Blanchard and Seneor~\cite{blanchard1975green} mentioned in the previous paragraph. Our result provides a method for the construction of the Wightman and Green functions in the class of models given above. To our knowledge this is the first construction of the Wightman functions in non-abelian Yang--Mills theories.

The general structure of our proof of the existence of the weak adiabatic limit resembles that of~\cite{blanchard1975green}. The most important distinguishing feature, which allows us to obtain a significantly more general result, is the use of the notion of regularity of a distribution expressed in terms of the generalized big O notation $\F{t}(q,q')=O^\textrm{dist}(|q|^\delta)$ (cf. Section~\ref{sec:math}). Like the proof in~\cite{blanchard1975green}, our proof needs specific normalization of some of the time-ordered products corresponding to Feynman diagrams with only massless external lines. The required normalization condition \ref{norm:wAL} (the {\bf N}ormalization condition sufficient for the existence of the {\bf w}eak {\bf A}diabatic {\bf L}imit) is formulated in Section~\ref{sec:PROOF}. We prove that this condition is compatible with all the standard normalization conditions like for example Poincar{\'e} covariance. In particular, in the case of QED the condition \ref{norm:wAL} is compatible with all the normalization conditions introduced in~\cite{dutsch1999local}. In purely massless models the condition \ref{norm:wAL} is a consequence of almost homogeneous scaling of the VEVs of the time-ordered products. This implies the compatibility of \ref{norm:wAL} with all the normalization conditions used in~\cite{hollands2008renormalized} in the definition of non-abelian Yang--Mills theories without matter (in the special case of the flat spacetime). We also prove that the following condition has to be satisfied for the weak adiabatic limit to exist: the self-energies of the massless fields (which are used in the definition of a given model) have to be normalized such that the physical masses of these fields vanish. We call this condition the correct mass normalization of massless fields. Note that it follows from the condition \ref{norm:wAL}. The Wightman and Green functions cannot be defined by means of the weak adiabatic limit in models in which the correct mass normalization of massless fields is not possible. An example of such model is the massless $\varphi^3$ theory.

Using the technique of the proof of the existence of the weak adiabatic limit we also show that it is possible to significantly restrict the freedom in the definition of the time-ordered products by imposing the condition which we call the central normalization condition. In the case of QED the time-ordered product defined in this way satisfy all the standard normalization conditions, including the Ward identities. The corresponding advanced and retarded products constitute the so-called central solution of the splitting problem~\cite{epstein1973role,scharf2014}, whose existence up to now was only established in purely massive models~\cite{epstein1973role}.

In the framework of perturbative algebraic quantum field theory  \cite{fredenhagen2015perturbative,rejzner2016perturbative} the existence of the weak adiabatic limit allows to define a real, normalized and Poincar\'e-invariant functional on the algebra of interacting fields obtained by means of the algebraic adiabatic limit. In the case of models which are constructed on the Fock space with a positive-definite covariant inner product this functional is positive and can be interpreted as an interacting vacuum state. 

The plan of the paper is as follows. In Section \ref{Sec:EG_approach} we describe the EG approach to the perturbation theory, present some auxiliary results and introduce the notation that is used throughout the paper. Section \ref{Sec:existence_wAL} contains the main results, including the proof of the existence of the weak adiabatic limit, which is outlined in Section~\ref{sec:idea} and given in Section~\ref{sec:PROOF}. In Section \ref{Sec:comp} we study the compatibility of the condition \ref{norm:wAL}, needed in our proof of the existence of the weak adiabatic limit, with the standard normalization conditions usually imposed on the time-ordered products. In Section~\ref{sec:central} we formulate the central normalization condition. Section~\ref{sec:properties} contains the list of the properties of the Wightman and Green functions constructed in the paper. In Section \ref{sec:vacuum_state} we construct a Poincar{\'e}-invariant functional on the algebra of interacting fields. In Appendix~\ref{sec:mass} we show that the correct mass normalization of massless fields is necessary for the existence of the weak adiabatic limit in the second order of perturbation theory. In Appendix~\ref{sec:general} we consider the case of models with the interaction vertices of different dimensions. Appendix \ref{sec:magic_formula} contains the comparison of the Green functions defined in the Epstein--Glaser framework and the Green functions obtained with the use of the Gell-Mann and Low formula. 

We use the $(+---)$ signature of the Minkowski metric. By $\overline{V}^+$ we denote the closed cone of the future-directed timelike and lightlike vectors. Our convention for the Fourier transform of Schwartz distributions is
\begin{equation}
 \tilde{t}(q):=\int\rd^N x\, \exp(\ri q \cdot x) t(x),~~~~~~t(x)=\int\frac{\rd^N q}{(2\pi)^N}\, \exp(-\ri q \cdot x) \tilde{t}(q).
\end{equation}

\section{Epstein--Glaser approach}\label{Sec:EG_approach}

In this section we outline the EG approach to perturbative QFT in the four-dimensional Minkowski space. We follow closely the exposition from~\cite{epstein1973role} and treat fields as operator-valued distributions defined on a suitable domain of the Fock space. For details we refer the reader to~\cite{epstein1973role,scharf2014,scharf2016gauge}.

\subsection{Algebra \texorpdfstring{$\Fa$}{F} of symbolic fields}\label{sec:ff}

In perturbative QFT the interacting models are built with the use of the free fields. Let
\begin{equation}\label{eq:basic_gen}
 \mathcal{G}_0:=\{A_1,\ldots,A_\mathrm{p}\}\simeq\{1,\ldots,\mathrm{p}\}
\end{equation}
be the set of symbols denoting types of free fields needed for the definition of a~model under consideration. The elements of this set are called the {\it basic generators}. All components of vector or spinor fields are included in this set as separate symbols (if fields of these types are present in the model). We assume that an involution denoted by ${}^*$ is defined in~$\mathcal{G}_0$. It means that if a~charged field $A_i$ belongs to $\mathcal{G}_0$, then also its hermitian conjugation denoted by $A^{*}_i$ belongs to it, i.e. $A_i^{*}=A_{i'}$ for some $i'\in\{1,\ldots,\mathrm{p}\}$. For example, in the case of QED there are 12 basic generators: the four components of the vector potential $A_\mu=A^{*}_\mu$, and the four components of the spinor field $\psi_a$ and its hermitian conjugate $\psi_a^*=(\overline{\psi}\gamma^0)_a$. 

The basic generators \eqref{eq:basic_gen} supplemented with the symbols corresponding to their derivatives form the set of the {\it generators}
\begin{equation}\label{eq:gen}
 \mathcal{G}:=\{\partial^\alpha\! A_i:~i\in\{1,\ldots,\mathrm{p}\}, \alpha\in\N_0^4 \}\simeq\{1,\ldots,\mathrm{p}\}\times \N_0^4,
\end{equation}
where $\alpha$ is a~multi-index. Note that $\mathcal{G}$ may be identified with the set of pairs $(i,\alpha)$ where $i\in\{1,\ldots,\mathrm{p}\}$ and $\alpha\in\N_0^4$. We set $(\partial^\alpha\! A_i)^*:=\partial^\alpha\! A^*_i$. To every generator we associate the following quantum numbers:
\begin{enumerate}[leftmargin=*,label={(\arabic*)}]
 \item fermion number $\mathbf{f}(\partial^\alpha\! A_i)\in\Z/2\Z$, 
 \item canonical dimension $\dim(\partial^\alpha\! A_i)\in\frac{1}{2}\Z$,
 \item electric charge number, barion number, lepton number, etc.
\end{enumerate}
For example in the case of QED we have $\mathbf{f}(\partial^\alpha\! A_\mu)=0$, $\mathbf{f}(\partial^\alpha\! \psi_a)=\mathbf{f}(\partial^\alpha\! \psi^*_a)=1$, $\dim(\partial^\alpha\! A_\mu)=1+|\alpha|$, $\dim(\partial^\alpha\! \psi_a)=\dim(\partial^\alpha\! \psi^*_a)=\frac{3}{2}+|\alpha|$, $\mathbf{q}(\partial^\alpha\! A_\mu)=0$, $\mathbf{q}(\partial^\alpha\! \psi_a)=-\mathbf{q}(\partial^\alpha\! \psi^*_a)=1$, where $\mathbf{q}\in\Z$ is the electric charge number.

Following~\cite{boas2000gauge} we define the {\it algebra of symbolic fields} denoted by $\Fa$. It is a~free unital graded-commutative ${}^*$-algebra over $\C$ generated by the elements of the set $\mathcal{G}$. The adjoint is defined uniquely by the following conditions: the anti-linearity and the identity $(B_1 B_2)^*=B_2^* B_1^*$ which holds for all \mbox{$B_1,B_2\in\Fa$}. The graded commutativity means that for any $B_1,B_2\in\Fa$ which are monomials in the generators it holds
\begin{equation}
 B_1 B_2 = (-1)^{\mathbf{f}(B_1)\mathbf{f}(B_2)} B_2 B_1. 
\end{equation}
The definition of $\mathbf{f}(B)$ is extended to all monomials $B\in\Fa$ by $\mathbf{f}(cB)=\mathbf{f}(B)$ for $c\in\C$ and the additivity $\mathbf{f}(B_1B_2)=\mathbf{f}(B_1)+\mathbf{f}(B_2)$ for any monomials $B_1,B_2\in\Fa$. The same holds for the other quantum numbers. We say that $B\in\Fa$ is a~homogeneous polynomial if it is a~linear combination of monomials with the same quantum numbers. The set of homogeneous polynomials is denoted by $\Fh$. The definitions of quantum numbers are naturally extended to $\Fh$.

By definition a monomial is an element of $\Fa$ which is a product of generators. The monomials are labeled by super-quadri-indices. The monomial $A^r\in\Fa$ labeled by the super-quadri-index $r$ is given by the formula
\begin{equation}\label{eq:wick_def}
 A^r := \prod_{i=1}^{\mathrm{p}} \prod_{\alpha\in\N_0^4}\,\,(\partial^\alpha\!A_{i})^{r(i,\alpha)}.
\end{equation} 
A {\it super-quadri-index} is by definition a~map
\begin{equation}
 r:~\mathcal{G} \ni (i,\alpha) \mapsto r(i,\alpha)\equiv r(\partial^\alpha\!A_i) \in \N_0
\end{equation}
supported on a~finite subset of $\mathcal{G}$. In addition, if $A_i$ is a fermionic field, then we demand that $r(i,\alpha)\in\{0,1\}$ for all $\alpha\in\N_0^4$. We say that the super-quadri-index $r$ involves only the field $A_i$ (involves only the massless fields) if $r(i',\alpha)=0$ for $i'\neq i$ (for all $i'$ such that $A_{i'}$ is a~massive field). We write $r\geq s$ iff $r(i,\alpha)\geq s(i,\alpha)$ for all $(i,\alpha)\in\mathcal{G}$. If $r\geq s$, then the super-quadri-index $u=r-s$ is given by $u(i,\alpha)=r(i,\alpha)-s(i,\alpha)$ for all $(i,\alpha)\in\mathcal{G}$. Moreover, we set
\begin{equation}
 |r| := \sum_{i=1}^{\mathrm{p}} \sum_{\alpha\in\N_0^4}r(i,\alpha),
 ~~~~~
 r!:=\prod_{i=1}^{\mathrm{p}} \prod_{\alpha\in\N_0^4}r(i,\alpha).
\end{equation}
Note that if fermionic fields are present, then the order of factors in the product \eqref{eq:wick_def} matters.  To remove the ambiguity in the definition of $A^r$, we fix some linear ordering in the set $\mathcal{G}\simeq\{1,\ldots,\mathrm{p}\}\times \N_0^4$ (its precise form is irrelevant for our purposes) and always multiply the generators in this order. All monomials in $\Fa$ are proportional to $A^r$ for some super-quadri-index $r$. The set 
\begin{equation}
 \{A^r\,:\,r ~\textrm{ is a~super-quadri-index}\}
\end{equation}
is a~linear basis of the algebra $\Fa$. For example, in the case of QED the electric current $j^\mu=\overline{\psi}\gamma^\mu\psi = \psi^*_a \gamma_{ab}^0 \gamma_{bc}^\mu \psi_c\in\Fa$ (we use the Einstein summation convention) is a~combination of the monomials $\psi^*_a \psi_c$.

The derivative of $B\in\Fa$ with respect to the generator $\partial^\alpha\!A_i$ is defined as a~graded derivation 
\begin{equation}
 \frac{\partial}{\partial(\partial^\alpha\!A_i)} BC = \frac{\partial B}{\partial(\partial^\alpha\!A_i)} \, C + 
 (-1)^{\mathbf{f}(B)\mathbf{f}(A_i)} \, B\, \frac{\partial C}{\partial(\partial^\alpha\!A_i)}
 ~~~~\forall_{B,C\in\Fh}
\end{equation}
such that
\begin{equation}
 \frac{\partial(\partial^{\alpha'}\!\!A_{i'})}{\partial(\partial^\alpha\!A_i)} = \delta_{ii'} \delta_{\alpha\alpha'}.
\end{equation}
Let $s$ be a~super-quadri-index. We define the linear map $\Fa\ni B\mapsto B^{(s)}\in\Fa$ given by
\begin{equation}
 B^{(s)}:= \prod_{i=1}^{\mathrm{p}} \prod_{\alpha\in\N_0^4}\,\,\left(\frac{\partial}{\partial(\partial^\alpha\!A_{i})}\right)^{s(i,\alpha)} B.
\end{equation}
For example, if $B=A^r$, then $B^{(s)}$ is proportional to $A^{r-s}$ for $r\geq s$ and vanishes otherwise. The polynomial $C\in\Fa$ is a~{\it sub-polynomial} of the polynomial $B\in\Fa$ iff $C$ equals $B^{(s)}$ up to a~multiplicative constant for some super-quadri-index $s$. All sub-polynomials of a~monomial $A^r$ are of the form $A^s$, where a super-quadri-index $s$ is such that $r\geq s$. 

The representation of the $SL(2,\C)$ group, which is the covering group of the Lorentz group, acts on the algebra of symbolic fields $\Fa$ in a standard way, and the action is called $\rho$.

\subsection{Wick polynomials}\label{sec:Wick}

Let $\cD$ be a~linear space over $\C$ equipped with a~non-degenerate sesquilinear inner product $(\cdot|\cdot)$ and let $L(\cD)$ be the space of linear maps $\cD\to\cD$. By an {\it operator-valued Schwartz distribution} on $\cD$ we mean a~map $T:\cS(\R^N)\!\to\! L(\cD)$ such that for any $\Psi,\Psi'\in\cD$
\begin{equation}
 \cS(\R^N)\ni g \mapsto (\Psi|\int \rd^N\!x\, g(x)\, T(x)\Psi') \in\C
\end{equation}
is a~Schwartz distribution. The space of operator-valued Schwartz distributions on $\cD$ is denoted by $\cS'(\R^N,L(\cD))$. If $T(x)\in\cS'(\R^{N},L(\cD))$ and $T'(x')\in\cS'(\R^{N'},L(\cD))$, then $T(x)T'(x')\in\cS'(\R^{N+N'},\hspace{-0.1mm}L(\cD))$, by the nuclear theorem.

To every symbolic field $B\in\Fa$ we associate the {\it Wick polynomial} $\normord{B(x)}$ which is an operator-valued Schwartz distribution~\cite{wightman1965fields,streater2000pct} on a~suitable domain $\cD_0$ in the Fock Hilbert space to be specified below. For example, ${A_i(x)}$ denotes one of the basic free fields and ${(\partial^\alpha\! A_i)(x)}\,=\partial^\alpha\!{A_i(x)}$ -- its derivative. All Wick monomials at point $x\in\R^4$ are up to a~multiplicative constant of the form 
\begin{equation}\label{eq:wick_monomial}
 \normord{A^{r}(x)}~ =~ \normord{ \prod_{i=1}^{\mathrm{p}} \prod_{\alpha\in\N_0^4} \,(\partial^\alpha\! A_{i}(x))^{r(i,\alpha)} }
\end{equation} 
for some super-quadri-index $r$. Observe that $\normord{B(x)}\,=0$ does not imply $B=0$. For example, we have ${(\square \varphi)(x)}=\square \varphi(x)=0$ if $\varphi$ is free field fulfilling the wave equation, whereas the symbol $\square \varphi\in\Fa$ is by definition a~nonzero generator.

The definition of all basic free fields ${A_i(x)}$ and the Fock spaces on which they act is standard~\cite{scharf2016gauge,weinberg1995quantum,derezinski2014quantum}. In the case of the vector field we use the Gupta--Bleuler approach~\cite{gupta1950theory,bleuler1950neue} which is nicely summarized in the monographs~\cite{scharf2016gauge,strocchi2013introduction}. Let us only mention that in this approach the vector field is defined on the Krein space with two inner products. The field is hermitian only with respect to the covariant inner product which is not positive definite.

The Hilbert space $\mathcal{H}$ on which the model is defined is the tensor product of the Fock spaces on which the fields ${A_1(x)},\ldots,{A_\mathrm{p}(x)}$ corresponding to basic generators act. The vacuum state in the Fock space $\mathcal{H}$ is denoted by $\Omega$. We introduce the following dense subspace in $\mathcal{H}$
\begin{multline}\label{eq:dom}
 \cD_0 :=\Span_\C\bigg\{ \int\rd^4 x_1\ldots\rd^4 x_n\,f(x_1,\ldots,x_n)\,A_{i_1}(x_1)\ldots A_{i_n}(x_n)\Omega ~:
 \\
 ~n\in\N_0, ~i_1,\ldots,i_n\in\{1,\ldots,\mathrm{p}\},~f\in\cS(\R^{4n}) \bigg\}.
\end{multline}
From now on, all the operators we consider are elements of $L(\cD_0)$. As shown in~\cite{wightman1965fields} Wick polynomials are are well defined as operator-valued Schwartz distributions on~$\cD_0$. The fact that $\cD_0$ is embedded in the Hilbert space $\mathcal{H}$ plays no role. The space $\cD_0$ is equipped with a~Poincar{\'e} covariant and non-degenerate inner product, which is denoted by $(\cdot|\cdot)$. The product  is positive definite unless the vector fields are present in the model. The hermitian conjugation and the notion of unitarity are defined with respect to this product. The unitary representation of the inhomogeneous $SL(2,\C)$ group,  which is the covering group of the Poincar{\'e} group, denoted by $U(a,\Lambda)$, where $a\in\R^4$ and $\Lambda$ is a~Lorentz transformation, is defined on $\cD_0$ in the standard way. In the case of pure translations we write $U(a)\equiv U(a,\id)$. 

\subsection{\texorpdfstring{$\Fa$}{F}-products}\label{sec:F_prod}

In this section we introduce the notion of the $\Fa$-product. The time-ordered products and all other products which are used in the EG approach are examples of the $\Fa$-products. We say that a multi-linear map
\begin{equation}
 F:~\Fa^n \ni (B_1,\ldots,B_n) \mapsto \,F(B_1(x_1),\ldots ,B_n(x_n)) \,\in \cS'(\R^{4n},L(\cD_0))
\end{equation}
is an $\Fa$-product iff the following conditions hold for any $B_1,\ldots,B_n\in\Fh$:
\begin{enumerate}[label=\bf{A.\arabic*},leftmargin=*]
\item\label{axiom1} Translational covariance: 
\begin{equation}
 U(a)F(B_1(x_1),\ldots,B_n(x_n))U(a)^{-1} 
 \\
 = F(B_1(x_1+a),\ldots,B_n(x_n+a)).
\end{equation}
\item\label{axiom2} If $\mathbf{f}(B_1)+\ldots+\mathbf{f}(B_n)$ is odd, then
\begin{equation}
 (\Omega|F(B_1(x_1),\ldots,B_n(x_n)) \Omega) = 0.
\end{equation}
\item\label{axiom3} Wick expansion: The product $F(B_1(x_1),\ldots,B_n(x_n))$ is uniquely determined by the VEVs of the product $F$ of the sub-polynomials of $B_1,\ldots,B_n$:
\begin{multline}\label{eq:T_expansion}
 F(B_1(x_1),\ldots,B_n(x_n)) 
 =
 \sum_{s_1,\ldots,s_n} (-1)^{\mathbf{f}(s_1,\ldots,s_n)}\, 
 \\
 ~(\Omega|F(B_1^{(s_1)}(x_1),\ldots,B_n^{(s_n)}(x_n))\Omega)
 ~\frac{\normord{A^{s_1}(x_1)\ldots A^{s_n}(x_n)}}{s_1!\ldots s_n!}.
\end{multline}
\end{enumerate}

The above properties of $\Fa$-products are modeled on the ordinary product of normally ordered operators
\begin{equation}
 \Fa^n \ni (B_1,\ldots,B_n) \mapsto \,\normord{B_1(x_1)}\ldots \normord{B_n(x_n)} \,\in \cS'(\R^{4n},L(\cD_0)),
\end{equation}
in which case the properties \ref{axiom1} and \ref{axiom2} are trivially satisfied and the property \ref{axiom3} is the usual Wick expansion:
\begin{multline}
 \normord{B_1(x_1)}\ldots \normord{B_n(x_n)} 
 \,=
 \sum_{s_1,\ldots,s_n} (-1)^{\mathbf{f}(s_1,\ldots,s_n)}\, \\
 ~(\Omega|\normord{B_1^{(s_1)}(x_1)}\ldots\normord{B_n^{(s_n)}(x_n)}\Omega)
 ~\frac{\normord{A^{s_1}(x_1)\ldots A^{s_n}(x_n)}}{s_1!\ldots s_n!}.
\end{multline}
The factor $(-1)^{\mathbf{f}(s_1,\ldots,s_n)}$ in \eqref{eq:T_expansion} may be read off from the above equation. In particular, if $B_1,\ldots,B_n$ have even fermion number, then ${(-1)^{\mathbf{f}(s_1,\ldots,s_n)}=1}$. The RHS of Equation \eqref{eq:T_expansion} is a~well-defined operator-valued Schwartz distribution as a~result of the following theorem.
\begin{thm}\label{thm:eg0}\emph{\cite{epstein1973role}} 
If $t\in\cS'(\R^{4n})$ is translationally invariant, then 
\begin{equation}
 t(x_1,\ldots,x_n) \normord{B_1(x_1)\ldots B_n(x_n)}~\in \cS'(\R^{4n},L(\cD_0)).
\end{equation}
for any $B_1,\ldots,B_n\in\Fa$.
\end{thm}
\begin{rem}\label{rem:transl}
A distribution $t\in\cS'(\R^{4(n+1)})$ is translationally invariant iff
\begin{equation}
 t(x_1,\ldots,x_{n+1}) = t(x_1+a,\ldots,x_{n+1}+a)
\end{equation}
for all $a\in\R^4$. For any translationally invariant $t\in\cS'(\R^{4(n+1)})$, $n\in\N_0$, we can define the associated distribution $\underline{t}\in\cS'(\R^{4n})$ by 
\begin{equation}\label{eq:dist_trans_inv}
 \underline{t}(x_1,\ldots,x_n)\equiv t(x_1,\ldots,x_n,0):=\int \rd^4 y \, t(x_1+y,\ldots,x_n+y,y)\, h(y),
\end{equation}
where $h\in\cS(\R^4)$, $\int\rd^4 y \, h(y)=1$. The RHS of the above equation is independent of $h$. The Fourier transforms of $t$ and $\underline{t}$ are related by
\begin{equation}\label{eq:dist_trans_inv_F}
 \F{t}(q_1,\ldots,q_{n+1}) = (2\pi)^4 \delta(q_1+\ldots+q_{n+1})\,\F{\underline{t}}(q_1,\ldots,q_n).
\end{equation}
\end{rem}

Let us define the product and the graded commutator of two $\Fa$-products:
\begin{equation}
\begin{aligned}\label{eq:F_Fprime}
 F:~\Fa^n \ni (B_1,\ldots,B_n) &\mapsto \,F(B_1(x_1),\ldots ,B_n(x_n)) \,\in \cS'(\R^{4n},L(\cD_0)),
  \\
 F':~\Fa^{n'} \ni (B'_1,\ldots,B'_{n'}) &\mapsto \,F(B'_1(x'_1),\ldots ,B'_{n'}(x'_{n'})) \,\in \cS'(\R^{4n'},L(\cD_0)).
\end{aligned}
\end{equation}
The product of $F$ and $F'$ is by definition the following multi-linear map
\begin{multline}\label{eq:F_product}
 \Fa^{n+n'}\ni (B_1,\ldots,B_n;B'_1,\ldots,B'_{n'})\mapsto
 \\
 \hspace{-2mm} F(B_1(x_1),\ldots,B_n(x_n)) F'(B'_1(x'_1),\ldots,B'_{n'}(x'_{n'}))\!\in\! \cS'(\R^{4(n+n')},L(\cD_0)),\hspace{-1mm}
\end{multline}
which, as it turns out, is again an $\Fa$ product. The only non-trivial part in the proof of this fact is the verification that the product of $F$ and $F'$ fulfills the condition \eqref{eq:T_expansion}, which follows from the properties of the Wick products (cf. Section 4 in \cite{epstein1973role}). The graded commutator of $F$ and $F'$ is defined by
\begin{multline}\label{eq:com}
 \left[ F(B_1(x_1),\ldots,B_n(x_n)), F'(B'_1(x'_1),\ldots,B'_n(x'_{n'})) \right] :=
 \\
 F(B_1(x_1),\ldots,B_n(x_n))\,F'(B'_1(x'_1),\ldots,B'_n(x'_{n'}))-
 \\
 (-1)^{\mathbf{f}(B_1\ldots B_n)\mathbf{f}(B'_1\ldots B'_n)}\,
 F'(B'_1(x'_1),\ldots,B'_n(x'_{n'}))\, F(B_1(x_1),\ldots,B_n(x_n)),
\end{multline}
where $B_1\ldots,B_n,B'_1,\ldots,B'_{n'}\in\Fh$ and $\mathbf{f}(B_1\ldots B_n)$ is the fermion number of $B_1\ldots B_n\in\Fa$. 

The VEV of the product of $F$ and $F'$ may be expressed as follows 
\begin{multline}\label{eq:vev_product_representation}
(\Omega|F(B_1(x_1),\ldots,B_n(x_n)) F'(B_n^{\prime}(x'_1),\ldots,B^{\prime}_{n'}(x'_{n'}))\Omega)
\\
 =\sum_{\bar u,\bar u',\bar i,\bar i',\bar \alpha,\bar \alpha'} ~\const~~
 \prod_{j=1}^l  (\Omega|\normord{\partial^{\bar \alpha(j)}\!A_{\bar i(j)}(x_{\bar u(j)})}\,
 \normord{\partial^{\bar \alpha'(j)}\!A_{\bar i'(j)}(x'_{\bar u'(j)})}\Omega)\times
 \\
 (\Omega|F(B^{(s_1)}_1(x_1),\ldots,B^{(s_n)}_n(x_n))\Omega) \, (\Omega| F'(B^{\prime(s'_1)}_1(x'_1),\ldots,B^{\prime(s'_{n'})}_{n'}(x'_{n'}))\Omega),
\end{multline}
where the sum is carried out over $l\in\N_0$ and functions $\bar u,\bar u',\bar i,\bar i',\bar \alpha,\bar \alpha'$ such that:
\begin{enumerate}[leftmargin=*,label={(\arabic*)}]
 \item $\bar u,\bar u'$ are functions from $\{1,\ldots,l\}$ to $\{1,\ldots,n\}$ or $\{1,\ldots,n'\}$,
 \item $\bar i,\bar i'$ are functions from $\{1,\ldots,l\}$ to the set $\{1,\ldots,\mathrm{p}\}$, which is identified with the set $\mathcal{G}_0$ of the basic generators \eqref{eq:basic_gen},
 \item $\bar \alpha,\bar \alpha'$ are functions from $\{1,\ldots,l\}$ to $\N_0^4$ (the set of multi-indices).
\end{enumerate}
The super-quadri-indices $s_1,\ldots,s_n$, $s'_1,\ldots,s'_n$ are determined uniquely by
\begin{equation}\label{eq:s_sPrime}
 A^{s_m}\propto \prod_{\substack{j\in \{1,\ldots,l\} \\ \bar u(j)=m}} \partial^{\bar \alpha(j)} A_{\bar i(j)},~~~~~
 A^{s'_m}\propto \prod_{\substack{j\in \{1,\ldots,l\} \\ \bar u'(j)=m}} \partial^{\bar \alpha'(j)} A_{\bar i'(j)},
\end{equation}
where $\propto$ indicates the equality up to a~possible factor $(-1)$.

Given a~list of functions $\bar u$, $\bar u'$, $\bar i$, $\bar i'$, $\bar \alpha$, $\bar \alpha'$ of the above type we define the following maps characterizing this list
\begin{equation}
 \bar\ext,\bar\der:\, \{1,\ldots,\mathrm{p}\}\simeq\mathcal{G}_0\to\N_0,
\end{equation}
where $\mathcal{G}_0$ is the set of the basic generators \eqref{eq:basic_gen}. We set
\begin{align}\label{eq:ext_s}
 \bar\ext(i)\equiv \bar\ext(A_i):= |\{j\,: \bar i(j)=i\}|+ |\{j\,: \bar i'(j)=i\}|,
 \\[2pt]
 \label{eq:der_s}
 \bar\der(i)\equiv \bar\der(A_i):= \sum_{\substack{j\in\{1,\ldots,l\}\\\bar i(j)=i}}|\bar \alpha(j)| + \sum_{\substack{j\in\{1,\ldots,l\}\\\bar i'(j)=i}}|\bar \alpha'(j)|.
\end{align} 
Given a~list of super-quadri-indices $\mathbf{s}=(s_1,\ldots,s_n)$ we define functions 
\begin{equation}
 \ext_{\mathbf{s}}, \der_{\mathbf{s}}:\, \{1,\ldots,\mathrm{p}\}\simeq\mathcal{G}_0\to\N_0
\end{equation}
characterizing this list by setting
\begin{equation}\label{eq:ext}
 \ext_{\mathbf{s}}(i)\equiv \ext_{\mathbf{s}}(A_i):= \sum_{\alpha\in\N_0^4}  s(i,\alpha),
 ~~~~~
 \der_{\mathbf{s}}(i)\equiv \der_{\mathbf{s}}(A_i):= \sum_{\alpha\in\N_0^4}  s(i,\alpha)~|\alpha|,
\end{equation}
where 
\begin{equation}\label{eq:total_index}
 s(i,\alpha):=\sum_{j=1}^n s_j(i,\alpha)
\end{equation}
is the total super-quadri-index of the list $\mathbf{s}$. 

It follows from \eqref{eq:s_sPrime} that the lists $\mathbf{s}=(s_1,\ldots,s_n)$, $\mathbf{s}'=(s'_1,\ldots,s'_n)$ and functions $\bar u,\bar u',\bar i,\bar i',\bar \alpha,\bar \alpha'$ which appear in the representation \eqref{eq:vev_product_representation} satisfy the following constraints
\begin{equation}\label{eq:constraints}
 \ext_{\mathbf{s}}(A_i) + \ext_{\mathbf{s}'}(A_i) = \bar\ext(A_i),~~~~~\der_{\mathbf{s}}(A_i) + \der_{\mathbf{s}'}(A_i) = \bar\der(A_i).
\end{equation}
The term in the sum on the RHS of \eqref{eq:vev_product_representation} for which $l=0$ (i.e. $\mathbf{s}=\mathbf{s}'=0$ and $\bar\ext(A_i)=\ext_{\mathbf{s}}(A_i)=\ext_{\mathbf{s}'}(A_i)=0$ for $A_i\in\mathcal{G}_0$) is called the \emph{vacuum contribution}. The terms in the sum on the RHS of \eqref{eq:vev_product_representation} corresponding to functions $\bar u$, $\bar u'$, $\bar i$, $\bar i'$, $\bar \alpha$, $\bar \alpha'$, which satisfy the condition $\bar\ext(A_i)=\ext_{\mathbf{s}}(A_i)=\ext_{\mathbf{s}'}(A_i)=0$ for all massive fields $A_i\in\mathcal{G}_0$ (i.e. the super-quadri-indices from the lists $\mathbf{s}$, $\mathbf{s}'$ involve only massless fields) are called the \emph{massless contributions}. 

We close this section with a~lemma which further restricts the form of terms which may appear in the sum on the RHS of \eqref{eq:vev_product_representation}.

\begin{lem}\label{lem:aux_lemma}
(A) There exist neighborhoods $\mathcal{O}\subset\R^{4n}$ and $\mathcal{O}'\subset\R^{4n'}$ of the origins in $\R^{4n}$ and $\R^{4n'}$, respectively, such that the following holds:
Let $g\in\cS(\R^{4k+4n'})$ be any test function with 
\begin{equation}
 \supp\,\F{g} \subset \R^{4n} \times \mathcal{O}' \cup \mathcal{O} \times \R^{4n'}.
\end{equation}
Then all contributions to the expansion \eqref{eq:vev_product_representation} vanish on $g$ with the exception of the vacuum and the massless contributions.

\noindent (B) Representation \eqref{eq:vev_product_representation} of the VEV of the graded commutator $[F,F']$ does not contain the vacuum contribution.
\end{lem}

\subsection{Time-ordered product}\label{sec:axioms}

The scattering matrix and the interacting fields in perturbative QFT are given in terms of the time-ordered products. By definition the time-ordered products form a~family of $\Fa$-products
\begin{equation}\label{eq:time_ordered_product}
 \Fa^n \ni (B_1,\ldots,B_n) \mapsto \,\T(B_1(x_1),\ldots,B_n(x_n)) \,\in \cS'(\R^{4n},L(\cD_0))
\end{equation}
indexed by $n\in\N_0$ which satisfies the following axioms (besides the conditions \ref{axiom1}-\ref{axiom3} stated in the previous section, which we also call the axioms):
\begin{enumerate}[label=\bf{A.\arabic*},leftmargin=*]
\setcounter{enumi}{3}
\item\label{axiom4} $\T(\emptyset)=\id$, $\T(B(x))=\,\,\normord{B(x)}$,
\begin{equation}
\T(B_1(x_1),\ldots,B_n(x_n),1(x_{n+1})) = \T(B_1(x_1),\ldots,B_n(x_n)),
\end{equation}
where $1$ on the LHS of the above equality is the unity in $\Fa$.
\item\label{axiom5} Graded symmetry: For any $B_1,\ldots,B_n\in\Fh$ it holds
\begin{equation}\label{eq:T_graded}
 \T(B_1(x_1),\ldots,B_n(x_n)) = (-1)^{\mathbf{f}(\pi)}\T(B_{\pi(1)}(x_{\pi(1)}),\ldots,B_{\pi(n)}(x_{\pi(n)})),
\end{equation}
where $\pi$ is any permutation of the set $\{1,\ldots,n\}$ and $\mathbf{f}(\pi)\in\Z/2\Z$ is the number of transpositions in $\pi$ that involve a~pair of fields with odd fermion number. In particular, the time-ordered product of polynomials $B_1,\ldots,B_n$ which have even fermion number is invariant under permutations of its arguments.
\item\label{axiom6} Causality: If none of the points $x_1,\ldots,x_m$ is in the causal past of any of the points $x_{m+1},\ldots,x_n$, then
\begin{multline}\label{eq:T_causality}
 \T(B_1(x_1),\ldots,B_n(x_n)) 
 \\
 = \T(B_1(x_1),\ldots,B_m(x_m))\T(B_{m+1}(x_{m+1}),\ldots,B_n(x_n)).
\end{multline}
\item\label{axiom7}
Bound on the Steinmann scaling degree: For all $B_1,\ldots,B_n\in\Fh$ it holds
\begin{equation}\label{eq:sd_bound}
 \sd(\,(\Omega|\T(B_1(x_1),\ldots,B_{n-1}(x_{n-1}),B_n(0)) \Omega)\,) \leq \sum_{j=1}^n (\dim(B_j)+\CC),
\end{equation}
where $\sd(\cdot)$ is the Steinmann scaling degree of a~distribution and the constant \mbox{$\CC\in \{0,1\}$} is fixed.
\end{enumerate}

\begin{dfn}\label{def:sd}
\emph{(Steinmann scaling degree~\cite{steinmann1971perturbative})}
Let $t\in\mathcal{S}'(\R^N)$. The scaling degree of $t$ with respect to the origin is given by
\begin{align}
 \label{eq:scaling_set}
 &\sd(t)
 :=
 \inf\left\{s\in \R~:~\forall_{g\in\mathcal{S}(\R^N)}
 \lim_{\lambda\searrow 0}\int\rd^N\!x\,\lambda^s t(\lambda x) g(x) =0 \right\} ,
 \\
 &\sd(t)\in \{-\infty\}\cup\R\cup\{+\infty\},
\end{align}
where by definition $\inf\emptyset=+\infty$. 
\end{dfn} 

For the proof of the existence of the time-ordered products satisfying the above axioms we refer the reader to~\cite{epstein1973role,brunetti2000microlocal}. We stress that the time-ordered products, like any other $\Fa$-products, are operator-valued Schwartz distributions on $\cD_0$ indexed by the symbolic fields $B\in\Fa$, and not the Wick polynomials $\normord{B(x)}\in\cS'(\R^{4},L(\cD_0))$. This is the feature of the so-called off-shell formalism~\cite{stora2002pedagogical,dutsch2004causal,brouder2008relating}. Observe that, for example, the time-ordered product $\T(\square\varphi(x_1),\varphi(x_2))$ does not have to vanish even if $\varphi$ is a~free field satisfying the wave equation.

Setting $\CC=0$ in the condition \eqref{eq:sd_bound} is the standard choice~\cite{epstein1973role,brunetti2000microlocal,fredenhagen2015perturbative}. The case $\CC>0$ corresponds in the BPHZ approach to the special form of the so-called over-subtractions~\cite{zimmermann1969convergence}. Its compatibility with the rest of axioms may be proved using Lemma 6.6 of~\cite{brunetti2000microlocal} as in the standard case. We always assume that $\CC=0$ in the case of models with at least one interaction vertex of dimension $4$. However, in the case of models with interaction vertices of dimension 3 it is also possible to set $\CC=1$. The choice $\CC=1$ makes the UV behavior of the theory worse, but at the same time allows more freedom in the construction of time-ordered products. This freedom may be used to improve its IR properties.

\subsection{Normalization freedom}\label{sec:freedom}

Assume that all the time-ordered products with at most $n$ arguments are given. In this section we will characterize the ambiguity in defining the time-ordered products with $n+1$ arguments. To this end, by the axiom~\ref{axiom3} it is enough to consider the family of numerical distributions
\begin{equation}\label{eq:T_freedom}
 \Fa^{n+1}\ni(B_1,\ldots,B_{n+1})\mapsto(\Omega|\T(B_1(x_1),\ldots,B_{n+1}(x_{n+1}))\Omega)\in\cS'(\R^{4(n+1)}).
\end{equation}
We claim that any two possible definitions of the above family differ by a~family of distributions
\begin{equation}\label{eq:v_freedom}
  \Fa^{n+1}\ni(B_1,\ldots,B_{n+1})\mapsto v(B_1(x_1),\ldots,B_{n+1}(x_{n+1}))\in\cS'(\R^{4(n+1)})
\end{equation}
such that: 
\begin{enumerate}[leftmargin=*,label={(\arabic*)}]
\item $v$ is graded-symmetric, i.e. 
\begin{equation}\label{eq:v_graded}
 v(B_1(x_1),\ldots,B_{n+1}(x_{n+1})) = (-1)^{\mathbf{f}(\pi)}v(B_{\pi(1)}(x_{\pi(1)}),\ldots,B_{\pi(n+1)}(x_{\pi(n+1)}))
\end{equation}
for all $B_1,\ldots,B_{n+1}\in\Fh$ and all permutations $\pi\in\mathcal{P}_{n+1}$, where $\mathbf{f}(\pi)\in\Z/2\Z$ is the number of transpositions in $\pi$ that involve a~pair of fields with odd fermion number,
\item for all $B_1,\ldots,B_{n+1}\in\Fh$ the distribution $v(B_1(x_1),\ldots,B_{n+1}(x_{n+1}))$ vanishes if $ \sum_{j=1}^{n+1} \mathbf{f}(B_j)$ is odd and otherwise it is of the form 
\begin{equation}\label{eq:freedom_form}
 \sum_{\substack{\gamma\\|\gamma|\leq\omega}} c_\gamma \partial^\gamma \delta(x_1-x_{n+1})\ldots \delta(x_{n}-x_{n+1})
\end{equation}
for some constants $c_\gamma\in\C$ indexed by multi-indices $\gamma$, $|\gamma|\leq\omega$, where
\begin{equation}\label{eq:omega}
 \omega := 4 - \sum_{j=1}^{n+1} (4 - \CC -\dim(B_j)).
\end{equation}
\end{enumerate}

The bound $|\gamma|\leq\omega$ in the sum \eqref{eq:freedom_form} is a consequence of the axiom \ref{axiom7}. It follows that $v(B_1(x_1),\ldots,B_{n+1}(x_{n+1}))=0$ if $\omega<0$ and the VEV of the corresponding time-ordered product is determined uniquely by the time-ordered products with at most $n$ arguments.

\subsection{Generating functional and \texorpdfstring{$\aT$, $\Adv$, $\Ret$, $\Dif$}{aT, A, R, D} products}\label{sec:generating}

Let $B_1,\ldots,B_m\in\Fh$, $g_1,\ldots,g_m\in\mathcal{S}(\R^4)$ (if $B_k$ has odd fermion parity, i.e. $\mathbf{f}(B_k)$ is odd, then $g_k$ is a~Schwartz function valued in the Grassmann algebra with odd Grassmann parity) and $\mathbf{g}=\sum_{j=1}^m g_j \otimes B_j\in \mathcal{S}(\R^4)\otimes \Fa$. The generating functional of the time-ordered products is given by
\begin{multline}\label{eq:def_Texp}
 S(\mathbf{g})\equiv\Texp\bigg( \ri \int \rd^4 x \, \sum_{j=1}^m g_j(x) B_j(x)\bigg) := 
 \\
 \sum_{n=0}^\infty \frac{\ri^n}{n!} \sum_{j_1,\ldots,j_n}
 \int\rd^4 x_1\ldots\rd^4 x_n
 \,g_{j_1}(x_1)\ldots g_{j_n}(x_n)
 \T(B_{j_1}(x_1),\ldots,B_{j_n}(x_n)).
\end{multline}
It is a~map that sends elements of $\mathcal{S}(\R^4)\otimes \Fa$ into formal power series in $g_1,\ldots,g_m$ with coefficients in $L(\cD_0)$. The argument of $\Texp$ should be treated symbolically. Since $S(0)=\id$ the formal power series $S(\mathbf{g})$ is invertible. A non-trivial but very important consequence of the axiom of causality \ref{axiom6} is the so-called {\it causal factorization property} \cite{epstein1973role,bogoliubov1959introduction}
\begin{equation}\label{eq:causal_fact}
 S(\mathbf{g}'+\mathbf{g}+\mathbf{g}'') = S(\mathbf{g}'+\mathbf{g})S(\mathbf{g})^{-1}S(\mathbf{g}+\mathbf{g}'')
\end{equation}
which holds whenever $\supp\,\mathbf{g}' + \overline{V}^+  \cap \supp\,\mathbf{g}''=\emptyset$, i.e. the support of $\mathbf{g}''$ does not intersect the causal future of the support of $\mathbf{g}'$. The support of $\mathbf{g}$ may be arbitrary.

Now we will introduce the abbreviated notation which will be used throughout the paper. Let $\mathbf{g}=\sum_{j=1}^m g_j \otimes B_j$, $I\equiv \left( B_1(x_1),\ldots,B_m(x_m)\right)$. We set
\begin{equation}\label{eq:fun_der}
 \frac{\delta}{\delta \mathbf{g}(I)} \equiv  \frac{\delta}{\delta g_m(x_m)}\ldots\frac{\delta}{\delta g_1(x_1)}
~~~~~\textrm{and}~~~~~
F(I)\equiv F(B_1(x_1),\ldots,B_m(x_m)),
\end{equation}
where $B_1,\ldots,B_m\in\Fh$, $g_1,\ldots,g_m\in\mathcal{S}(\R^4)$ and $F$ is an $\Fa$-product. The functional $S(\mathbf{g})=S(\sum_{j=1}^m g_j \otimes B_j)$ generates the time-ordered products in the following sense
\begin{equation}\label{eq:time_ordered_derivative}
 \T(B_1(x_1),\ldots,B_m(x_m)) 
 = (-\ri)^m \frac{\delta}{\delta g_m(x_m)}\ldots\frac{\delta}{\delta g_1(x_1)} S(\mathbf{g})\bigg|_{\mathbf{g}=0}, 
\end{equation}
where we used the identities
\begin{equation}
 \frac{\delta g_{j'}(y)}{\delta g_j(x)} = \delta_{jj'} \delta(x-y)~~~\textrm{and}~~~
 \frac{\delta}{\delta g_j(x)}\frac{\delta}{\delta g_{j'}(y)}=(-1)^{\mathbf{f}(B_j)\mathbf{f}(B_{j'})}\frac{\delta}{\delta g_{j'}(y)}\frac{\delta}{\delta g_{j}(x)}.
\end{equation}

The anti-time-ordered products are defined by
\begin{equation}
 \aT(I) 
 := (-\ri)^n \frac{\delta}{\delta \mathbf{g}(I)} S(-\mathbf{g})^{-1}\bigg|_{\mathbf{g}=0}.
\end{equation}
The anti-time-ordered product $\aT(I)$ can be expressed as a~combination of products of the time-ordered products $T(I_1)\ldots\T(I_n)$, where the concatenation of the sequences $I_1,\ldots,I_n$ is some permutation of the sequence $I$. 

Consider $\mathbf{g}=\sum_{j=1}^n g_j \otimes B_j$ and $\mathbf{h}=\sum_{j=1}^m h_j\otimes C_j$ and set 
\begin{equation}
 I=\left( B_1(y_1),\ldots,B_n(y_n)\right),~~~ J=\left( C_1(x_1),\ldots,C_m(x_m)\right).
\end{equation}
The advanced and retarded products are given by
\begin{align}\label{eq:def_adv}
 \Adv(I;J):=
 (-\ri)^{n+m} 
 \frac{\delta}{\delta \mathbf{h}(J)}
 \frac{\delta}{\delta \mathbf{g}(I)} 
 S(\mathbf{g}+\mathbf{h}) S(\mathbf{g})^{-1} \bigg|_{\substack{\mathbf{g}=0\\\mathbf{h}=0}},
\\
 \label{eq:def_ret}
 \Ret(I;J):=
 (-\ri)^{n+m} 
 \frac{\delta}{\delta \mathbf{h}(J)}
 \frac{\delta}{\delta \mathbf{g}(I)} 
 S(\mathbf{g})^{-1} S(\mathbf{g}+\mathbf{h}) \bigg|_{\substack{\mathbf{g}=0\\\mathbf{h}=0}}.
\end{align}
Moreover, we introduce the following products
\begin{align}\label{eq:def_dif}
 \Dif(I;J) &:=
 \Adv(I;J)
 -\Ret(I;J),
 \\
 \label{eq:def_adv_prime}
 \Adv'(I;J)
 &:=
 \Adv(I;J)
 -\T(I;J).
\end{align}

Let us list the properties of the products introduced above. It is evident that they are $\Fa$-products. Moreover, they are graded-symmetric separately in arguments collectively denoted by $I$ and $J$. A closer inspection reveals that $\Dif$ and $\Adv'$ products with $n$ arguments can be expressed by products of the time-ordered products with at most $n-1$ arguments. Using the causal factorization property \eqref{eq:causal_fact} one proves that the advanced, retarded and causal products have the following support properties:
\begin{equation}\label{eq:supp_adv_ret_dif}
\begin{aligned}
 &\supp \Adv(I;J)\subset \Gamma^+_{n,m},
 \\
 &\supp \Ret(I;J)\subset \Gamma^-_{n,m},
 \\
 &\supp \Dif(I;J)\subset \Gamma^+_{n,m}\cup \Gamma^-_{n,m},\hspace{-1mm}
\end{aligned} 
\end{equation}
where
\begin{equation}\label{eq:def_gen_cones}
 \Gamma^\pm_{n,m}:=\{(y_1,\ldots,y_n;x_1,\ldots,x_m) ~:~\exists_{u} \forall_j ~y_j \in x_{u(j)} \pm \overline{V}^+ \}
\end{equation}
and $u:\,\{1,\ldots,n\}\to\{1,\ldots,m\}$. The condition in the definition of the sets $\Gamma^\pm_{n,m}$ means that each of the points $y_1,\ldots,y_n$ is in the causal future/past of some of the points $x_1,\ldots,x_m$. In particular, in the case when $J$ contains only one element, using translational invariance of the VEVs of the $\Fa$-products we obtain:
\begin{equation}\label{eq:supp_VEVs}
\begin{aligned}
 &\supp\,(\Omega|\Adv(B_1(x_1),\ldots,B_n(x_n);B_{n+1}(0))\Omega) \subset \Gamma_{n}^+,
 \\
 &\supp\,(\Omega|\Ret(B_1(x_1),\ldots,B_n(x_n);B_{n+1}(0))\Omega) \subset \Gamma_{n}^-,
 \\
 &\supp\,(\Omega|\Dif(B_1(x_1),\ldots,B_n(x_n);B_{n+1}(0))\Omega) \subset \Gamma_{n}^+ \cup \Gamma_{n}^-,
\end{aligned} 
\end{equation}
where
\begin{equation}\label{eq:def_cones}
 \Gamma_{n}^+ = -\Gamma_{n}^- := \{ (x_1,\ldots, x_{n}) \,:\, \forall_j \,x_j \in \overline{V}^+ \}.
\end{equation}

\subsection{Definition of Wightman and Green functions}\label{sec:W_G_IR}

In order to define a model of the interacting QFT one has to specify its interaction vertices 
\begin{equation}
 \cL_1,\ldots,\cL_\mathrm{q}\in\Fh,
\end{equation}
which are distinguished homogeneous elements of the algebra of symbolic fields $\Fa$ such that for all $l\in\{1,\ldots,\mathrm{q}\}$:
\begin{enumerate}[leftmargin=*,label={(\arabic*)}]
 \item $\cL_l$ is a~Lorentz scalar,
 \item $\cL_l=\cL_l^*$,
 \item $\dim(\cL_l)\leq 4-\CC$,
 \item $\mathbf{f}(\cL_l)$ is even,
\end{enumerate}
where $\CC$ is the constant which appears in the axiom \ref{axiom7}. The first condition is crucial for the Lorentz covariance of the model. The second one is needed for quantum mechanical consistency of the model, e.g. for the unitarity of the scattering matrix. The third reflects the fact that we consider only renormalizable interactions. The last condition guarantees that the scattering matrix commutes with the fermion number operator. To prove the existence of the Wightman and Green function in the case of models with massless particles we impose some additional conditions on the interaction vertices which are formulated in Section~\ref{sec:PROOF} as Assumption~\ref{asm}. 

To each interaction vertex $\cL_l$ we associate a~parameter $e_l\in\R$ also called a~coupling constant. The scattering matrix and the interacting fields are defined as formal power series in the independent parameters $e_1,\ldots,e_\mathrm{q}$. In some models the physical coupling constants accompanying different interaction terms are related. In this case we assume that there exists a~set of independent parameters such that all  $e_1,\ldots,e_\mathrm{q}$ are polynomial functions of them (for example, it may hold $e_1=e$, $e_2=e^2$ for some $e\in\R$). In our proof of the existence of the weak adiabatic limit it is possible to assume that the parameters $e_1,\ldots,e_\mathrm{q}$ are independent. The relation between $e_1,\ldots,e_\mathrm{q}$ of the above-mentioned form may be always imposed in the final expression for the Wightman and Green functions by reorganizing the resulting formal power series. 

Set $\mathbf{g}=\sum_{l=1}^\mathrm{q} e_l g_l \otimes \cL_l$ and $\mathbf{h}=h\otimes C$, where $g_1,\ldots,g_\mathrm{q},h\in\cS(\R^4)$ and $C\in\Fh$. The advanced and retarded interacting fields with IR regularization are defined by the Bogoliubov formulas~\cite{epstein1973role,bogoliubov1959introduction}
\begin{align}\label{eq:bogoliubov_adv}
 C_\adv(\mathbf{g};x) 
 :=(-\ri) \frac{\delta}{\delta h(x)}
  S(\mathbf{g}+\mathbf{h})S(\mathbf{g})^{-1}\bigg|_{\mathbf{h}=0} 
 =
 \sum_{n=0}^\infty \frac{\ri^n}{n!} \sum_{l_1,\ldots,l_n}e_{l_1}\ldots e_{l_n} \\\times
 \int\rd^4 y_1\ldots\rd^4 y_n\,
 g_{l_1}(y_1)\ldots g_{l_n}(y_n)
 ~\Adv(\cL_{l_1}(y_1),\ldots,\cL_{l_n}(y_n);C(x)),
\end{align}
\begin{align}
\label{eq:bogoliubov_ret}
 C_\ret(\mathbf{g};x) 
 :=(-\ri) \frac{\delta}{\delta h(x)}
 S(\mathbf{g})^{-1} S(\mathbf{g}+\mathbf{h})\bigg|_{\mathbf{h}=0} 
 =
 \sum_{n=0}^\infty \frac{\ri^n}{n!} \sum_{l_1,\ldots,l_n}e_{l_1}\ldots e_{l_n}\\ \times
 \int\rd^4 y_1\ldots\rd^4 y_n\,
 g_{l_1}(y_1)\ldots g_{l_n}(y_n)
 ~\Ret(\mathcal{L}_{l_1}(y_1),\ldots,\mathcal{L}_{l_n}(y_n);C(x)). 
\end{align}
The functions $g_1,\ldots,g_\mathrm{q}\in\cS(\R^{4n})$ are called the switching functions. They switch off the interaction as $|x|\to\infty$ and play the role of the IR regularization. Now, let $\mathbf{h}=\sum_{j=1}^m h_j\otimes C_j$, where $C_1,\ldots,C_m\in\Fh$, $h_1,\ldots,h_m\in\cS(\R^4)$. The time-ordered product of advanced fields with IR regularization is given by~\cite{epstein1973role}
\begin{multline}\label{eq:time_ordered_adv}
 \T(C_{1,\adv}(\mathbf{g};x_1),\ldots ,C_{m,\adv}(\mathbf{g};x_m))\equiv
  \T_{\adv}(\mathbf{g};C_1(x_1),\ldots,C_{m}(x_m)):=
 \\ 
 (-\ri)^{m} \frac{\delta}{\delta \mathbf{h}(J)}
 S(\mathbf{g}+\mathbf{h})S(\mathbf{g})^{-1}\bigg|_{\mathbf{h}=0}
 = \sum_{n=0}^\infty  \frac{\ri^n}{n!} \sum_{l_1,\ldots,l_n}e_{l_1}\ldots e_{l_n} \times
 \\
 \int\rd^4  y_1\ldots\rd^4 y_n \,g_{l_1}(y_1)\ldots g_{l_n}(y_n)
 \Adv(\cL_{l_1}(y_1),\ldots,\cL_{l_n}(y_n);J),
\end{multline}
where $J:=(C_1(x_1),\ldots,C_m(x_m))$ and we used the notation \eqref{eq:fun_der}. The time-ordered product of the retarded fields is given by a~similar formula with the advanced product replaced by the retarded product. The time-ordered products of the advanced or retarded fields are graded-symmetric and satisfy the axiom of causality. However, they are not translationally covariant.

The Wightman function of the advanced/retarded fields with IR regularization are the VEVs of the product of the advanced/retarded fields
\begin{equation}\label{eq:wightman_IR}
 (\Omega| C_{1,\adv/\ret}(\mathbf{g};x_1)\ldots C_{m,\adv/\ret}(\mathbf{g};x_m) \Omega),
\end{equation}
whereas the Green functions of the advanced/retarded fields with IR regularization are defined by
\begin{equation}\label{eq:green_IR}
 (\Omega| \T(C_{1,\adv/\ret}(\mathbf{g};x_1),\ldots ,C_{m,\adv/\ret}(\mathbf{g};x_m)) \Omega).
\end{equation}
They are formal power series in coupling constants $e_1,\ldots,e_\mathrm{q}$ with coefficients in $\cS'(\R^{4m})$. In order to obtain the physical Wightman and Green functions one has to get rid of the switching functions. To this end, one takes the adiabatic limit \eqref{eq:intro_W_G} of each term in the expansion  of \eqref{eq:wightman_IR} and \eqref{eq:green_IR} in powers of the coupling constants.

\subsection{Generalized \texorpdfstring{$\Adv$, $\Ret$, $\Dif$}{A, R, D} products with a~partition}\label{sec:aux}

Set $I \!=\! (B_1(y_1),\ldots,B_n(y_n))$, $J = (C_1(x_1),\ldots,C_m(x_m))$. Fix a~strictly increasing sequence of natural numbers $P = (p_0,\ldots,p_k)$, such that $p_0=0$ and $p_k=m$ and define a~partition $J_1,\ldots,J_k$ of $J$ such that $J_j$ is a~contiguous subsequence of $J$ starting at the position $p_j+1$ and ending at $p_{j+1}$. The generalized advanced product with the partition $P$ is given by
\begin{multline}\label{eq:def_gen_adv}
 \Adv(I;J;P):=
 (-\ri)^{n+m} 
 \frac{\delta}{\delta \mathbf{h}_k(J_k)}\ldots\frac{\delta}{\delta \mathbf{h}_1(J_1)}\frac{\delta}{\delta \mathbf{g}(I)}
 \\
  S(\mathbf{g}+\mathbf{h}_1)S(\mathbf{g})^{-1} \ldots  S(\mathbf{g}+\mathbf{h}_k)S(\mathbf{g})^{-1}\bigg|_{\substack{\mathbf{g}=0\\\mathbf{h}=0}},
\end{multline}
where $\mathbf{g}=\sum_{j=1}^n g_j\otimes B_j$ and $\mathbf{h}_i=\sum_{j=p_{i-1}}^{p_i} h_j\otimes C_j$. Similarly, the generalized retarded product with the partition $P$ is given by
\begin{multline}\label{eq:def_gen_ret}
 \Ret(I;J;P):=
 (-\ri)^{n+m} 
 \frac{\delta}{\delta \mathbf{h}_k(J_k)}\ldots\frac{\delta}{\delta \mathbf{h}_1(J_1)}\frac{\delta}{\delta \mathbf{g}(I)}~
 \\
 S(\mathbf{g})^{-1}S(\mathbf{g}+\mathbf{h}_1) \ldots S(\mathbf{g})^{-1} S(\mathbf{g}+\mathbf{h}_k)\bigg|_{\substack{\mathbf{g}=0\\\mathbf{h}=0}}.
\end{multline}
The generalized $\Dif$ product with the partition $P$ is by definition
\begin{equation}\label{eq:def_dif_gen}
 \Dif(I;J;P):= \Adv(I;J;P) - \Ret(I;J;P).
\end{equation}
If $k=1$, i.e. $P=(0,m)$, then the generalized $\Adv$, $\Ret$, $\Dif$ products coincide with the standard $\Adv$, $\Ret$, $\Dif$ products. If $n=0$, i.e. $I=\emptyset$, we have
\begin{equation}\label{eq:def_gen_T}
 \Adv(\emptyset;J;P) = \Ret(\emptyset;J;P) =\T(J_1)\T(J_2)\ldots \T(J_k) =:\T(J;P).
\end{equation}
If $I$ is the list of the interaction vertices
\begin{equation}\label{eq:I_vertices}
  I = (\mathcal{L}_{l_1}(y_1),\ldots,\mathcal{L}_{l_n}(y_n)),
\end{equation}
then we have
\begin{align}\label{eq:gen_adv_ret_dif_std}
 \Adv(I;J;P) &= (-\ri)^{n} \frac{\delta}{\delta \mathbf{g}(I)}\T_\adv(\mathbf{g};J_1) \ldots \T_\adv(\mathbf{g};J_k),
  \\
 \Ret(I;J;P) &=  (-\ri)^{n} \frac{\delta}{\delta \mathbf{g}(I)}\T_\ret(\mathbf{g};J_1) \ldots \T_\ret(\mathbf{g};J_k),
\end{align}
where $\T_\adv(\mathbf{g};J)$ is given by \eqref{eq:time_ordered_adv}. Since $\T_{\adv/\ret}(\mathbf{g};C(x))=\,C_{\adv/\ret}(\mathbf{g};x)$, the terms of order $e_{l_1}\ldots e_{l_n}$ in the formal expansions of the Wightman and Green functions with IR regularization may be expressed by the VEVs of $\Adv(I;J;P)$ or $\Ret(I;J;P)$. More precisely, the terms of order $e_{l_1}\ldots e_{l_n}$ in the expansion of the expressions \eqref{eq:wightman_IR} and \eqref{eq:green_IR} in powers of the coupling constants have the form
\begin{multline}\label{eq:w_g_adv}
 \int\rd^4 y_1\ldots\rd^4 y_n\,g_{l_1}(y_1)\ldots g_{l_n}(y_n)  
 \\
 (\Omega|\Adv(\mathcal{L}_{l_1}(y_1),\ldots,\mathcal{L}_{l_n}(y_n);C_1(x_1),\ldots,C_m(x_m);P)\Omega),
\end{multline}
where $P=(0,1,2,\ldots,m)$ in the case of Wightman functions and $P=(0,m)$ in the case of Green functions (in what follows we consider only these two types of sequences~$P$). The formula \eqref{eq:w_g_adv} holds for the Wightman and Green functions with IR regularization defined in terms of the advanced fields. In the case of the Wightman and Green functions with IR regularization defined in terms of the retarded fields the product $\Adv(I;J;P)$ has to be replaced by $\Ret(I;J;P)$.

Note that the Green functions are expressed by the ordinary advanced product $\Adv(I;J)$. This is not true for the Wightman functions. In the formula \eqref{eq:w_g_adv} the arguments of $\Adv(I;J;P)$ collectively denoted by $I$ are always of the form \eqref{eq:I_vertices} and the identity \eqref{eq:gen_adv_ret_dif_std} applies. However, in the inductive proof of the existence of the weak adiabatic limit presented in Section~\ref{Sec:existence_wAL} we will have to consider also the case when $I$ is the list of sub-polynomials of the interaction vertices
\begin{equation}
  I = (\mathcal{L}_{l_1}^{(s_1)}(y_1),\ldots,\mathcal{L}_{l_n}^{(s_n)}(y_n)).
\end{equation}

It follows from the definitions of the generalized $\Adv$, $\Ret$ and $\Dif$ products that
\begin{align}
 \Dif(I;J;P)=&
 (-\ri)^{n+m} 
 \frac{\delta}{\delta \mathbf{h}_k(J_k)}\ldots\frac{\delta}{\delta \mathbf{h}_1(J_1)}\frac{\delta}{\delta \mathbf{g}(I)}
 \\
 &[S(\mathbf{g}),S(\mathbf{g})^{-1} S(\mathbf{g}+\mathbf{h}_1)S(\mathbf{g})^{-1} \ldots  S(\mathbf{g}+\mathbf{h}_k)S(\mathbf{g})^{-1}]\bigg|_{\substack{\mathbf{g}=0\\\mathbf{h}=0}}
 \\
 =&
 (-\ri)^{n+m} 
 \frac{\delta}{\delta \mathbf{h}_k(J_k)}\ldots\frac{\delta}{\delta \mathbf{h}_1(J_1)}\frac{\delta}{\delta \mathbf{g}(I)}
 \\
 &S(\mathbf{g})^{-1}[S(\mathbf{g}), S(\mathbf{g}+\mathbf{h}_1)S(\mathbf{g})^{-1} \ldots  S(\mathbf{g}+\mathbf{h}_k)S(\mathbf{g})^{-1}]\bigg|_{\substack{\mathbf{g}=0\\\mathbf{h}=0}}.
\end{align}
The above formula implies that $\Dif(I;J;P)$ may be expressed as a~combination of terms of the form
\begin{equation}\label{eq:dif_com}
 [T(I_1),\aT(I_2)\Adv(I_3;J;P)]
\end{equation}
or a~combination of terms of the form
\begin{equation}\label{eq:dif_com2}
 \aT(I_1)[T(I_2),\Adv(I_3;J;P)],
\end{equation}
where the concatenation of the sequences $I_1,I_2,I_3$ is some permutation of the sequence $I$ and the number of elements of $I_3$ is strictly less than the number of elements of $I$. Moreover, as a~consequence of the causal factorization property the generalized $\Adv$, $\Ret$ and $\Dif$ products have the same support properties \eqref{eq:supp_adv_ret_dif} as the standard $\Adv$, $\Ret$ and $\Dif$ products.

\section{Existence of weak adiabatic limit}\label{Sec:existence_wAL}

In this section we present our main results. The objective is to prove the existence of the weak adiabatic limit in a~large class of models including all models with interaction vertices of dimension four. To this end, we have to prove that the limits \eqref{eq:intro_W_G} exist, where the interacting fields are the advanced or retarded fields defined in Section \ref{sec:W_G_IR}. In view of the discussion of Section \ref{sec:aux} it is enough to show that the limits:
\begin{multline}\label{eq:wAL}
 \lim_{\epsilon\searrow 0}\int\rd^4 y_1\ldots\rd^4 y_n\,\rd^4 x_1\ldots\rd^4 x_m\, g_\epsilon(y_1,\ldots,y_n) \, f(x_1,\ldots,x_m)
 \\
 (\Omega|\Adv(\cL_{l_1}(y_1),\ldots,\cL_{l_n}(y_n);C_1(x_1),\ldots,C_m(x_m);P)\Omega)
\end{multline}
and the analogous limit with the generalized advanced product $\Adv(I;J;P)$ replaced by the generalized retarded product $\Ret(I;J;P)$
\begin{enumerate}[leftmargin=*,label={(\arabic*)}]
 \item  exist for all $n,m\in\N_0$, $l_1,\ldots,l_n\in\{1,\ldots,\mathrm{q}\}$, $C_1,\ldots,C_m\in\Fa$, $f\in\cS(\R^{4m})$ and an arbitrary sequence $P$ of the form which was considered in Section~\ref{sec:aux},
 \item are independent of $g\in\cS(\R^{4n})$ such that $g(0,\ldots,0)=1$ and
 \item have the same values in the retarded and advanced case.
\end{enumerate} 
By definition $g_\epsilon(y_1,\ldots,y_n):=g(\epsilon y_1,\ldots,\epsilon y_n)$.

We consider first theories with only massive particles. In this case the existence of the limit \eqref{eq:wAL} was shown by Epstein and Glaser in~\cite{epstein1973role}. In Section~\ref{sec:weak_massive_proof} we give a slightly modified proof of this fact. It shall be regarded as the preparation for the more involved proof for theories with massless particles which is outlined in Section~\ref{sec:idea} and presented in full detail in Section~\ref{sec:PROOF}. Sections \ref{sec:math}, \ref{sec:prod} and \ref{sec:split} contain intermediate results.

\subsection{Proof for massive theories}\label{sec:weak_massive_proof}

The proof of the existence of the weak adiabatic limit for theories with only massive particles relies on the presence of the mass gap in the energy--momentum spectrum of these theories.   This property, which holds only in the case of purely massive models, means that the vacuum state is separated from the rest of the spectrum. The presence of the mass gap implies that there is a~neighborhood $\mathcal{O}$ of $0$ in $\R^{4n}$ such that the distribution
\begin{multline}\label{eq:weak_massive_proof}
 \F{\mathrm{a}}(q_1,\ldots,q_n)= \int \rd^4 y_1\ldots\rd^4 y_n\rd^4 x_1\ldots\rd^4 x_m 
 \exp(\ri q_1\cdot y_1+\ldots+\ri q_n\cdot y_n) 
 \\
 f(x_1,\ldots,x_m) (\Omega|\Adv(\cL_{l_1}(y_1),\ldots,\cL_{l_n}(y_n);C_1(x_1),\ldots,C_m(x_m);P)\Omega)
\end{multline}
restricted to test functions supported in $\mathcal{O}$ is a~smooth function. Before proving this let us observe that with the use of the above distribution the limit \eqref{eq:wAL} can be rewritten in the form
\begin{equation}\label{eq:weak_massive_adiabatic}
 \lim_{\epsilon\searrow0}\int\mP{q_1}\ldots\mP{q_n}\,\F{\mathrm{a}}(-q_1,\ldots,-q_n)\, \F{g}_\epsilon(q_1,\ldots,q_n)
\end{equation}
Assuming that $\F{a}$ is smooth in $\mathcal{O}$, the existence of the above limit is an immediate consequence of the following obvious lemma. 

\begin{lem}\label{lem:weak_massive_simple}
(A) Let $t\in C(\R^N)$. For every $g\in\cS(\R^N)$ we have
\begin{equation}
 \lim_{\epsilon\searrow0} \int \frac{\rd^N\!q}{(2\pi)^N}\, t(-q)\,
 \frac{1}{\epsilon^N}g(q/\epsilon)
 = t(0) \int \frac{\rd^N\!q}{(2\pi)^N} g(q).
\end{equation}

\noindent (B) For any $g\in\cS(\R^N)$ and $\chi\in C^\infty(\R^N)$ such that $0\notin\supp\,\chi$ we have 
\begin{equation}
 \lim_{\epsilon\searrow0}\frac{1}{\epsilon^N}g(q/\epsilon)\,\chi(q)=0
\end{equation}
in the topology of $\cS(\R^N)$. As a~consequence, if $t\in\cS'(\R^N)$, $0\notin\supp\,t$, then for every $g\in\cS(\R^N)$ it holds
\begin{equation}
\lim_{\epsilon\searrow0} \int \frac{\rd^N\!q}{(2\pi)^N}\, t(-q)\,
 \frac{1}{\epsilon^N}g(q/\epsilon) = 0.
\end{equation}
\end{lem}

Let us show that \eqref{eq:weak_massive_proof} is indeed a~smooth function in some neighborhood of the origin. To this end, for each $n\in\N_+$ we define a function $\Theta_n\in C^\infty(\R^{4n})$, called a~UV regular splitting function, such that it vanishes in some-neighborhood of the origin and
\begin{equation}\label{eq:def_example_splitting_fun}
 \Theta_n(y_1,\ldots,y_n)
 :=\rho\left(\frac{3n(y_1^0+\ldots+y_n^0)}{|(y_1,\ldots,y_n)|}\right) ~~~~\textrm{for}~~~|(y_1,\ldots,y_n)|\geq \ell,
\end{equation}
where $\ell$ is an arbitrary positive constant of the dimension of length, $$|(y_1,\ldots,y_n)|=(|y_1|^2+\ldots+|y_n|^2)^{1/2}$$ and $\rho\in C^\infty(\R)$ is a~real function having the following properties:
\begin{enumerate}[leftmargin=*,label={(\arabic*)}]
\item $\forall_{s\in\R}\,0\leq\rho(s)\leq1$,
\item $\forall_{s\in\R}\,\rho(s)=1-\rho(-s)$,
\item $\forall_{s>1}\,\rho(s)=1$, $\forall_{s<-1}\,\rho(s)=0$.
\end{enumerate}
The precise form of the functions $\Theta_n$ for $|(y_1,\ldots,y_n)| < \ell$ is irrelevant for our purposes -- we only assume that $\Theta_n$ is supported outside the origin. For $|(y_1,\ldots,y_n)|\geq \ell$ it holds:
\begin{equation}
 1-\Theta_n(y_1,\ldots,y_n)=\Theta_n(-y_1,\ldots,-y_n)
\end{equation}
and
\begin{equation}\label{eq:supp_theta}
 \Theta_n(\pm y_1,\ldots,\pm y_n) = 0~~~~\textrm{if}~~~\mp(y_1^0+\ldots+y_n^0)\geq\frac{1}{3n}|(y_1,\ldots,y_n)|.
\end{equation}
Moreover, $\F{\Theta}_n$ -- the Fourier transform of $\Theta_n$ -- is a~smooth function outside the origin and $\F{\Theta}_n(k)$ vanishes at infinity faster than any power of $1/|k|$ (the proof of this statement is postponed to Section~\ref{sec:math}; it is a~simple consequence of Lemma~\ref{lem:splitting_theta_alpha}).

By the definition \eqref{eq:def_dif_gen} of the generalized product $\Dif$ the following identity holds
\begin{equation}\label{eq:decomposition_weak_massive}
\begin{split}
 (\Omega|\Adv(\cL_{l_1}(y_1),\ldots,\cL_{l_n}&(y_n);C_1(x_1),\ldots,C_m(x_m);P)\Omega)=
 \\
 (1-\Theta_n(y_1,\ldots,y_n))\,&(\Omega|\Adv(\cL_{l_1}(y_1),\ldots,\cL_{l_n}(y_n);C_1(x_1),\ldots,C_m(x_m);P)\Omega)
 \\
 +\Theta_n(y_1,\ldots,y_n)\,&(\Omega|\Ret(\cL_{l_1}(y_1),\ldots,\cL_{l_n}(y_n);C_1(x_1),\ldots,C_m(x_m);P)\Omega)
 \\
 +\Theta_n(y_1,\ldots,y_n)\,&(\Omega|\Dif(\cL_{l_1}(y_1),\ldots,\cL_{l_n}(y_n);C_1(x_1),\ldots,C_m(x_m);P)\Omega).
\end{split} 
\end{equation}
We will show that the contribution to \eqref{eq:weak_massive_proof} from each of the three terms of the RHS of Equation \eqref{eq:decomposition_weak_massive} is a~smooth function in some neighborhood of zero. For the first two terms this follows from the lemma below with $q'_1=\ldots=q'_m=0$, $B_1=\cL_{l_1},\ldots,B_{n}=\cL_{l_n}$. The lemma will be also used in the proof of the existence of the weak adiabatic limit for theories with massless particles. Its proof uses only the support properties of the functions $\Theta_n$, $(1-\Theta_n)$ and of the generalized advanced and retarded products.

\begin{lem}\label{lem:splitting_smooth}
Let $n,m\in\N_0$ and $B_1,\ldots,B_n,C_1,\ldots,C_m\in\Fa$. Moreover, let $P$ be any sequence of the form considered in Section~\ref{sec:aux}. For every $f\in\cS(\R^{4m})$ the distribution
\begin{multline}\label{eq:lem_splitting_ind_two_smooth}
 (q_1,\ldots,q_n,q'_1,\ldots,q'_m)\mapsto \int\rd^4 y_1\ldots\rd^4 y_n\rd^4 x_1\ldots\rd^4 x_m 
 \\[2pt]
 \exp(\ri q_1\cdot y_1+\ldots+\ri q_n\cdot y_n + \ri q'_1\cdot x_1+\ldots+\ri q'_m\cdot x_m) ~f(x_1,\ldots,x_m)
 \\[6pt]
 ~\Theta_n(y_1,\ldots,y_n)\,
 (\Omega|\Ret(B_1(y_1),\ldots,B_n(y_n);C_1(x_1),\ldots,C_m(x_m);P)\Omega)
\end{multline} 
is a~smooth function. The same holds for 
\begin{equation}
 (1-\Theta_n(y_1,\ldots,y_n))\,
 (\Omega|\Adv(B_1(y_1),\ldots,B_n(y_n);C_1(x_1),\ldots,C_m(x_m);P)\Omega).
\end{equation}
\end{lem}
\begin{proof}
Because of the support property of the generalized retarded distribution (cf. Section~\ref{sec:aux}) and the presence of the function $\Theta_n$ the integrand in \eqref{eq:lem_splitting_ind_two_smooth} for $|(y_1,\ldots,y_n)|\geq \ell$ may be nonzero only in the region
\begin{multline}\label{eq:proof_inclusion}
 \Gamma^-_{n,m} \cap \{ (y_1,\ldots,y_n):\, 
 (y_1^0+\ldots+y_n^0)+\tfrac{1}{3n}|(y_1,\ldots,y_n)|\geq0\} \times \R^{4m}
 \\
 \subset \{ (y_1,\ldots,y_n;x_1,\ldots,x_m): |(y_1,\ldots,y_n)| \leq \const \, |(x_1,\ldots,x_m)| \},
\end{multline} 
where the cone $\Gamma^-_{n,m}$ is given by \eqref{eq:def_gen_cones}. To prove the above inclusion we first note that for any $(y_1,\ldots,y_n;x_1,\ldots,x_m)\in\Gamma^-_{n,m}$ it holds
\begin{equation} \label{eq:proof_inclusion_bound}
 y^0_j \leq \,|(x_1,\ldots,x_m)|
 ~~~~\mathrm{and}~~~~
 |\vec{y}_j| \leq \,2|(x_1,\ldots,x_m)| - y^0_j.
\end{equation}
Since $(y_1^0+\ldots+y_n^0) + \frac{1}{3n}|(y_1,\ldots,y_n)|\geq 0$ we have
\begin{equation}
 -y^0_j \leq (y^0_1 + \ldots + y^0_n) - y^0_j + \frac{1}{3n}(|\vec{y}_1|+\ldots+|\vec{y}_n|) + \frac{1}{3n}(|y^0_1|+\ldots+|y^0_n|).
\end{equation}
Combining the above inequality with the first bound in \eqref{eq:proof_inclusion_bound} we get
\begin{equation}\label{eq:proof_inclusion_third_bound} 
 |y^0_j| \leq (n-1) \,|(x_1,\ldots,x_m)| + \frac{1}{3n}(|\vec{y}_1|+\ldots+|\vec{y}_n|) + \frac{1}{3n}(|y^0_1|+\ldots+|y^0_n|).
\end{equation}
After summing both sides of the second bound in \eqref{eq:proof_inclusion_bound} and of the bound \eqref{eq:proof_inclusion_third_bound} over $j$ from $1$ to $n$ we obtain
\begin{align} 
 &|\vec{y}_1|+\ldots+|\vec{y}_n| \leq 2n \,|(x_1,\ldots,x_m)| + |y^0_1|+\ldots+|y^0_n|,
 \\
 &\frac{2}{3}(|y^0_1|+\ldots+|y^0_n|) \leq n(n-1) \,|(x_1,\ldots,x_m)| + \frac{1}{3}(|\vec{y}_1|+\ldots+|\vec{y}_n|),
\end{align}
respectively. The inclusion \eqref{eq:proof_inclusion} follows from the above bounds.

As a~result there exists a~function $\chi\in C^\infty(\R^{4n+4m})$ such that $\chi\equiv 1$ in some neighborhood of the support of the integrand in \eqref{eq:lem_splitting_ind_two_smooth} and 
\begin{equation}
 (y_1,\ldots,y_n;x_1,\ldots,x_m)\mapsto \chi(y_1,\ldots,y_n;x_1,\ldots,x_m) f(x_1,\ldots,x_m)
\end{equation}
is a~Schwartz function. To show the statement of the lemma we replace $f$ in \eqref{eq:lem_splitting_ind_two_smooth} by the above function and use the following fact. For any $t\in\cS'(\R^N)$ and $h\in\cS(\R^N)$, the Fourier transform of $h(x)t(x)$ is a~smooth function.
\end{proof}

Let us investigate the third term on the RHS of Equation~\eqref{eq:decomposition_weak_massive} involving the product $\Dif$. We will first show that the distribution
\begin{multline} \label{eq:weak_massive_f_dif}
 \F{\mathrm{d}}(q_1,\ldots,q_n):=\int \rd^4 x_1\ldots\rd^4 x_m 
 \exp(\ri q_1\cdot y_1+\ldots+\ri q_n\cdot y_n) f(x_1,\ldots,x_m)
 \\
 (\Omega|\Dif(\cL_{l_1}(y_1),\ldots,\cL_{l_n}(y_n);C_1(x_1),\ldots,C_m(x_m);P)\Omega)
\end{multline}
vanishes in some neighborhood of zero. Using the result of Section~\ref{sec:aux} we represent $\F{d}$ as as a~linear combination of the following distributions
\begin{multline} 
 (q_1,\ldots,q_n)\mapsto \int \rd^4 y_1\ldots\rd^4 y_n\rd^4 x_1\ldots\rd^4 x_m 
 \exp(\ri q_1\cdot y_1+\ldots+\ri q_n\cdot y_n) 
 \\
 f(x_1,\ldots,x_m) (\Omega|[\T(I_1),\aT(I_2)\Adv(I_3;C_1(x_1),\ldots,C_m(x_m))]\Omega),
\end{multline}
where the concatenation of the sequences $I_1,I_2,I_3$ is some permutation of a~sequence $(\cL_{l_1}(y_1),\ldots,\cL_{l_n}(y_n))$. The above distributions vanish in some neighborhood of the origin as a~consequence of both parts of Lemma~\ref{lem:aux_lemma} and the assumption that there are no massless fields. Note that the validity of the last statement follows ultimately from the presence of the mass gap in the energy--momentum spectrum. This is the only place in the proof where we use the assumption that all fields are massive.

The contribution to \eqref{eq:weak_massive_proof} from the last term on the RHS of Equation~\eqref{eq:decomposition_weak_massive} is of the form
\begin{equation}
 (q_1,\ldots,q_n)\mapsto\int\mP{k_1}\ldots\mP{k_n}\, \F{\mathrm{d}}(k_1,\ldots,k_n) \,\F{\Theta}_n(q_1-k_1,\ldots,q_n-k_n).
\end{equation}
Because of support properties of \eqref{eq:weak_massive_f_dif} and the smoothness of $\F{\Theta}_n$ outside the origin the above distribution is indeed a~smooth function in some neighborhood of the origin. The above result implies the existence of the limit \eqref{eq:wAL}. Since $\Dif(I;J;P)=\Adv(I;J;P)-\Ret(I;J;P)$ and the distribution $\F{d}$ vanishes in some neighborhood of zero the limit \eqref{eq:wAL} with the $\Adv$ product replaced by $\Ret$ also exists and has the same value as \eqref{eq:wAL}. This shows the existence of the weak adiabatic limit in massive theories.

\subsection{Idea of proof for theories with massless particles}\label{sec:idea}

The proof of the existence of the Wightman and Green functions in massive models was based on the fact that the distribution \eqref{eq:weak_massive_f_dif} vanishes in some neighborhood of the origin. This in turn follows from the existence of the mass gap in the energy--momentum spectrum of massive theories and is no longer true when massless particles are present. In fact, in theories with massless particles the distribution \eqref{eq:weak_massive_proof} is usually not a~continuous function in any neighborhood of zero. Note that by Part~(B) of Lemma~\ref{lem:weak_massive_simple} it is enough to control the behavior of this distribution in the~vicinity of the origin. To quantify the regularity of distributions $t\in\cS'(\R^N)$ near $0\in\R^N$ we shall introduce a~distributional condition  $t(q)=O^{\textrm{dist}}(|q|^\delta)$, where $\delta\in\R$. It generalizes the condition $f(q)=O(|q|^\delta)$, expressed in terms of the standard big O notation, which applies when $f$ is a~function. We will prove that in a~large class of models with massless particles (cf. Assumption \ref{asm} stated in Section \ref{sec:PROOF} for the precise specification of this class) it is possible to normalize the time-ordered products such that the distribution \eqref{eq:weak_massive_proof} is of the form
\begin{equation}\label{eq:idea_a}
 \F{\mathrm{a}}(q_1,\ldots,q_n) = c + O^{\textrm{dist}}(|q_1,\ldots,q_n|^{1-\varepsilon}) 
 \textrm{~~for some~~} c\in\C
 \textrm{~~and any~} \varepsilon > 0 
\end{equation}
and the distribution \eqref{eq:weak_massive_f_dif} is of the form
\begin{equation}\label{eq:idea_d}
 \F{\mathrm{d}}(q_1,\ldots,q_n) = O^{\textrm{dist}}(|q_1,\ldots,q_n|^{1-\varepsilon})
 \textrm{~~for any~} \varepsilon > 0. 
\end{equation}
As we will see the existence of the weak adiabatic limit follows immediately from the above conditions. 

The conditions \eqref{eq:idea_a} and \eqref{eq:idea_d} are satisfied if the time-ordered products fulfill the normalization condition \ref{norm:wAL}, which is equivalent to \ref{norm:wAL2} (both conditions are formulated in Section~\ref{sec:PROOF}). The latter condition says that for any $\Fa$-product $F$ which is a product of the time-ordered products and all lists $\mathbf{u}=(u_1,\ldots,u_{k})$ of super-quadri-indices which involve only massless fields\footnote{We recall that a~super-quadri-index $u$ involves only massless fields if $u(i,\alpha)=0$ for all $i$ such that $A_i$ is a~massive field. If a~super-quadri-index $u$ involves only massless fields, then the monomial $A^u$ is a~product of massless generators.} 
\begin{align}
&(\Omega|F(\F{\cL}_{l_1}^{(u_1)}(q_1), \ldots, \F{\cL}_{l_{k}}^{(u_{k})}(q_{k}))\Omega)
 =(2\pi)^4 \delta(q_1+\ldots+q_{k})\,t(q_1,\ldots,q_{k-1}),
\\[6pt]
&\textrm{where~~}t(q_1,\ldots,q_{k-1})=O^{\textrm{dist}}(|q_1,\ldots,q_{k-1}|^{\omega-\varepsilon}) \textrm{~~for every~~} \varepsilon>0
\\
\label{eq:idea_wAL}
&\textrm{and~~}\omega :=  4 - \sum_{i=1}^{\mathrm{p}} [\dim(A_i) \ext_{\mathbf{u}}(A_i) + \der_{\mathbf{u}}(A_i)].
\end{align} 
The functions $\ext_{\mathbf{u}}(\cdot)$, $\der_{\mathbf{u}}(\cdot)$ are given by \eqref{eq:ext} and $\mathrm{p}$ is the number of basic generators. According to Theorem \ref{thm:main1}, stated in Section \ref{sec:PROOF}, this normalization condition may be imposed in all models satisfying Assumption \ref{asm}. We prove this theorem by induction. We assume that the time-ordered products with at most $n$ arguments satisfy the condition \ref{norm:wAL2} and show that it is possible to define time-ordered products with $n+1$ arguments such that this conditions holds. We first use the following fact which is an immediate consequence of the statement (1') of Theorem~\ref{thm:product_F} stated in Section~\ref{sec:prod} entitled \emph{Product}. Let $F$ and $F'$ be two $\Fa$-products. Assume that their VEVs satisfy the condition \eqref{eq:idea_wAL}. Then the VEV of their product \eqref{eq:F_product} also satisfies this condition. Consequently, the VEVs of the $\Dif$ and $\Adv'$ products with $n+1$ arguments fulfill the condition \eqref{eq:idea_wAL}. The second part of the proof of the inductive step is based on Theorem~\ref{thm:split} from Section~\ref{sec:split} entitled \emph{Splitting}. Using this theorem we show that if the VEV of the $\Dif$ product with $n+1$ arguments satisfies the condition \eqref{eq:idea_wAL}, then it is possible to define the advanced product with $n+1$ arguments such that its VEV also satisfies the condition \eqref{eq:idea_wAL}. 

The existence of the weak adiabatic limit in the class of models satisfying our assumptions is stated as Theorem \ref{thm:main2}. Let us explain the intuitive content of this theorem. Consider the distribution
\begin{multline}\label{eq:idea_adv}
 \F{\textrm{a}}^{u_1,\ldots,u_{k+m}}(q_1,\ldots,q_k;q'_1,\ldots,q'_m):=\int\rd^4 y_1\ldots\rd^4 y_k\rd^4 x_1\ldots\rd^4 x_m 
 \\
 \exp(\ri q_1\cdot y_1+\ldots+\ri q_k\cdot y_k + \ri q'_1\cdot x_1+\ldots+\ri q'_m\cdot x_m) ~f(x_1,\ldots,x_m)
 \\
 (\Omega|\Adv(\cL_{l_1}^{(u_1)}(y_1),\ldots,\cL_{l_k}^{(u_k)}(y_k);C_1^{(u_{k+1})}(x_1),\ldots,C^{(u_{k+m})}_m(x_m);P)\Omega).
\end{multline} 
Note that after smearing it with a test function in $q_1,\ldots,q_n$ we get a~continuous function of $q'_1,\ldots,q'_m$. For $q'_1=\ldots=q'_m=0$ and $u_1=\ldots u_{k+m}=0$ the above distribution coincides with \eqref{eq:weak_massive_proof}. We will prove that: 

\noindent (A) for any list $\mathbf{u}=(u_1,\ldots,u_{k+m})$ of super-quadri-indices involving only massless fields such that at least one of them is nonzero it holds
\begin{align}
 &\F{\textrm{a}}^{u_1,\ldots,u_{k+m}}(q_1,\ldots,q_k;q'_1,\ldots,q'_m) =O^\textrm{dist}(|q_1,\ldots,q_k|^{d-\varepsilon}) \textrm{~for any~} \varepsilon>0,
 \\
 \label{eq:idea_2a}
 &\textrm{where~~}
 d:= 1 - \sum_{i=1}^{\mathrm{p}} [\dim(A_i) \ext_{\mathbf{u}}(A_i)+ \der_{\mathbf{u}}(A_i)],
\end{align} 
(B) the condition \eqref{eq:idea_a} holds.

The proof of the theorem is by induction on $k$. It relies on the representation of $\Dif(I;J;P)$ as a combination of terms of the form \eqref{eq:dif_com2}. The VEV of $\Adv(I_2;J;P)$ satisfies the condition \eqref{eq:idea_2a} by the induction assumption and the VEVs of $\aT(I_1)$ and $\T(I_3)$ fulfill the condition \eqref{eq:idea_wAL}. Using the above-mentioned properties and the statement (2') and (3') of Theorem~\ref{thm:product_F} in Section~\ref{sec:prod} entitled \emph{Product} one shows that for all super-quadri-indices $u_1,\ldots,u_{n+m}$ which involve only massless fields it holds 
\begin{equation}\label{eq:idea_2d}
 \F{\textrm{d}}^{u_1,\ldots,u_{n+m}}(q_1,\ldots,q_n;q'_1,\ldots,q'_m) =O^\textrm{dist}(|q_1,\ldots,q_n|^{d-\varepsilon})  \textrm{~for any~} \varepsilon>0,
\end{equation}
where the above distribution is defined by \eqref{eq:idea_adv} with the $\Adv$ product replaced by the $\Dif$ product. Because of the presence of the graded commutator in \eqref{eq:dif_com2} it is possible to prove \eqref{eq:idea_2d} for all $u_1,\ldots,u_{n+m}$ which involve only massless fields using the validity of \eqref{eq:idea_2a} for all $u_1,\ldots,u_{n+m}$ which involve only massless fields such that at least one of them is nonzero. The condition \eqref{eq:idea_2d} implies \eqref{eq:idea_d}. Using \eqref{eq:idea_2d} and \eqref{eq:idea_d} as well as Theorem~\ref{thm:split_gen} from Section~\ref{sec:split} entitled \emph{Splitting} we obtain  \eqref{eq:idea_2a} and \eqref{eq:idea_a} with $k=n$.

\subsection{Mathematical preliminaries}\label{sec:math}

\begin{dfn}\label{def:O_not}
Let $t\in\cS'(\R^N\times \R^M)$. For $\delta\in\R$ we write 
\begin{equation}
 t(q,q') = O^{\mathrm{dist}}(|q|^\delta),
\end{equation}
where $q\in\R^N$ and $q'\in\R^M$ iff there exist a~neighborhood $\mathcal{O}$ of the origin in $\R^N\times\R^M$ and a~family of functions $t_\alpha\in C(\mathcal{O})$ indexed by multi-indices $\alpha$ such that 
\begin{enumerate}[label=(\arabic*),leftmargin=*]
 \item $t_\alpha \equiv 0$ for all but finite number of multi-indices $\alpha$,
 \item $t_\alpha \equiv 0$ if $\delta+|\alpha|< 0$,
 \item $|t_\alpha(q,q')|\leq \const\,|q|^{\delta+|\alpha|}$ for $(q,q')\in\mathcal{O}$,
 \item $t(q,q') = \sum_{\alpha}
 \partial_q^\alpha t_{\alpha}(q,q')$ for $(q,q')\in\mathcal{O}$.
\end{enumerate}
Note that the differential operator $\partial^\alpha_q$ and the factor $|q|^{\delta+|\alpha|}$ above involve only the variable $q\in\R^N$. If $N=0$ we write $t(q')=O^{\mathrm{dist}}(|\cdot|^\delta)$. By definition for any $\delta\leq 0$ we have $t(q')=O^{\mathrm{dist}}(|\cdot|^\delta)$ iff $t\in C(\R^M)$ and for $\delta>0$ we have $t(q')=O^\mathrm{dist}(|\cdot|^\delta)$ iff $t=0$. 
\end{dfn}

\noindent Let us make a~couple of remarks about the above definition:
\begin{enumerate}[label=(\arabic*),leftmargin=*]
 \item In our applications the exponent $\delta$ which appears in Definition \ref{def:O_not} will never be an integer. Usually, we set $\delta = d-\varepsilon$, where $d\in\Z$ and $\varepsilon\in(0,1)$. 

 \item The condition $t(q,q')=O^{\mathrm{dist}}(|q|^\delta)$ controls the behavior of the distribution $t$~only near the origin. In particular, if a~distribution $t\in\cS'(\R^N\times\R^M)$ vanishes in a~neighborhood of $0$, then $t(q,q')=O^{\mathrm{dist}}(|q|^{\delta})$ for arbitrarily large $\delta$. Moreover, if $t(q,q')=O^{\mathrm{dist}}(|q|^\delta)$, then $t(q,q')=O^{\mathrm{dist}}(|q|^{\delta'})$ for all $\delta'\leq \delta$. 

 \item If $t\in\cS'(\R^N\times\R^M)$, $t(q,q')=O^{\mathrm{dist}}(|q|^\delta)$, then there exist neighborhoods $\mathcal{O}_1$ and $\mathcal{O}_2$ of the origin in $\R^N$ and $\R^M$, respectively, such that   for every $g\in\cS(\R^N)$, $\supp\,g\subset\mathcal{O}_1$ the distribution
 \begin{equation}
  \int \frac{\rd^N q}{(2\pi)^N} \, t(q,q') g(q) 
 \end{equation}
 is a~continuous function for $q'\in\mathcal{O}_2$. In particular, the distribution $t(q,0)$ is well defined in $\mathcal{O}_1$ and it holds $t(q,0)=O^{\mathrm{dist}}(|q|^\delta)$.
 
 \item If $t\in C(\R^N\times\R^M)$ is such that $|t(q,q')|\leq \const\,|q|^\delta$ in some neighborhood of $0$, then $t(q,q')=O^{\mathrm{dist}}(|q|^\delta)$. There are, however, $t\in C(\R^N\times\R^M)$ such that $t(q,q')=O^{\mathrm{dist}}(|q|^\delta)$ and the bound $|t(q,q')|\leq \const\,|q|^\delta$ is violated in every neighborhood of $0$. An example of such function in the case $N=1$, $M=0$ and $\delta=2$ is $t(q)=q \sin(1/q)+3q^2\cos(1/q) = \partial_q [q^3 \cos(1/q)]$. 

 \item Finally, let us remark that a~very similar characterization of the regularity of distributions near the origin in one dimension was introduced by Estrada in~\cite{estrada1998regularization} for the investigation of the existence of an extension $t\in\cD'(\R)$ of a~distribution $t^0\in\cD'((0,\infty))$. 
\end{enumerate}

\begin{dfn}\label{def:lojasiewicz}
We say that a~distribution $t\in\cS'(\R^N)$ has a~value $t(0)\in\C$ at zero in the sense of {\L}ojasiewicz~\cite{lojasiewicz1957valeur} iff the limit below
\begin{equation}
 t(0):=\lim_{\epsilon\searrow0}\int \frac{\rd^N q}{(2\pi)^N}\, t(q) g_\epsilon(q)
\end{equation}
exists for any $g\in\cS(\R^N)$ such that $\int \frac{\rd^N q}{(2\pi)^N} g(q)=1$, $g_\epsilon(q)=\epsilon^{-N} g(q/\epsilon)$. The value $t(0)$ is usually called in the physical literature the adiabatic limit of the distribution $t$ at $0$. The distribution $t\in\cS'(\R^N)$ has zero of order $\omega\in\N_+$ at the origin in the sense of {\L}ojasiewicz iff $\partial^\gamma_q t(q)\big|_{q=0}=0$ for all multi-indices $\gamma$ such that $|\gamma|<\omega$, where $\partial^\gamma_q t(q)\big|_{q=0}$ is defined in the sense of {\L}ojasiewicz.
\end{dfn}

\begin{thm}\label{thm:math_adiabatic_limit}
(A) Let $t\in\mathcal{S}'(\R^N)$, $t(q)=O^\mathrm{dist}(|q|^\delta)$, where $\delta\in\R$. It follows that for any $g\in\mathcal{S}(\R^N)$ it holds
\begin{equation}
 \int \frac{\rd^N q}{(2\pi)^N}\, t(q) g_\epsilon(q) = O(\epsilon^\delta).
\end{equation}

\noindent (B) Let $t\in\cS'(\R^N)$ such that $t(q)=c+O^\mathrm{dist}(|q|^\delta)$, where $\delta > 0$ and $c\in\C$. We have $t(0)=c$ in the sense of {\L}ojasiewicz. 

\noindent (C) If $t(q)=O^\mathrm{dist}(|q|^\delta)$, where $\delta+1>\omega\in\N_+$, then $t$ has zero of order $\omega$ at the origin in the sense of {\L}ojasiewicz.
\end{thm}
\begin{proof}
Let us begin with the proof of part (A). It follows from Definition \ref{def:O_not} that there exists $\chi\in\cD(\R^N)$, $\chi\equiv 1$ on some neighborhood of $0$, such that
\begin{equation}
 \chi(q) t(q) = \chi(q)\sum_{\alpha}
 \partial_q^\alpha t_{\alpha}(q)
\end{equation}
for all $q\in\R^N$. For arbitrarily large $\rho>0$ it holds
\begin{equation}
 \lim_{\epsilon\searrow0} \,\frac{1}{\epsilon^{\rho}} \, (1-\chi(q))\, g_\epsilon(q)= 0~~~\textrm{in}~~\cS(\R^N),
\end{equation}
and consequently,
\begin{equation}
 \int \frac{\rd^N q}{(2\pi)^N}\, (1-\chi(q))\,t(q) g_\epsilon(q) = o(\epsilon^\rho).
\end{equation}
On the other hand, 
\begin{multline}
 \int \frac{\rd^N q}{(2\pi)^N}\, \chi(q)t(q) g_\epsilon(q) = \lim_{\epsilon\searrow0} \sum_{\alpha} 
 \int\frac{\rd^N q}{(2\pi)^N}\,
 \partial_q^\alpha t_{\alpha}(q)\, \chi(q)\,g_\epsilon(q)
 \\
 =\lim_{\epsilon\searrow0} \sum_{\alpha} \epsilon^{-|\alpha|}  (-1)^{|\alpha|}
 \int\frac{\rd^N q}{(2\pi)^N}\,t_{\alpha}(\epsilon q)\, \partial_q^\alpha(\chi(\epsilon q)\,g(q)) = O(\epsilon^\delta),
\end{multline}
since $|t_{\alpha}(q)|\leq\const\,|q|^{|\alpha|+\delta}$. This finishes the proof of Part~(A). Part~(B) follows immediately from Part (A). Part~(C) is a consequence of Part~(B) and the fact that $\partial_q^\gamma t(q)=O^\mathrm{dist}(|q|^{\delta-|\gamma|})$ if $t(q)=O^\mathrm{dist}(|q|^\delta)$.
\end{proof}

\begin{dfn}\label{def:splitting_function_theta}
The function $\Theta:\R^N\to\R$ is called a~UV regular splitting function in $\R^N$ iff
\begin{enumerate}[label=(\arabic*),leftmargin=*]
 \item $\Theta$ is smooth,
 \item $0\leq\Theta(y)\leq1$,
 \item $\Theta\equiv0$ in some neighborhood of the origin,
 \item $\forall_{\lambda>1}\Theta(\lambda y)=\Theta(y)$ for $|y|>\ell$,
\end{enumerate}
where $\ell$ is some positive constant of dimension of length.
\end{dfn}
An example of a~UV regular splitting function with $N=4n$ is the function $\Theta_n$ defined in Section~\ref{sec:weak_massive_proof}. Note that in comparison with \cite{epstein1973role,blanchard1975green,epstein1976adiabatic} the splitting functions which we will use are homogeneous only outside certain neighborhood of the origin. Because they are smooth everywhere we call them UV regular splitting functions. The UV-regularized splitting of a~Schwartz distribution $t\in\cS'(\R^N\times\R^M)$ is defined in the position space by
\begin{equation}\label{eq:splitting_general_position}
 t_\Theta(y,x):=\Theta(y)t(y,x).
\end{equation}
The result of the splitting $t_\Theta$ is again a~Schwartz distribution. Equivalently, in momentum space the splitting of a~distribution $t\in\cS'(\R^N\times\R^M)$ is given by
\begin{equation}\label{eq:splitting_general_momentum}
 \int \frac{\rd^N q}{(2\pi)^N}\,\FF{t_\Theta}(q,q') g(q) = \int \frac{\rd^N k}{(2\pi)^N}\, \F{t}(k,q')\int \frac{\rd^N q}{(2\pi)^N}\,\F{\Theta}(q-k)g(q)
\end{equation} 
for any $g\in\cS(\R^N)$. The splitting \eqref{eq:splitting_general_position}, which was introduced above, will be applied to distributions $t$ of the form
\begin{equation}\label{eq:t_dist_form}
 t(y,x) = s(y,x) f(x), ~~\textrm{where}~~ s\in\cS'(\R^N\times\R^M) ~~\textrm{and}~~ f\in\mathcal{S}(\R^M). 
\end{equation}
If $M=0$, the above condition means that the distribution $t$ is an arbitrary element of $\cS'(\R^N)$.

\begin{lem}\label{lem:splitting_function_derivative}
Let $N\geq 2$ and $\Theta$ be a~UV regular splitting function in $\R^N$. There exist functions $\F{\Theta}_\beta,\F{\Theta}^\mathrm{hom}_\beta,\F{\Theta}^\mathrm{rest}_\beta:\,\R^N\to\C$ for each multi-index $\beta$, $|\beta|=1$ such that
\begin{equation}\label{eq:lem_theta}
 \F{\Theta}(k) = \sum_{|\beta|=1}  \partial^\beta_k \F{\Theta}_\beta(k),
 ~~~~\F{\Theta}_\beta(k) = \F{\Theta}^\mathrm{hom}_\beta(k) + \F{\Theta}^\mathrm{rest}_\beta(k),
\end{equation}
where 
\begin{enumerate}[label=(\arabic*),leftmargin=*]
 \item $\F{\Theta}^\mathrm{hom}_\beta$ is smooth on $\R^N\setminus\{0\}$ and homogeneous of degree $-N+1$,
 \item $\F{\Theta}^\mathrm{rest}_\beta \in C^\infty(\R^N)$,
 \item $\F{\Theta}_\beta(k)$ vanishes at infinity faster than any power of $|k|$.
\end{enumerate}
It follows that $\F{\Theta}_\beta(k)$ and $\F{\Theta}(k)$ are smooth outside the origin and vanish at infinity faster than any power of $|k|$. Moreover, $\F{\Theta}_\beta$ is absolutely integrable.
\end{lem}
\begin{proof}
Let $\Theta_\beta(y):=(-\ri)\Theta(y) y^\beta/|y|^2$, $\Theta^\mathrm{hom}_\beta$ be the unique homogeneous distribution of degree $-1$ such that $\Theta^\mathrm{hom}_\beta(y)=\Theta_\beta(y)$ for $|y|>\ell$ and $\Theta^\mathrm{rest}_\beta(y)=\Theta_\beta(y)-\Theta^\mathrm{hom}_\beta(y)$. The first equality in \eqref{eq:lem_theta} follows from
\begin{equation}
 \Theta(y) = \sum_{|\beta|=1}  \ri y^\beta \Theta_\beta(y).
\end{equation}
It is evident that $\F{\Theta}^\mathrm{rest}_\beta \in C^\infty(\R^N)$ as $\Theta^\mathrm{rest}_\beta(y)$ is of compact support. We have
\begin{equation}
 \int\frac{\rd^N k}{(2\pi)^N}\,\F{\Theta}_\beta(k) \tilde{g}(k)
 =\!\int\rd^N y\, \left[(-\Delta)^{m}\Theta_\beta(y)\right] 
  \int \frac{\rd^N k}{(2\pi)^N}\, \frac{1}{|k|^{2m}} \tilde{g}(k)\exp(-\ri k y)
\end{equation}
for every $g\in\cS(\R^N)$ such that $0\notin\supp\,g$, where $\Delta$ is the Laplacian on $\R^N$. Due to the fact that $\Theta_\beta$ is smooth everywhere and homogeneous of degree $-1$ outside some neighborhood of zero it holds
\begin{equation}
 |(-\Delta)^{m}\Theta_\beta(y)|\leq \const \, (1+|y|)^{-2m-1}.
\end{equation}
Because we can choose $m$ at will the function $\F{\Theta}_\beta(k)$ is smooth outside the origin and vanishes at infinity faster than any power of $|k|$. Since $\Theta^\mathrm{hom}_\beta$ is homogeneous of degree $-1$, $\F{\Theta}^\mathrm{hom}_\beta$ is homogeneous of degree \mbox{$-N+1$}. 
\end{proof}

\begin{thm}\label{thm:math_splitting}
Let $N\geq 2$ and $t\in\cS'(\R^N\times\R^M)$ of the form \eqref{eq:t_dist_form} such that $\F{t}(q,q') = O^\mathrm{dist}(|q|^{d-\varepsilon})$, $d\in\Z$, $\varepsilon\in(0,1)$. There exist $c_\gamma\in C(\R^M)$ for multi-indices $\gamma$, $|\gamma|<d$ such that
\begin{equation}
 \FF{t_\Theta}(q,q') = O^\mathrm{dist}(|q|^{d-\varepsilon}) + \sum_{|\gamma|< d} c_\gamma(q') q^\gamma,  
\end{equation}
where $\Theta$ is an arbitrary UV regular splitting function.
\end{thm}
\begin{proof}
Suppose that $0\notin \supp \, \F{t}$. It follows from the definition \eqref{eq:splitting_general_momentum} of the splitting of a distribution, the fact that $\F{\Theta}$ is smooth outside the origin and the assumed form \eqref{eq:t_dist_form} of the distribution $t$ that in some neighborhood of the origin $\FF{t_\Theta}(q,q')$ is a smooth function of $q$ and $q'$. This concludes the proof in this case. Thus, because of the assumption $\F{t}(q,q') = O^\mathrm{dist}(|q|^{d-\varepsilon})$ it is enough to consider distributions $\F{t}$ which are globally of the form
\begin{equation}
 \F{t}(q,q') = \sum_{\alpha}
 \partial_q^\alpha t_{\alpha}(q,q'),
\end{equation}
where the functions $t_\alpha\in C(\R^N\times\R^M)$ are of compact support and satisfy the conditions listed in Definition \ref{def:O_not} with $\mathcal{O}=\R^N\times\R^M$ and $\delta=d-\varepsilon$. In particular, $t_\alpha\equiv0$ if $d+|\alpha|\leq 0$. By the definition \eqref{eq:splitting_general_momentum} of $\FF{t_\Theta}$ we get
\begin{multline}
 \int \frac{\rd^N q}{(2\pi)^N}\,\FF{t_\Theta}(q,q') g(q) 
 \\
 =
 \sum_{\alpha} (-1)^{|\alpha|}
  \int \frac{\rd^N k}{(2\pi)^N}\, t_\alpha(k,q')
 \int \frac{\rd^N q}{(2\pi)^N}\,\F{\Theta}(q-k)\, \partial_q^\alpha g(q)
\end{multline}
for any $g\in\cS(\R^N)$. Using Lemma~\ref{lem:splitting_function_derivative} we obtain
\begin{equation}\label{eq:splitting_thm_terms}
 \FF{t_\Theta}(q,q')
 =
 \sum_{\substack{\alpha,\beta\\|\beta|=1}} \partial^{\alpha+\beta}_q  
 \int \frac{\rd^N k}{(2\pi)^N}\,t_\alpha(k,q') \,\F{\Theta}_\beta(q-k),
\end{equation}
where the integrand on the RHS of the above equation is absolutely integrable. The theorem follows from Lemma~\ref{lem:splitting_theta_alpha} with $D=d+|\alpha|\geq1$.
\end{proof}

\begin{lem}\label{lem:splitting_theta_alpha}
Let $N\geq2$, $t\in C(\R^N\times\R^M)$ be of compact support, $|t(q,q')|\leq \const\, |q|^{D-\varepsilon}$ for all $(q,q')\in\R^N\times\R^M$ and fixed $D\in\N_+$ and $\varepsilon\in(0,1)$. For any UV regular splitting function $\Theta$ and any $\beta$, $|\beta|=1$, there exist $t'\in C(\R^N\times\R^M)$ and $c_\gamma\in C(\R^M)$ for multi-indices $\gamma$, $|\gamma|\leq D$, such that
\begin{equation}\label{eq:lem_splitting_identity}
 \int\frac{\rd^N k}{(2\pi)^N}\,t(k,q') \F{\Theta}_\beta(q-k) 
 = t'(q,q')
 + \sum_{|\gamma|\leq D} c_\gamma(q') q^\gamma
\end{equation}
and $|t'(q,q')|\leq \const\, |q|^{D+1-\varepsilon}$ for all $(q,q')\in\R^N\times\R^M$.
\end{lem}
\begin{proof}
We have
\begin{equation}\label{eq:lem_splitting_second_term}
 c_\gamma(q') = \frac{1}{\gamma!} \int\frac{\rd^N k}{(2\pi)^N}\,t(k,q') \left.\left[\partial_q^\gamma \F{\Theta}_\beta(q-k)\right]\right|_{q=0}
\end{equation}
and
\begin{equation}\label{eq:lem_splitting_first_term}
 t'(q,q')=\int\frac{\rd^N k}{(2\pi)^N}\, t(k,q') 
 \left[\F{\Theta}_\beta(q-k)- \F{\Theta}_\beta^D(q,k)\right], 
\end{equation} 
where
\begin{equation}\label{eq:lem_splitting_Taylor_def}
 \F{\Theta}_\beta^D(q,k):= \sum_{|\gamma|\leq D} \frac{1}{\gamma!}\, q^\gamma \left.\left[\partial_q^\gamma \F{\Theta}_\beta(q-k)\right]\right|_{q=0}
\end{equation}
is the Taylor approximation of degree $D$ of the function $q\mapsto \F{\Theta}_\beta(q-k)$. We will show below that expressions  \eqref{eq:lem_splitting_second_term} and \eqref{eq:lem_splitting_first_term} are well defined and have the required properties. Using Lemma~\ref{lem:splitting_theta_alpha} we get
\begin{equation}\label{eq:lem_splitting_Taylor_sum_bound}
 \left.\left[\partial_q^\gamma \F{\Theta}_\beta(q-k)\right]\right|_{q=0}  \leq \const\, \frac{1}{|k|^{N-1+|\gamma|}}
\end{equation}
and
\begin{equation}\label{eq:lem_splitting_Taylor_rest_bound}
 \left|\F{\Theta}_\beta(q-k)- \F{\Theta}_\beta^D(q,k)\right| \leq
 \const\,\sup_{\lambda\in[0,1]} \frac{|q|^{D+1}}{|\lambda q-k|^{N+D}}.
\end{equation}
The second bound is a~consequence of the Taylor theorem. Because of the assumptions about $t$ the functions $c_\gamma$ given by \eqref{eq:lem_splitting_second_term} are well defined and continuous.

Now let us turn to the function $t'$ defined by \eqref{eq:lem_splitting_first_term}. Upon introducing new integration variables $u=k/|q|$, we find 
\begin{multline}\label{eq:lem_splitting_def_t_prime}
 t'(q,q')=|q|^{D+1-\varepsilon}  
 \\
 \times\int\frac{\rd^N u}{(2\pi)^N}\,|q|^{-(D-\varepsilon)}
  t(|q|u,q') ~|q|^{N-1}\left[\F{\Theta}_\beta(q-|q|u)- \F{\Theta}_\beta^D(q,|q|u)\right].
\end{multline}
Since $|\lambda q-|q|u|\geq \frac{1}{2}|q| |u|$ for $|u|\geq 2$ and $\lambda\in[0,1]$, it follows from \eqref{eq:lem_splitting_Taylor_rest_bound} and $|t(q,q')|\leq \const \,|q|^{D-\varepsilon}$ that for $|u|\geq2$ the integrand above is bounded by $\const\,|u|^{-N-\varepsilon}$. For $|u|<2$ we use the inequalities:
\begin{equation}
\begin{split}
 |q|^{N-1}|\F{\Theta}_\beta(q-|q|u)|\leq \const \, |q/|q|-u|^{-N+1},
 \\[4pt]
 |u|^{D-\varepsilon}|q|^{N-1}|\F{\Theta}_\beta^D(q,|q|u)| \leq \const\,  |u|^{-N+1-\varepsilon}.
\end{split}
\end{equation}
Consequently, the integral in \eqref{eq:lem_splitting_def_t_prime} exists and is a~bounded function of $q$ and $q'$. The continuity of $t'$ is evident because the LHS of \eqref{eq:lem_splitting_identity} is a~continuous function of $q$ and $q'$ and $c_\gamma\in C(\R^M)$.
\end{proof}

\subsection{Product}\label{sec:prod}

\begin{dfn}[\underline{IR}-index]\label{def:IR1}
The translationally invariant distribution 
\begin{equation}
 t(x_1,\ldots,x_{n+1})\in\cS'(\R^{4(n+1)})
\end{equation}
has \underline{IR}-index $d\in\Z$ iff for any $\varepsilon > 0$ it holds
\begin{equation}
 \F{\underline{t}}(q_1,\ldots,q_n) = O^\mathrm{dist}(|q_1,\ldots,q_n|^{d-\varepsilon}),
\end{equation}
where $\underline{t}(x_1,\ldots,x_{n})=t(x_1,\ldots,x_{n},0)$.  If $n=0$ and $d\leq 0$, then by definition the distribution $t$ is an arbitrary constant. If $n=0$ and $d>0$, then \mbox{$t=0$}. Observe that if $t\in\cS'(\R^{4(n+1)})$ vanishes in some neighborhood of $0$ in momentum space, then the \underline{IR}-index of $t$ is arbitrary large. If $t$ has \underline{IR}-index $d\in\Z$, then it also has \underline{IR}-index $d'\in\Z$ for any $d'\leq d$. 
\end{dfn}
\noindent The following lemma is an immediate consequence of Part~(C) of Theorem~\ref{thm:math_adiabatic_limit}.
\begin{lem}\label{lem:IR_der}
Let $t\in\cS'(\R^{4(n+1)})$ be translationally invariant distribution which has \underline{IR}-index $d$. Then the distribution $\F{\underline{t}}$ has zero of order $d$ at the origin in the sense of {\L}ojasiewicz. 
\end{lem}

\begin{dfn}[IR-index]\label{def:IR2}
The distribution 
\begin{equation}
 t(x_1,\ldots,x_n;x'_1,\ldots,x'_m)\in\cS'(\R^{4n}\times\R^{4m})
\end{equation}
has IR-index $d\in\Z$, with respect to the variables $x_1,\ldots,x_n$ iff for any $\varepsilon > 0$ and every $f\in\cS(\R^{4m})$ it holds
\begin{multline}\label{eq:IR_index_integration}
 \int\mP{p_1}\ldots\mP{p_m}\,\F{t}(q_1,\ldots,q_n;q'_1-p_1,\ldots,q'_m-p_m) \F{f}(p_1,\ldots,p_m) 
 \\
 = O^\mathrm{dist}(|q_1,\ldots,q_n|^{d-\varepsilon}).
\end{multline}
The LHS of the above equation is the Fourier transform of the distribution 
\begin{equation}
 f(x'_1,\ldots,x'_m)\, t(x_1,\ldots,x_n;x'_1,\ldots,x'_m).
\end{equation}
By definition, if $n=0$ and $d\leq 0$, then the above condition is equivalent to saying that the distribution $t$ is a~continuous function. If $n=0$ and $d>0$, then $t=0$. Observe that if $t\in\cS'(\R^{4n}\times\R^{4m})$ vanishes in some neighborhood of $0\times \R^{4m}$ in momentum space, then the IR-index of $t$ with respect to the first $n$ variables is arbitrary large.  If $t$ has IR-index $d\in\Z$, then it also has IR-index $d'\in\Z$ for any $d'\leq d$. 
\end{dfn}

\begin{thm}\label{thm:product}
Let $t\in\cS'(\R^{4n})$, $t'\in\cS'(\R^{4n'})$ be translationally invariant distributions. Consider the translationally invariant distribution
\begin{multline}\label{eq:product_thm_dist}
 t''(x_1,\ldots,x_n,x'_1,\ldots,x'_{n'}):=
 t(x_1,\ldots,x_n)~ t'(x'_1,\ldots,x'_{n'})
 \\
 \times \prod_{k=1}^l  (\Omega|\normord{\partial^{\bar \alpha(k)} A_{\bar i(k)}(x_{\bar u(k)})}\,
 \normord{\partial^{\bar \alpha'(k)} A_{\bar i'(k)}(x'_{\bar u'(k)})}\Omega)
\end{multline}
such that for all $k\in\{1,\ldots,l\}$ the fields $A_{\bar i(k)}$, $A_{\bar i'(k)}$ are massless. In the above equation we used the notation introduced in Section~\ref{sec:F_prod}. Note that the product of distributions in \eqref{eq:product_thm_dist} is well defined in the sense of H{\"o}rmander. Assume that one of the following conditions is satisfied:
\begin{enumerate}[leftmargin=*,label={(\arabic*)}]
 \item $t$ and $t'$ have \underline{IR}-indices $d$ and $d'$, respectively, 
 \item $t$ has \underline{IR}-index $d$ and $t'$ has IR-index $d'$ with respect to $x'_1,\ldots,x'_{k'}$,
 \item $t$ has IR-index $d$ with respect to $x_1,\ldots,x_k$ and $t'$ has \underline{IR}-index $d'$,
\end{enumerate}
where $k\leq n$, $k'\leq n'$. We set
\begin{equation}
 d''=d+d'+ \sum_{i=1}^{\mathrm{p}} [\dim(A_i) \bar\ext(A_i)+ \bar\der(A_i)] - 4.
\end{equation}
For the definition of $\bar\ext(A_i)$ and $\bar\der(A_{i})$ see Equations \eqref{eq:ext_s} and \eqref{eq:der_s}. Observe that, by assumption, $\bar\ext(A_i)=0$, $\bar\der(A_i)=0$ unless $A_i$ is a~massless field. In respective cases, the distribution $t''$ has
\begin{enumerate}[leftmargin=*,label={(\arabic*')}]
 \item \underline{IR}-index $d''$,
 \item IR-index $d''$ with respect to all variables but $x'_{k'+1},\ldots,x'_{n'}$,
 \item IR-index $d''$ with respect to all variables but $x_{k+1},\ldots,x_{n}$.
\end{enumerate}
\end{thm}
\begin{proof}
First, let us observe that for massless fields $A_i$ and $A_{i'}$ it holds
\begin{multline}
 (\Omega|{\partial^{\alpha} A_{i}(x)}\,{\partial^{\alpha'} A_{i'}(x')}\Omega)
 =
 \partial^{\alpha}_x \partial^{\alpha'}_{x'} (\Omega|{A_{i}(x)}\,{A_{i'}(x')}\Omega)
 \\
 =
 \sum_{\gamma\in\N_0^4}c_{ii'}^\gamma  \int \mH{0}{k} \, k^{\alpha+\alpha'+\gamma} \exp(-\ri k\cdot (x-x')) ,
\end{multline}
where $\rd\mu_0(k) = \frac{1}{(2\pi)^3} \rd^4 k\,\theta(k^0) \delta(k^2)$ and $c_{ii'}^\gamma   \in\C$ are constants such that if $c_{ii'}^\gamma\neq 0$, then $\dim(A_i)=\dim(A_{i'})$ and $|\gamma|=\dim(A_i)+\dim(A_{i'})-2$. The above statement follows immediately from the form of the expressions for the two-point functions of the massless basic generators.

As a~consequence the Fourier transform of the distribution $t''$ is a~sum of terms which can be represented (up to a~multiplicative constant) in each of the following three forms:
\leqnomode
\begin{multline}\label{eq:Fourier_product_T_Tprime_one}
\tag{A} (2\pi)^4 \delta(q_1+\ldots+q_n+q'_1+\ldots+q'_{n'}) 
 \\
 \int\mH{0}{k_1}\ldots\mH{0}{k_l}
 \,(2\pi)^4 \delta(q_1+\ldots+q_n-k_1-\ldots-k_l) 
 \\
 k^\beta\,\underline{\F{t}}(q_1-k(I_1),\ldots,q_{n-1}-k(I_{n-1})) 
 \,\underline{\F{t}}'(q'_1+k(I'_1),\ldots,q'_{n'-1}+k(I'_{n'-1})),
\end{multline}
\begin{multline}\label{eq:Fourier_product_T_Tprime_two}
\tag{B} \int\mH{0}{k_1}\ldots\mH{0}{k_l}
 \,(2\pi)^4 \delta(q_1+\ldots+q_n-k_1-\ldots-k_l) 
 \\
 k^\beta\,\underline{\F{t}}(q_1-k(I_1),\ldots,q_{n-1}-k(I_{n-1})) 
 \,\F{t}'(q'_1+k(I'_1),\ldots,q'_{n'}+k(I'_{n'})),
\end{multline}
\begin{multline}\label{eq:Fourier_product_T_Tprime_three}
\tag{C} \int\mH{0}{k_1}\ldots\mH{0}{k_l}
 \,(2\pi)^4 \delta(q'_1+\ldots+q'_{n'}+k_1+\ldots+k_l) 
 \\
 k^\beta\,\F{t}(q_1-k(I_1),\ldots,q_{n}-k(I_{n})) 
 \,\underline{\F{t}}'(q'_1+k(I'_1),\ldots,q'_{n'-1}+k(I'_{n'-1})), 
\end{multline}
\reqnomode
where $\beta$ is a~multi-index such that
\begin{equation}
 |\beta| = \sum_{i=1}^{\mathrm{p}} [\dim(A_i) \bar\ext(A_i)+ \bar\der(A_i)] - 2l,
\end{equation}
the sets $I_1,\ldots,I_n$, $I'_1,\ldots,I'_{n'}$ are subsets of $\{1,\ldots,l\}$ such that
\begin{equation}
 \bigcup_{i=1}^{n} I_i = \bigcup_{i=1}^{n'} I'_i = \{1,\ldots,l\}
 ~~\textrm{and}~~I_i\cap I_j = \emptyset, 
 ~~I'_i\cap I'_j = \emptyset~~\textrm{for}~~i\neq j
\end{equation}
and 
\begin{equation}
 k(I) := \sum_{i\in I} k_i
\end{equation}
for any subset $I$ of $\{1,\ldots,l\}$. We remind the reader that the distribution $\underline{t}$ and $\underline{t}'$ are defined in terms of translationally invariant distributions $t$ and $t'$ by Equation \eqref{eq:dist_trans_inv}.

The representations \eqref{eq:Fourier_product_T_Tprime_one}, \eqref{eq:Fourier_product_T_Tprime_two} and \eqref{eq:Fourier_product_T_Tprime_three} will be used to prove the statements (1'), (2'), (3') of the theorem, respectively. The theorem follow from Lemma~\ref{lem:IR_index_product_general} and the fact that for $t\in\cS'(\R^N\times\R^M)$ the condition $t(q,q')=O^\mathrm{dist}(|q|^\delta)$ is invariant under the linear changes of variables collectively denoted by $q\in\R^M$. Note that the Fourier transform of the distribution $\underline{t}''$ may be easily read off from the representation \eqref{eq:Fourier_product_T_Tprime_one} which is used in the proof of the statement (1') of the theorem.
\end{proof}

\noindent The following fact will be needed in the proof of Lemma~\ref{lem:IR_index_product_general}.
\begin{lem}\label{lem:IR_index_product_general_aux}
The Riesz distribution
\begin{equation}\label{eq:lem_IR_index_product_general_s}
 \R^4\ni k \mapsto s(k) = \frac{\pi^3}{4}\, k^2\theta(k^0)\theta(k^2) \in\R
\end{equation}
has the following properties:
\begin{enumerate}[leftmargin=*,label={(\arabic*)}]
 \item $s\in C(\R^4)$,
 \item $\supp \,s\subset \{k\in\R^4:~k^2\geq 0,\, k^0\geq 0\}$,
 \item $|s(k)|\leq \const \,|k|^2$ for all $k\in\R^4$,
 \item $(2\pi)^4\delta(k) = \square^3 s(k)$ where $\square$ is d'Alembertian.
\end{enumerate}
\end{lem}

\begin{lem}\label{lem:IR_index_product_general}
Let $t \in\cS'(\R^{4n})$, $t'\in\cS'(\R^{4n'})$ be such that
\begin{equation}
\begin{gathered}
 t(q_1,\ldots,q_n)=O^\mathrm{dist}(|q_1,\ldots,q_k|^{\delta}),
 \\
 t'(q'_1,\ldots,q'_{n'})=O^\mathrm{dist}(|q'_1,\ldots,q'_{k'}|^{\delta'}),
\end{gathered}
\end{equation}
where $1\leq k\leq n$, $1\leq k'\leq n'$, $\delta,\delta'\in\R$. Then for any linear functionals
\begin{equation}\label{eq:functionals_K}
 K_1,\ldots,K_n,K'_1,\ldots,K'_{n'}: \R^{4l} \rightarrow \R^4
\end{equation}
it holds
\begin{multline}\label{eq:thm_prod_form}
 t''(Q,q_1,\ldots,q_n,q'_1,\ldots,q'_{n'}):=
 \int\mH{0}{k_1}\ldots\mH{0}{k_l}\,
 (2\pi)^4\delta(Q-k_1-\ldots-k_l)
 \\
 k^\beta\,t(q_1-K_1(k),\ldots,q_n-K_n(k)) 
 \,t'(q'_1+K'_1(k),\ldots,q'_{n'}+K'_{n'}(k)) 
 \\
 = O^\mathrm{dist}(|Q,q_1,\ldots,q_k,q'_1,\ldots,q'_{k'}|^{\delta''}),
\end{multline}
where $k=(k_1,\ldots,k_l)$ and $\delta'' = \delta + \delta'+2l+|\beta|-4$.
\end{lem}
\begin{proof}
The variables $q_1,\ldots,q_n$ will be denoted collectively by $q$, the variables $q_1,\ldots,q_k$ -- by $\mathbf{q}$ and similarly for the primed variables. By Definition \ref{def:O_not} the distributions $t$ and $t'$ admit (in some neighborhood of the origin) the following representations: 
\begin{equation}
 t(q)=\sum_\alpha\partial_\mathbf{q}^\alpha t_\alpha(q)~~~~\textrm{and}
 ~~~~ t'(q')=\sum_{\alpha'} \partial_{\mathbf{q}'}^{\alpha'} t'_{\alpha'}(q'),
\end{equation}
where the functions $t_\alpha$ and $t'_{\alpha'}$ are continuous and such that
\begin{equation}\label{eq:lem_product_bounds}
 |t_\alpha(q)|\leq \const\, |\mathbf{q}|^{|\alpha|+\delta}~~~~\textrm{and}
 ~~~~|t'_{\alpha'}(q')|\leq\const\, |\mathbf{q}'|^{|\alpha'|+\delta'}.
\end{equation}

The distribution $t''$ defined by \eqref{eq:thm_prod_form} is expressed in some neighborhood of $0$ by the following sum
\begin{equation}\label{eq:IR_index_product_general_derivatives}
 t''(Q,q,q')=\sum_{\alpha,\alpha'} \partial_{\mathbf{q}}^{\alpha}\partial_{\mathbf{q}'}^{\alpha'} \square_{Q}^3 t''_{\alpha\alpha'}(Q,q,q'),
\end{equation}
where
\begin{multline}\label{eq:thm_IR_index_product_general_T}
 t''_{\alpha\alpha'}(Q,q,q') := \int\mH{0}{k_1}\ldots\mH{0}{k_l}\,k^\beta
 s(Q-k_1-\ldots-k_l)
 \\
 t_\alpha(q_1-K_1(k),\ldots,q_n-K_n(k)) 
 ~t'_{\alpha'}(q'_1+K'_1(k),\ldots,q'_{n'}+K'_{n'}(k))
\end{multline}
and $s$ is a~continuous function given in Lemma~\ref{lem:IR_index_product_general_aux}. To prove the above statement we observe that as a~result of the support properties of the function $s$ and of the measure $\rd\mu_0$ the integration region in \eqref{eq:thm_IR_index_product_general_T} may be restricted to $|k_1|,\ldots,|k_l|\leq |Q|$. Thus, it follows from the linearity of the maps \eqref{eq:functionals_K} that \eqref{eq:IR_index_product_general_derivatives} holds for $Q,q,q'$ in sufficiently small neighborhood of $0$ in $\R^{4n+4n'+4}$.

The continuity of the functions $t''_{\alpha\alpha'}$ is evident. Thus, to prove the statement of the lemma it is enough to show that 
\begin{equation}\label{eq:lem_product_bound}
 |t''_{\alpha\alpha'}(Q,q,q')|\leq \const \,|Q,\mathbf{q},\mathbf{q'}|^{|\alpha|+\delta+|\alpha'|+\delta'+2l+|\beta|+2}.
\end{equation}
Upon changing the variables of integration $u_i=k_i/|Q|$ in \eqref{eq:thm_IR_index_product_general_T} we obtain
\begin{multline}
 \hspace{-2mm}t''_{\alpha\alpha'}(Q,q,q') = |Q|^{2l+|\beta|}
 \int_{|u_j|\leq 1}\hspace{-1.5mm}\mH{0}{u_1}\ldots\mH{0}{u_l}\,u^\beta
 s(Q-|Q| u_1-\ldots-|Q| u_l)
 \\ 
 t_\alpha(q_1-|Q| K_1(u),\ldots,q_n-|Q| K_n(u))\, 
 t'_{\alpha'}(q'_1+|Q| K'_1(u),\ldots,q'_{n'}+|Q| K'_{n'}(u)).
\end{multline}
The bound \eqref{eq:lem_product_bound} follows from \eqref{eq:lem_product_bounds} and the inequalities
\begin{equation}
 |s(Q-|Q| u_1-\ldots-|Q| u_l)|\leq \const\, |Q|^2
\end{equation}
and
\begin{equation}
 |\mathbf{q}-|Q|\mathbf{K}(u)|^\eta,|\mathbf{q}'-|Q|\mathbf{K}'(u)|^\eta \leq \const\, |Q,\mathbf{q},\mathbf{q'}|^{\eta}
\end{equation}
valid for $|u_1|,\ldots,|u_l|\leq 1$ and $\eta\geq 0$, where $\mathbf{K}(u)=(K_1(u),\ldots,K_k(u))$ and similarly for $\mathbf{K}'(u)$. This ends the proof.
\end{proof}

\noindent Finally, we arrive at the main result of this section.
\begin{thm}\label{thm:product_F}
Let $F$ and $F'$ be $\Fa$-products with $n$ and $n'$ arguments, respectively. Fix $B_1,\ldots,B_n$, $B'_1,\ldots,B'_{n'}\in\Fh$ and set 
\begin{equation}\label{eq:thm_product_d}
 d = d_0 - \sum_{i=1}^{\mathrm{p}} [\dim(A_i)\ext_\mathbf{s}(A_i)+ \der_\mathbf{s}(A_i)],
 ~~~
 d'= d'_0 - \sum_{i=1}^{\mathrm{p}} [\dim(A_i)\ext_{\mathbf{s}'}(A_i)+ \der_{\mathbf{s}'}(A_i)]
\end{equation}
for any lists of super-quadri-indices $\mathbf{s}=(s_1,\ldots,s_n)$ and $\mathbf{s}'=(s'_1,\ldots,s'_{n'})$, where $d_0,d'_0\in\Z$ are some fixed constants. For definitions of $\ext_{\mathbf{s}}(A_i)$ and $\der_\mathbf{s}(A_i)$ see \eqref{eq:ext}. Assume that the VEVs of the operator-valued distributions
\begin{equation}\label{eq:product_families}
 F(B_1^{(s_1)}(x_1),\ldots,B_n^{(s_n)}(x_n)),
 ~~~~
 F'(B_1^{\prime(s'_1)}(x'_1),\ldots,B_{n'}^{\prime(s'_{n'})}(x'_{n'}))
\end{equation}
satisfy one of the following conditions. They have
\begin{enumerate}[leftmargin=*,label={(\arabic*)}]
 \item \underline{IR}-indices $d$ and $d'$, respectively,
 \item \underline{IR}-index $d$ and IR-index $d'$ with respect to $x'_1,\ldots,x'_{k'}$, respectively,
 \item IR-index $d$ with respect to $x_1,\ldots,x_k$ and \underline{IR}-index $d'$, respectively,
\end{enumerate}
for all super-quadri-indices $s_1, \ldots, s_n$, $s'_1,\ldots, s'_{n'}$ which involve only massless fields. By the axiom \ref{axiom2} it holds $d,d'\in\Z$ unless the VEVs of the products \eqref{eq:product_families} vanish. Then the VEV of the product of $F$ and $F'$
\begin{equation}\label{eq:product_T_bis}
 F(B_1(x_1),\ldots,B_n(x_n))F'(B_1^{\prime}(x'_1),\ldots,B_{n'}^{\prime}(x'_{n'}))
\end{equation}
has
\begin{enumerate}[leftmargin=*,label={(\arabic*')}]
 \item \underline{IR}-index $d_0+d'_0-4$,
 \item IR-index $d_0+d'_0-4$ with respect to all variables but $x'_{k'+1},\ldots,x'_{n'}$,
 \item IR-index $d_0+d'_0-4$ with respect to all variables but $x_{k+1},\ldots,x_{n}$,
\end{enumerate} 
in the respective cases. In the case of the graded commutator 
\begin{equation}\label{eq:product_com}
 [F(B_1(x_1),\ldots,B_n(x_n)),F'(B_1^{\prime}(x'_1),\ldots,B_{n'}^{\prime}(x'_{n'}))],
\end{equation}
the above statements (1')--(3') are valid if the assumptions (1)-(3) hold for all super-quadri-indices $s_1, \ldots, s_n$, $s'_1, \ldots, s'_{n'}$ which involve only massless fields such that at least one of them is nonzero.
\end{thm}
\begin{proof}
We will prove first the statements about the product \eqref{eq:product_T_bis}. The proof is based on the formula \eqref{eq:vev_product_representation}. We will show that the statements (1')-(3') hold independently for each term of the sum on the RHS of this formula. First note that as a~result of Definitions \ref{def:IR1} and \ref{def:IR2} of the \underline{IR}- and IR-index and Part~(A) of Lemma~\ref{lem:aux_lemma} this is true for all the terms for which at least one of the super-quadri-indices $s_1,\ldots,s_n$, $s'_1,\ldots,s'_{n'}$ involves a~massive field. The statements in the cases (1')-(3') follow now directly from Theorem~\ref{thm:product} and the constraints \eqref{eq:constraints}. 

In order to prove the claim about the graded commutator \eqref{eq:product_com} we use Part~(B) of Lemma  \ref{lem:aux_lemma} and the claim about the product \eqref{eq:product_T_bis} which has just been proved.
\end{proof}

\subsection{Splitting}\label{sec:split}

\begin{thm}\label{thm:split}
Assume that the distribution
\begin{equation}
  (\Omega|\Dif(B_1(y_1),\ldots,B_n(y_n);B_{n+1}(y_{n+1}))\Omega)
\end{equation}
has \underline{IR}-index $d\in\Z$. Then there exist constants $c_\alpha\in\C$ for multi-indices $\alpha$, $|\alpha|<d$ such that the distribution
\begin{multline}
  (\Omega|\Adv(B_1(y_1),\ldots,B_n(y_n);B_{n+1}(y_{n+1}))\Omega) 
  \\
  - \sum_{\substack{\alpha\\|\alpha|<d}}\! c_\alpha \partial^\alpha \delta(y_1-y_{n+1})\ldots \delta(y_n-y_{n+1})
\end{multline}
has \underline{IR}-index $d$.
\end{thm}
\begin{proof}
Let $I=(B_1(y_1),\ldots,B_n(y_n))$. Using the definition \eqref{eq:def_dif} of the $\Dif$ product we get
\begin{equation}\label{eq:split1_decomposition}
\begin{split}
 (\Omega|\Adv(I;B_{n+1}(0))\Omega)
 =(1-\Theta_n(y_1,\ldots,y_n))\,&(\Omega|\Adv(I;B_{n+1}(0))\Omega)
 \\
 +\Theta_n(y_1,\ldots,y_n)\,&(\Omega|\Ret(I;B_{n+1}(0))\Omega)
 \\
 +\Theta_n(y_1,\ldots,y_n)\,&(\Omega|\Dif(I;B_{n+1}(0))\Omega),
\end{split} 
\end{equation}
where the functions $\Theta_n$ were introduced in Section~\ref{sec:weak_massive_proof}. It is enough to prove that the Fourier transform of each term of the RHS of the above equation is of the form 
\begin{equation}\label{eq:split1_form}
 \tilde{t}(q_1,\ldots,q_n) = \sum_{|\gamma|<d} c_\gamma (-\ri)^{|\gamma|} q^\gamma + O^\mathrm{dist}(|q_1,\ldots,q_n|^{d-\varepsilon}) 
\end{equation}
for any $\varepsilon\in(0,1)$ and some constants $c_\gamma\in\C$, where $\gamma$ is a~multi-index. For the last term it is an immediate consequence of Theorem~\ref{thm:math_splitting} with ${M=0}$ and Definition \ref{def:IR1} of the \underline{IR}-index. In the case of the Fourier transforms of the first and the second terms we prove, along the lines of Lemma \ref{lem:splitting_smooth}, that they are smooth functions and apply the Taylor theorem. 
\end{proof}

\begin{thm}\label{thm:split_gen}
Assume that the distribution
\begin{equation}\label{eq:thm_split_gen_D}
 (\Omega|\Dif(B_1(y_1),\ldots,B_n(y_n);C_1(x_1),\ldots,C_m(x_m);P)\Omega)
\end{equation}
has IR-index $d\in\Z$ with respect to $y_1,\ldots,y_n$. Then:
\begin{enumerate}[leftmargin=*,label={(\Alph*)}]
\item The distribution
\begin{equation}\label{eq:splitting2_proof_index}
 (\Omega|\Adv(B_1(y_1),\ldots,B_n(y_n);C_1(x_1),\ldots,C_m(x_m);P)\Omega)
\end{equation}
has IR-index $\min\{d,0\}$ with respect to $y_1,\ldots,y_n$. 
\item If $d=1$, then for every $f\in\cS(\R^{4m})$ and $\varepsilon>0$ there exists $c\in\C$ such 
\begin{equation}
\begin{aligned}
 &\int\rd^4 x_1\ldots\rd^4 x_m\,f(x_1,\ldots,x_m) \times 
 \\
 &\hspace{2cm}(\Omega|\Adv(\F{B}_1(q_1),\ldots,\F{B}_n(q_n);C_1(x_1),\ldots,C_m(x_m);P)\Omega) 
 \\
 =&
 \int\rd^4 x_1\ldots\rd^4 x_m\,f(x_1,\ldots,x_m) \times 
 \\
 &\hspace{2cm}(\Omega|\Ret(\F{B}_1(q_1),\ldots,\F{B}_n(q_n);C_1(x_1),\ldots,C_m(x_m);P)\Omega) 
 \\
 =&\, O^\mathrm{dist}(|q_1,\ldots,q_n|^{1-\varepsilon})+ c.
\end{aligned}
\end{equation}

\end{enumerate}
\end{thm}
\begin{proof}
Let $I= (B_1(y_1),\ldots,B_n(y_n))$, $J= (C_1(x_1),\ldots,C_m(x_m))$. Using the definition \eqref{eq:def_dif_gen} of the generalized $\Dif$ product we obtain 
\begin{equation}\label{eq:split_gen_decomposition}
\begin{split}
 (\Omega|\Adv(I;J;P)\Omega)
 =(1-\Theta_n(y_1,\ldots,y_n))\,&(\Omega|\Adv(I;J;P)\Omega)
 \\
 +\Theta_n(y_1,\ldots,y_n)\,&(\Omega|\Ret(I;J;P)\Omega)
 \\
 +\Theta_n(y_1,\ldots,y_n)\,&(\Omega|\Dif(I;J;P)\Omega),
\end{split} 
\end{equation}
where the functions $\Theta_n$ were introduced in Section~\ref{sec:weak_massive_proof}. By Definition \ref{def:IR2} of the IR-index it holds
\begin{multline}\label{eq:thm_split_gen_D_int}
 \int\!\mP{p_1}\ldots\mP{p_m}~\F{f}(p_1,\ldots,p_m) \times
 \\
 (\Omega|\Dif(\F{B}_1(q_1),\ldots,\F{B}_n(q_n);\F{C}_1(q'_1-p_1),\ldots,\F{C}_m(q'_m-p_m);P)\Omega)
 \\
 = O^\mathrm{dist}(|q_1,\ldots,q_n|^{d-\varepsilon})
\end{multline}
for every $\varepsilon>0$ and $f\in\cS(\R^{4m})$. We will prove that if $t(y_1,\ldots,y_n;x_1,\ldots,x_m)$ is any of the three terms on the RHS of Equation \eqref{eq:split_gen_decomposition}, then there exist $c_\gamma\in C(\R^{4m})$ for multi-indices $\gamma$, $|\gamma|<d$ such that
\begin{multline}\label{eq:splitting2_proof}
 \int\mP{p_1}\ldots\mP{p_m}\,\tilde{t}(q_1,\ldots,q_n;q'_1-p_1,\ldots,q'_m-p_m) \F{f}(p_1,\ldots,p_m) 
 \\
 = O^\mathrm{dist}(|q_1,\ldots,q_n|^{d-\varepsilon}) + \sum_{|\gamma|< d} c_\gamma(q') q^\gamma.
\end{multline}
For the first and the second terms this follows form Lemma~\ref{lem:splitting_smooth} from Section~\ref{sec:weak_massive_proof} and the Taylor theorem. For the last term it is an immediate consequence of Theorem~\ref{thm:math_splitting}. Thus, Equation \eqref{eq:splitting2_proof} is valid also for the distribution $t(y_1,\ldots,y_m;x_1,\ldots,x_m)=(\Omega|\Adv(I;J;P)\Omega)$. Part~(A) follows now from Definition of the IR-index since for $d\leq 0$ the last term on the RHS of \eqref{eq:splitting2_proof} is absent. 

Next, we observe that by the above result and the third remark below Definition \ref{def:O_not} at the beginning of Section~\ref{sec:math} we have
\begin{multline}
 \int\mP{p_1}\ldots\mP{p_m}\,(\Omega|\Adv(\F{B}_1(q_1),\ldots,\F{B}_n(q_n);\F{C}_1(p_1),\ldots,\F{C}_m(p_m);P)\Omega)
 \\ \times\F{f}(-p_1,\ldots,-p_m) 
 = O^\mathrm{dist}(|q_1,\ldots,q_n|^{d-\varepsilon})+ \sum_{|\gamma|< d} c_\gamma(0) q^\gamma.
\end{multline}
To show Part~(B) we note that the last term on the RHS of the above equation is a~constant for $d=1$ and use \eqref{eq:def_dif_gen} and \eqref{eq:thm_split_gen_D_int}.
\end{proof}

\subsection{Proof for theories with massless particles}\label{sec:PROOF}

In this section we prove the existence of the weak adiabatic limit in theories containing massless particles. The idea of the proof was described in Section~\ref{sec:idea}. Our proof is valid only if the time-ordered products of the sub-polynomials of the interaction vertices satisfy certain condition which is formulated below as the normalization condition \ref{norm:wAL}. In Theorem~\ref{thm:main1}, we prove that this condition may be imposed in all models satisfying Assumption~\ref{asm} stated below and is equivalent to the normalization condition \ref{norm:wAL2}. The latter condition is used in Theorem~\ref{thm:main2}, which states the existence of the weak adiabatic limit. The proof of the compatibility of \ref{norm:wAL} with other standard normalization conditions usually imposed on the time-ordered products, like for example, the unitarity or Poincar{\'e} covariance, is postponed to Section~\ref{sec:std_norm_con}.

Our proof of the existence of the Wightman and Green functions applies to the following class of models.
\begin{asm}\label{asm}
We assume that all interaction vertices $\cL_1,\ldots,\cL_\mathrm{q}$ of the models under consideration have even fermion number and satisfy one of the following conditions:
\begin{enumerate}[leftmargin=*,label={(\arabic*)}]
\item $\forall_{l\in\{1,\ldots,\mathrm{q}\}} \dim(\cL_l)=4$,
\item $\forall_{l\in\{1,\ldots,\mathrm{q}\}} \dim(\cL_l)=3$ and $\cL_l$ contains at least one massive field.\footnote{A polynomial $B\in\Fa$ contains at least one massive field if it is a~combination of products of generators, each with a~massive field factor.}
\end{enumerate}
Moreover, we assume throughout that the axiom \ref{axiom7} holds with $\CC=0$ in the case (1) and $\CC=1$ in the case (2), i.e. $\CC = 4 -\dim(\cL_l)$.
\end{asm}
Other conditions which have to be fulfilled by the interaction vertices are listed in Section~\ref{sec:W_G_IR}. They ensure the correct physical properties of the model, but do not play any role in the proof of the existence of the weak adiabatic limit presented in this section. We will, however, use them afterward to prove that the Wightman and Green functions have all the standard properties (cf. Section \ref{sec:properties}). Note that in the proof of the existence of the Wightman and Green functions in purely massive theories we only use the fact that the interaction vertices have even fermion number.

Let us introduce some notations which will be used throughout this section. For any list of super-quadri-indices $\mathbf{u}=(u_1,\ldots,u_k)$ we set
\begin{equation}\label{eq:proof_omega}
 \omega' := 4 - \sum_{i=1}^{\mathrm{p}} [\dim(A_i) \ext_{\mathbf{u}}(A_i)+ \der_{\mathbf{u}}(A_i)] = 4 - \sum_{j=1}^k \dim(A^{u_j}).
\end{equation}
The functions $\ext_{\mathbf{u}}(\cdot)$, $\der_{\mathbf{u}}(\cdot)$ are given by \eqref{eq:ext}, and $\mathrm{p}$ is the number of basic generators. Observe that since $\CC = 4 -\dim(\cL_l)$, by the axiom \ref{axiom7} we get
\begin{equation}\label{eq:sd_bound2}
 \sd\big((\Omega|\T(\cL^{(u_1)}_{l_1}(x_1),\ldots,\cL^{(u_{k-1})}_{l_{k-1}}(x_{k-1}),\cL^{(u_k)}_{l_k}(0)) \Omega)\big) 
 \leq 4(k-1) +\omega'.
\end{equation}
Moreover, if $B_1=\cL^{(u_1)}_{l_1},\ldots,B_k=\cL^{(u_{k})}_{l_k}$, then $\omega'=\omega$, where by definition
\begin{equation}\label{eq:proof_omega_prime}
 \omega:=4-\sum_{j=1}^{k} (4-\CC- \dim(B_j))
\end{equation}
for any $B_1,\ldots,B_k\in\Fh$. Note that $\omega$ defined above coincides for $k=n+1$ with $\omega$ given by \eqref{eq:omega}. By $\Fa^\cL$ we denote the subspace of $\mathcal{F}$ spanned by $\cL^{(u)}_l$ with $l\in\{1,\ldots,q\}$ and $u$ being a~super-quadri-index involving only massless fields. 

Let us recall that according to our definition the time-ordered products are any $\Fa$-products which satisfy the axioms formulated in Section \ref{sec:axioms}. The freedom in the definition of the time-ordered products is characterized in Section~\ref{sec:freedom}. Our proof of the existence of the weak adiabatic limit is applicable only if the time-ordered products are defined in such a way that they fulfill the following normalization condition. 
\begin{enumerate}[label=\bf{N.wAL},leftmargin=*]
\item\label{norm:wAL} 
For all $k\in\N_+$, $l_1,\ldots,l_k\in\{1,\ldots,\mathrm{q}\}$ and lists \mbox{$\mathbf{u}=(u_1,\ldots,u_k)$} of super-quadri-indices such that every $u_j$ involves only massless fields the Fourier transform of the distribution
\begin{equation}\label{eq:wAL_vev}
 (\Omega|\T(\mathcal{L}^{(u_1)}_{l_1}(x_1),\ldots,\mathcal{L}^{(u_{k})}_{l_k}(x_k))\Omega)
\end{equation}
is of the form
\begin{equation}
 (2\pi)^4 \delta(q_1+\ldots+q_k)\,t(q_1,\ldots,q_{k-1}), 
\end{equation}
where $t(q_1,\ldots,q_{k-1})$ has zero of order $\omega'$ at $q_1=\ldots=q_{k-1}=0$ in the sense of {\L}ojasiewicz. Note that $\omega'\in\Z$ unless the distribution \eqref{eq:wAL_vev} vanishes by the axiom \ref{axiom2}.
\end{enumerate}
Let us consider also the following normalization condition.
\begin{enumerate}[label=\bf{N.wAL'},leftmargin=*]
\item\label{norm:wAL2} 
The distribution 
\begin{equation}\label{eq:wAL_dist2}
 (\Omega|F(\mathcal{L}^{(u_1)}_{l_1}(x_1),\ldots,\mathcal{L}^{(u_{k})}_{l_k}(x_k))\Omega)
\end{equation}
has \underline{IR}-index equal to $\omega'$ given by \eqref{eq:proof_omega} for all $\Fa$-products $F$ which are finite linear combinations of products of the time-ordered products, all $k\in\N_+$, $l_1,\ldots,l_k\in\{1,\ldots,\mathrm{q}\}$ and all list of super-quadri-indices $\mathbf{u}=(u_1,\ldots,u_k)$ such that each $u_j$ involves only massless fields (i.e. $\ext_{\mathbf{u}}(A_i)=0$ unless $A_i$ is a~massless field).
\end{enumerate}

\begin{thm}\label{thm:main1} Suppose that Assumption \ref{asm} holds.
\begin{enumerate}[leftmargin=*,label={(\Alph*)}]
\item It is possible to define the time-ordered products in such a~way that the normalization condition \ref{norm:wAL2} is satisfied. 

\item The time-ordered products fulfill the normalization condition \ref{norm:wAL} iff they fulfill the normalization condition \ref{norm:wAL2}.
\end{enumerate}
\end{thm}

\begin{proof}
We will prove Part (A) of the theorem by induction with respect to the number of arguments of the time-ordered products. First, we have to consider the distributions of the form
\begin{equation}\label{eq:thm_main_one}
 (\Omega|\cL^{(u)}_{l}(x)\Omega),
\end{equation}
where the super-quadri-index $u$ involves only massless fields and $l\in\{1,\ldots,\mathrm{q}\}$. By assumption $\cL_{l}$ contains at least one massive field or $\dim(\cL_l)=4$. First note that \eqref{eq:thm_main_one} is always a~constant, which is nonzero iff $\cL^{(u)}_{l}$ contains a~nonzero constant term in the decomposition into monomials. The latter condition is never satisfied if $\cL_l$ contains a~massive field. Thus, \eqref{eq:thm_main_one} is zero in this case. If $\dim(\cL_l)=4$, then $\cL_l^{(u)}$ may contain a~nonzero constant term in the decomposition into monomials only if $\dim(A^u)=4$ which implies that $\omega'=4-\dim(A^u)=0$ (set $k=1$ in Equation \eqref{eq:proof_omega}). As a~result, the normalization condition \ref{norm:wAL2} is satisfied in both cases.

We proceed to the proof of the inductive step. Fix $n\in\N_+$ and assume that \ref{norm:wAL2} holds for all $k\leq n$ (it is enough to assume that the normalization condition \ref{norm:wAL2} holds only for $F$ being the time-ordered product). We will demonstrate that it is possible to define the time-ordered products such that \ref{norm:wAL2} is satisfied for $k=n+1$. Using the inductive hypothesis and the statement (1') of Theorem~\ref{thm:product_F} we obtain that for all ${m\in\N_+}$, super-quadri-indices $u_1,\ldots,u_{m}$ which involve only massless fields and $\Fa$-products $F$ which may be expressed as a~linear combinations of products of the time-ordered products with at most $n$ arguments, the distribution
\begin{equation}\label{eq:thm_main1_F}
 (\Omega|F(\cL_{l_1}^{(u_1)}(x_1),\ldots,\cL_{l_m}^{(u_{m})}(x_m))\Omega) 
\end{equation}
has \underline{IR}-index $\omega'$ given by \eqref{eq:proof_omega} with $k=m$. In particular, the distributions
\begin{gather}
 (\Omega|\Dif(B_1(x_1),\ldots,B_n(x_n);B_{n+1}(x_{n+1}))\Omega),
 \\
 (\Omega|\Adv'(B_1(x_1),\ldots,B_n(x_n);B_{n+1}(x_{n+1}))\Omega) 
\end{gather}
have \underline{IR}-indices $\omega$ given by \eqref{eq:proof_omega_prime} with $k={n+1}$ for all $B_1,\ldots,B_{n+1}\in \Fa^\cL$.

With the use of the above results, Theorem~\ref{thm:split} and the relation \eqref{eq:def_adv_prime} between the products $\Adv$ and $\Adv'$ we get that for all $B_1,\ldots,B_{n+1}\in \Fa^\cL$ there exist constants $c_\gamma\in\C$ indexed by multi-indices $\gamma$, $|\gamma|<\omega$ such that 
\begin{equation}\label{eq:main1_redefine}
 (\Omega|\T(B_1(x_1),\ldots,B_{n+1}(x_{n+1}))\Omega) + v(B_1(x_1),\ldots,B_{n+1}(x_{n+1}))
\end{equation}
has \underline{IR}-index $\omega$ given by \eqref{eq:proof_omega_prime} with $k=n+1$, where
\begin{equation}
 v(B_1(x_1),\ldots,B_{n+1}(x_{n+1})) := \sum_{\substack{\gamma\\|\gamma|<\omega}} c_\gamma \partial^\gamma \delta(x_1-x_{n+1})\ldots \delta(x_n-x_{n+1}).
\end{equation}
Using the normalization freedom described in Section~\ref{sec:freedom} we will show that it is possible to redefine the time-ordered products such that the condition \ref{norm:wAL2} is fulfilled. To this end, we modify the definition of the VEV of time-ordered product with ${n+1}$ arguments \eqref{eq:T_freedom} by adding to it the graded-symmetric map 
\begin{multline}
 v_{\textrm{sym}}(B_1(x_1),\ldots,B_{n+1}(x_{n+1})) :=
 \\
 \frac{1}{(n+1)!}\sum_{\pi\in \mathcal{P}_{n+1}}\,(-1)^{\mathbf{f}(\pi)} v(B_{\sigma(1)}(x_{\sigma(1)}),\ldots,B_{\sigma(n+1)}(x_{\sigma(n+1)})),
\end{multline}
where $\mathbf{f}(\pi)\in\Z/2\Z$ is the number of transpositions in $\pi$ that involve a~pair of fields with odd fermion number. Since the graded-symmetric part of \eqref{eq:main1_redefine} has \underline{IR}-index $\omega$ given by \eqref{eq:proof_omega_prime} with $k=n+1$ the same holds for the VEV of the redefined time-ordered product. This proves the condition \ref{norm:wAL2} for $F=\T$. The generalization for arbitrary $F$ considered in \ref{norm:wAL2} follows from the statement about the distribution \eqref{eq:thm_main1_F} made above. This ends the proof of Part (A).

To prove Part (B), we first observe that by Lemma~\ref{lem:IR_der} the condition \ref{norm:wAL2} with $F=\T$ implies the condition \ref{norm:wAL}. The reverse implication is also true since after imposing the condition \ref{norm:wAL} the freedom in defining the time-ordered products with $n+1$ arguments given the time-ordered products with $n$ arguments is the same as after imposing \ref{norm:wAL2}: two possible definitions of \eqref{eq:T_freedom} differ by the map \eqref{eq:v_freedom} which satisfies the conditions stated in Section \ref{sec:freedom} and is such that for all $B_1,\ldots,B_{n+1}\in\Fa^\cL$ the distribution $v(B_1(x_1),\ldots,B_{n+1}(x_{n+1}))$ is of the form \eqref{eq:freedom_form}, where the sum over $\gamma$ is restricted to $|\gamma|=\omega'$.
\end{proof}

Now we formulate our main theorem according to which the weak adiabatic limit exists in models satisfying Assumption \ref{asm}.

\begin{thm}\label{thm:main2}
Suppose that Assumption \ref{asm} holds and the time-ordered products satisfy the normalization condition \ref{norm:wAL}. Fix $m\in\N_0$, $C_1,\ldots,C_m\in\Fa$ and a~sequence $P$ of the form considered in Section~\ref{sec:aux}.
\begin{enumerate}[leftmargin=*,label={(\Alph*)}]
\item
Suppose that $\sum_{j=1}^{m}\mathbf{f}(C_j)$ is even. For any $k\in\N_0$ and any list $\mathbf{u}=(u_1,\ldots,u_{k+m})$ of super-quadri-indices such that every $u_j$ involves only massless fields (i.e. $\ext_{\mathbf{u}}(A_i)=0$, $\der_{\mathbf{u}}(A_i)=0$ unless $A_i$ is a~massless field) and at least one super-quadri-index from $\mathbf{u}$ is nonzero (i.e. $\sum_{i=1}^\mathrm{p}\ext_{\mathbf{u}}(A_i)\geq 1$) the distribution
\begin{equation}\label{eq:main_distribution_adv}
 (\Omega|\Adv(\cL_{l_1}^{(u_1)}(y_1),\ldots,\cL_{l_k}^{(u_k)}(y_k);C_1^{(u_{k+1})}(x_1),\ldots,C_m^{(u_{k+m})}(x_m);P)\Omega)
\end{equation}
has IR-index 
\begin{equation}\label{eq:main2_IR}
 d= 1 - \sum_{i=1}^{\mathrm{p}} [\dim(A_i) \ext_{\mathbf{u}}(A_i)+ \der_{\mathbf{u}}(A_i)]
\end{equation}
with respect to the variables $y_1,\ldots,y_k$ for all $l_1,\ldots,l_k\in\{1,\ldots,\mathrm{q}\}$. Note that $d\in\Z$ if the distribution \eqref{eq:main_distribution_adv} is nonzero. 
\item 
For all $k\in\N_0$, $g\in\cS(\R^{4n})$, $f\in\cS(\R^{4m})$, $l_1,\ldots,l_k\in\{1,\ldots,\mathrm{q}\}$ and $\varepsilon>0$ there exists $c\in\C$ such that
\begin{equation}\label{eq:thm_main2_adv}
\begin{aligned}
 &\int\rd^4 x_1\ldots\rd^4 x_m\,f(x_1,\ldots,x_m) \times
 \\
 &~~~(\Omega|\Adv(\F{\cL}_{l_1}(q_1),\ldots,\F{\cL}_{l_k}(q_k);C_1(x_1),\ldots,C_m(x_m);P)\Omega) 
 \\
 =&\int\rd^4 x_1\ldots\rd^4 x_m\,f(x_1,\ldots,x_m) \times
 \\
 &~~~(\Omega|\Ret(\F{\cL}_{l_1}(q_1),\ldots,\F{\cL}_{l_k}(q_k);C_1(x_1),\ldots,C_m(x_m);P)\Omega) 
 \\
  =&\, O^\mathrm{dist}(|(q_1,\ldots,q_k)|^{1-\varepsilon})+ c.
\end{aligned}  
\end{equation}
This, by Part (B) of Theorem~\ref{thm:math_adiabatic_limit}, implies the existence of the weak adiabatic limit \eqref{eq:wAL}.
\end{enumerate}
\end{thm}
\begin{proof}
We will prove the theorem by induction with respect to $k$. For $k=0$ by Equation \eqref{eq:def_gen_T} we have $\Adv(\emptyset;J;P)=\Ret(\emptyset;J;P) = \T(J;P)$. It follows that for any $f\in\cS(\R^{4m})$
\begin{equation}
 \int\!\mP{p_1}\ldots\mP{p_m}\tilde{f}(p_1,\ldots,p_m) 
 (\Omega|\T(\F{C}_1^{(u_{1})}(q'_1-p_1),\ldots,\F{C}_m^{(u_{m})}(q'_m-p_m);P)\Omega)
\end{equation}
is a~smooth function of $q'_1,\ldots,q'_m$. This implies Part~(A) for $k=0$ as a~result of Definition \ref{def:IR2} of the IR-index (note that $d\leq0$ since at least one super-quadri-index $u_1,\ldots,u_{m}$ is nonzero). Part~(B) is trivially true in this case. 

Now, let us assume that Part~(A) holds for $k\leq n-1$, $n\in\N_+$. We shall prove that both Part~(A) and (B) hold for $k=n$. We will first show that for $k=n$ and any list $\mathbf{u}=(u_1,\ldots,u_{k+m})$ of super-quadri-indices such that every $u_j$ involves only massless fields, including the case when all $u_j$ vanish, the distribution
\begin{equation}\label{eq:thm_main2_Dif}
 (\Omega|\Dif(\cL^{(u_1)}_{l_1}(y_1),\ldots,\cL^{(u_k)}_{l_k}(y_k);C_1^{(u_{k+1})}(x_1),\ldots,C_m^{(u_{k+m})}(x_m);P)\Omega)
\end{equation}
has IR-index $d$ given by \eqref{eq:main2_IR} with respect to the variables $y_1,\ldots,y_k$. The proof of this fact is based on the representation of $\Dif(I;J;P)$ as a combination of terms of the form \eqref{eq:dif_com2} and the statements (2'), (3') of Theorem \ref{thm:product_F}. Note that the assumptions of this theorem are satisfied because of the inductive assumption and the fact that the time-ordered and anti-time-ordered products satisfy the normalization condition \ref{norm:wAL2}.

Part~(A) for $k=n$ follows now from Part~(A) of Theorem~\ref{thm:split_gen} since by assumption in this case the IR-index $d$ is non-positive. To show Part~(B) for $k=n$, we first observe that the distribution
\begin{equation}\label{eq:thm_main2_Dif2}
 (\Omega|\Dif(\cL_{l_1}(y_1),\ldots,\cL_{l_k}(y_k);C_1(x_1),\ldots,C_m(x_m);P)\Omega)
\end{equation}
has IR-index $d=1$. Indeed, this is the distribution \eqref{eq:thm_main2_Dif} with $u_1=\ldots=u_{k+m}=0$. Part~(B) follows now from Part~(B) of Theorem~\ref{thm:split_gen}.
\end{proof}

\section{Compatibility of normalization conditions}\label{Sec:comp}

In this section we show the compatibility of the condition \ref{norm:wAL} with the standard normalization conditions usually imposed on the time-ordered products. In the first subsection we show that in purely massless models the condition \ref{norm:wAL} is implied by almost homogeneous scaling of VEVs of time-ordered products. In Section \ref{sec:central} we formulate the central normalization condition \ref{norm:gc} which is stronger than \ref{norm:wAL} and, in particular, fixes uniquely all time-ordered products of the sub-polynomials of the interaction vertex in the case of QED. Section \ref{sec:std_norm_con} contains the proof of the compatibility of the standard normalization conditions with \ref{norm:gc}, and thus with \ref{norm:wAL}. In Section~\ref{sec:properties} we list the properties of the Wightman and Green functions.

\subsection{Almost homogeneous scaling}\label{sec:hom_sc}

In this section we assume that all fields under consideration are massless and set $\CC=0$ in the axiom \ref{axiom7}. Following~\cite{hollands2001local,hollands2002existence} we introduce the definition of the almost homogeneous scaling of a~distribution.
\begin{dfn}\label{def:aHS}
A distribution $t\in \cS'(\R^N)$ scales almost homogeneously with degree $D\in\R$ and power $P\in\N_0$ iff
\begin{equation}\label{eq:a_h_d}
(\sum_{j=1}^N x_j\partial_{x_j}+D)^{P+1}t(x)=0
\end{equation}
and $P$ is the minimal natural number with this property. If $P=0$ the above condition states that the distribution $t$ is homogeneous of degree $-D$.
\end{dfn}

The condition \eqref{eq:a_h_d} is equivalent to
\begin{equation}
 (\rho\partial_\rho)^{P+1}\left(\rho^D t(\rho x)\right)=0.
\end{equation}
Thus, $t$ scales almost homogeneously with degree $D$ and power $P$ iff $\rho^D t(\rho x)$ is a~polynomial of $\log\rho$ with degree $P$. This implies that in particular $\sd(t)=D$, where $\sd(\cdot)$ is the Steinmann scaling degree (cf. Definition \ref{def:sd}). Moreover,
\begin{equation}
 (\rho\partial_\rho)^{P+1}\left(\rho^{D-N} \F{t}(q/\rho)\right)=0.
\end{equation}
Hence, it holds $\sd(\F{t})=N-D$.

If the axiom \ref{axiom7} holds with $\CC=0$, then the time-ordered products of polynomials of massless fields can be normalized such that the following condition holds.
\begin{enumerate}[label=\bf{N.aHS},leftmargin=*]
\item\label{norm:sc} Almost homogeneous scaling:
For all polynomials $B_1,\ldots,B_k\in\Fh$ of massless fields and their derivatives the distribution
\begin{equation}
 (\Omega|\T(B_1(x_1),\ldots,B_k(x_k))\Omega)
\end{equation}
scales almost homogeneously with degree 
\begin{equation}
 D = \omega + 4(k-1) = \sum_{j=1}^{k}\dim(B_j).
\end{equation}
\end{enumerate}
The above normalization condition is a~special case of the condition called {\it Scaling} in~\cite{dutsch2004causal}, which was imposed also for the time-ordered products of massive fields as a~substitute of \ref{axiom7}. It is clear that in the case of purely massless models \ref{norm:sc} is stronger than \ref{axiom7} since for any $t\in\cS'(\R^N)$ which scales almost homogeneously with degree $D$ it holds $\sd(t)=D$. The condition similar to \ref{norm:sc} was introduced for the first time in~\cite{hollands2001local,hollands2002existence} in the context of QFT in curved spacetime. For massless fields this method of normalization was used, for example, in~\cite{grigore2001scale,gracia2003improved,lazzarini2003improved}.

\begin{thm}
If all fields under consideration are massless and $\CC=0$ in \ref{axiom7}, then the normalization condition \ref{norm:sc} implies \ref{norm:wAL}.
\end{thm}
\begin{proof}
By the comment below Definition \ref{def:aHS} the following equality
\begin{equation}
 \sd\left( (\Omega|\T(\F{B}_1(q_1),\ldots,\F{B}_k(q_k))\Omega) \right) = \sum_{j=1}^{k} (4 -\dim(B_j))
\end{equation}
follows from the normalization condition \ref{norm:sc}. As a~result, we obtain
\begin{equation}
 \sd\big( t(q_1,\ldots,q_{k-1}) \big)  = \sum_{j=1}^{k} (4 -\dim(B_j))-4 = -\omega',
\end{equation}
where $\omega'$ is given by \eqref{eq:proof_omega_prime} and $t\in\cS'(\R^{4(k-1)})$ is such that
\begin{equation}
 (\Omega|\T(\F{B}_1(q_1),\ldots,\F{B}_k(q_k))\Omega) = (2\pi)^4 \delta(q_1+\ldots+q_k)\,t(q_1,\ldots,q_{k-1}).
\end{equation}
Thus, for any $g\in\cS(\R^{4(k-1)})$ and multi-index $\gamma$, $|\gamma|<\omega$ it holds
\begin{equation}
 \lim_{\epsilon\searrow 0} \int\mP{q_1}\ldots\mP{q_{k-1}} \, \partial^\gamma_qt(q_1,\ldots,q_{k-1})~ g_\epsilon(q_1,\ldots,q_{k-1}) = 0, 
\end{equation}
where $g_\epsilon$ is defined in terms $g$ as in Definition \ref{def:lojasiewicz} of the value of a~distribution at a~point in the sense of {\L}ojasiewicz. This implies the normalization condition \ref{norm:wAL}.
\end{proof}

\subsection{Central normalization condition}\label{sec:central}

In this section we formulate the normalization condition for the time-ordered products which is a generalization of a condition introduced by Epstein and Glaser in \cite{epstein1973role} for purely massive models. Let us describe the content of the latter condition. It involves the VEVs of the advanced products of the form 
\begin{equation}
 (\Omega|\Adv(B_1(x_1),\ldots,B_n(x_n);B_{n+1}(0))\Omega)
\end{equation}
and says that their Fourier transforms (which in the case of purely massive models is an analytic function in some neighborhood of the origin) have zero of order $\omega+1$ at the origin (the {\it central} and the most {\it symmetrical} point), where $\omega$ is given by \eqref{eq:omega}. The advanced products which fulfill the above condition are known in the literature as the {\it central} or {\it symmetrical} solutions of the splitting problem, or sometimes the {\it central} or {\it symmetrical} extensions.\footnote{Finding a solution of {\it the splitting problem}, which is {\it an extension} of a certain distribution initially defined outside the origin, is a part of the EG construction of the time-ordered products.} This condition fixes uniquely all the time-ordered products and is compatible with the standard normalization conditions listed in the next section. So far, the above normalization condition has been formulated rigorously only for theories without massless fields~\cite{epstein1973role}. However, an expectation that a condition of a similar type can also be imposed in the massive spinor QED has been expressed by some authors \cite{scharf2014,dutsch1990gauge,dutsch1996finite,hurth1995nonabelian}. 

Using the method of the proof of Theorem \ref{thm:main1} one can show \cite{duch2017massless} that it is possible to define the time-ordered products such that the following condition holds. It may be viewed as a generalization of the condition which was discussed above.
\begin{enumerate}[label=\bf{N.C},leftmargin=*]
\item\label{norm:gc}
Let $k\in\N_+$ and $A^{r_1},\ldots,A^{r_k}\in\Fa$ be arbitrary monomials. Consider the following distribution
\begin{equation}\label{eq:cent_dist}
 (\Omega|\T(A^{r_1}(x_1),\ldots,A^{r_k}(x_k))\Omega).
\end{equation}
\begin{enumerate}[label=(\arabic*),leftmargin=*]
\item The distribution \eqref{eq:cent_dist} has \underline{IR}-index  
\begin{equation}\label{eq:gc_index1}
  d_1 = 4 + \sum_{j=1}^{k} \dim(A^{r_j}) - 4k.
\end{equation}
\item If all of the super-quadri-indices $r_1,\ldots,r_k$ involve at least one massive field, then the distribution \eqref{eq:cent_dist} has \underline{IR}-index
\begin{equation}\label{eq:gc_index2}
 d_2=5+\sum_{j=1}^{k} (\dim(A^{r_j})+\CC) - 4k.
\end{equation}
\item 
If all but one of the super-quadri-indices $r_1,\ldots,r_k$ involve at least one massive field, then the distribution \eqref{eq:cent_dist} has \underline{IR}-index
\begin{equation}\label{eq:gc_index3}
 d_3=5+\sum_{j=1}^{k} \dim(A^{r_j}) -4k.
\end{equation}
\end{enumerate}
\end{enumerate}
The above condition significantly restricts the normalization freedom of the time-ordered products. Its usefulness lies in the fact that it respects many symmetries, e.g. the Poincar{\'e} symmetry (cf. Definition \ref{def:IR1} of the \underline{IR}-index). Part~(2) of \ref{norm:gc} fixes uniquely the time-ordered products to which it refers. Parts (1) and (2) of \ref{norm:gc} imply the normalization condition \ref{norm:wAL}. In the case of monomials which are products of massless fields the condition \ref{norm:gc} is equivalent to the normalization condition \ref{norm:sc}. 

Let us consider the application of the condition \ref{norm:gc} in QED. We recall that the interaction vertex in QED is $\cL=\overline{\psi}\slashed{A}\psi$, where $\psi$ is a massive Dirac spinor field and $A_\mu$ is a massless vector field. Since $\dim(\cL)=4$ we set $\CC=0$ in the axiom \ref{axiom7}. The time-ordered products satisfying the condition \ref{norm:gc} have the following property.
\begin{enumerate}[label=\bf{N.C${}_\textrm{QED}$},leftmargin=*]
 \item\label{norm:qed_spinor} The distribution
\begin{equation}\label{eq:norm_c_qed_dist}
 (\Omega|\T(B_1(x_1),\ldots,B_{n+1}(x_{n+1}))\Omega) 
\end{equation} 
has \underline{IR}-index 
\begin{equation}\label{eq:norm_qed_d}
 d  = 1 + \sum_{j=1}^{n+1} \dim(B_j) - 4n 
\end{equation}
for any $n\in\N_+$ and any sub-polynomials $B_1,\ldots,B_n$ of of the interaction vertex $\cL$ of QED such that $d \geq 0$.
\end{enumerate}
By Lemma \ref{lem:IR_der} the above normalization condition implies that the Fourier transform of the distribution \eqref{eq:norm_c_qed_dist} is of the form 
\begin{equation}\label{eq:qed_form}
 (2\pi)^4\delta(q_1+\ldots+q_{n+1}) \,t(q_1,\ldots,q_n),
\end{equation}
where $t(q_1,\ldots,q_n)$ has zero of order $d$ at $q_1=\ldots=q_n=0$ in the sense of {\L}ojasiewicz. The Fourier transform of the VEV of the corresponding advanced product 
\begin{equation}\label{eq:qed_adv}
 (\Omega|\Adv(B_1(x_1),\ldots,B_n(x_n),B_{n+1}(x_{n+1}))\Omega)
\end{equation}
has the same form. Consequently, the advanced product defined in terms of the time-ordered products satisfying \ref{norm:qed_spinor} may be interpreted as the central splitting solution.

It follows from the results of Section \ref{sec:freedom} that \ref{norm:qed_spinor} fixes uniquely all the time-ordered products of the sub-polynomials of the interaction vertex $\cL$. As we will show in the next section this condition implies all the standard normalization conditions: \ref{norm:u}, \ref{norm:p}, \ref{norm:pct}, \ref{norm:one} and \ref{norm:ward} for the time-ordered products of the sub-polynomials of the interaction vertex $\cL$ of QED. 

\subsection{Standard normalization conditions}\label{sec:std_norm_con}

In this section we list the standard normalization conditions which are usually imposed on the time-ordered products and argue that they are compatible with the normalization condition \ref{norm:gc}. 

\begin{enumerate}[label=\bf{N.U},leftmargin=*]
\item\label{norm:u} Unitarity:
\begin{equation}\label{eq:n_U}
 \T(B_1(x_1),\ldots,B_k(x_k))^* = \aT(B_k^*(x_k),\ldots,B_1^*(x_1))
\end{equation}
for all $B_1,\ldots,B_k\in\Fa$. The definitions of the adjoints in $\Fa$ and $L(\cD_0)$ are given in Section~\ref{sec:ff} and \ref{sec:Wick}, respectively.
\end{enumerate}
\begin{enumerate}[label=\bf{N.P},leftmargin=*]
\item\label{norm:p} Poincar\'e covariance: 
\begin{multline}\label{eq:n_P}
 U(a,\Lambda)\T(B_1(x_1),\ldots,B_k(x_k))U(a,\Lambda)^{-1} 
 \\
 = \T((\rho(\Lambda)B_1)(\Lambda x_1+a),\ldots,(\rho(\Lambda)B_k)(\Lambda x_k+a)),
\end{multline}
where $B_1,\ldots,B_k\in\Fa$, $\rho$ is the representation of $SL(2,\C)$ acting on $\mathcal{F}$ and $U$ is the unitary representation of the Poincar{\'e} group on $\cD_0$.
\end{enumerate}
\begin{enumerate}[label=\bf{N.CPT},leftmargin=*]
\item\label{norm:pct} Covariance with respect to the discrete group of CPT transformations (the charge conjugation, the spatial inversion and the time reversal). 
\end{enumerate}

\begin{enumerate}[label=\bf{N.FE},leftmargin=*]
\item\label{norm:one} Let $B_1,\ldots,B_k\in\Fh$ be sub-polynomials of the interaction vertex, and let $A_i\in\mathcal{G}_0$ be a basic generator. It holds: 
\begin{multline}\label{eq:one}
 \hspace{-2mm}(\Omega|\T(A_i(x),B_1(x_1),\ldots,B_k(x_k))\Omega)
 \!= \!
 \sum_{j=1}^k\! \sum_{C\in\mathcal{G}} (-1)^{\mathbf{f}(C)(\mathbf{f}(B_1)+\ldots+\mathbf{f}(B_{j-1}))} \!
 \\
 \times(\Omega|\T(A_i(x),C(x_j))\Omega)
 ~(\Omega|\T\big(B_1(x_1),\ldots,\frac{\partial B_j}{\partial C}(x_j),\ldots,B_k(x_k)\big)\Omega),
\end{multline}
where the second sum is over generators $C\in\mathcal{G}$. This condition says that the time-ordered products with a basic generator among its arguments are uniquely determined by the time-order products with less arguments. It also implies that the interacting fields satisfy the field equations.
\end{enumerate}

\begin{enumerate}[label=\bf{N.W},leftmargin=*]
\item\label{norm:ward} Ward identities in QED~\cite{dutsch1999local}: For any $B_1,\ldots,B_{k}\in\Fh$ which are sub-polynomials of the interaction vertex $\mathcal{L}$ it holds
\begin{multline}\label{eq:n_W}
 \partial^x_{\mu}\T(j^\mu(x),B_1(x_1),\ldots,B_{k}(x_{k}))
 \\
 = \ri \sum_{j=1}^{k} \mathbf{q}(B_j) \,\delta(x_j-x)\,\T(B_1(x_1),\ldots,B_{k}(x_{k})),
\end{multline}
where $\mathbf{q}(B)$ is the charge number (cf. Section \ref{sec:ff}) and $j^\mu=\overline{\psi}\gamma^\mu\psi$ is the free electric current.
\end{enumerate}

For the proof that \ref{norm:u}, \ref{norm:p}, \ref{norm:pct}, \ref{norm:one}, \ref{norm:ward} (the last condition holds only in QED) can be simultaneously imposed see, for example, ~\cite{boas2000gauge,scharf2014}. The compatibility of \ref{norm:u}, \ref{norm:p}, \ref{norm:pct} with \ref{norm:gc} follows from the invariance of \ref{norm:gc} under Poincar{\'e} and CPT transformations and the fact that the condition \ref{norm:gc} implies the analogous condition for the VEVs of the anti-time-ordered products. 

The compatibility of \ref{norm:one} and \ref{norm:gc} is less straightforward. We will outline its proof in the case of QED. Assume that the time-ordered products with at most $n$ arguments satisfy \ref{norm:one} and \ref{norm:gc}. Define the time-ordered products with $n+1$ arguments such \ref{norm:gc} holds. We have
\begin{multline}\label{eq:one_A}
(\Omega|\T(B_1(x_1),\ldots,B_n(x_n),A_\mu(x))\Omega)
 =  \sum_{\substack{\gamma\\|\gamma|\leq d}} c_\gamma \partial^\gamma \delta(x_1-x)\ldots \delta(x_{n}-x) 
 \\
 + \ri \sum_{k=1}^n D_0^F(x_k-x)~
 (\Omega|\T\left(B_1(x_1),\ldots,\frac{\partial B_k}{\partial A^\mu}(x_k),\ldots,B_n(x_n) \right) \Omega),
\end{multline}
where $D_0^F(x)$ is the massless Feynman propagator, $c_\gamma\in\C$ are some constants, and
\begin{equation}
  d = 1+\sum_{j=1}^n\dim(B_j)-4n.
\end{equation}
Equation \eqref{eq:one_A} implies that
\begin{multline}\label{eq:one_A2}
 \hspace{-1mm}\square_x\,(\Omega|\T(B_1(x_1),\ldots,B_n(x_n),A_\mu(x))\Omega)
 =  \hspace{-0.5mm} \sum_{\substack{\gamma\\|\gamma|\leq d}}\hspace{-1mm} c_\gamma \partial^\gamma\square_x \delta(x_1-x)\ldots \delta(x_{n}-x) 
 \\
  -\ri \sum_{k=1}^n  \delta(x_k-x)~
 (\Omega|\T\left(B_1(x_1),\ldots,\frac{\partial B_k}{\partial A^\mu}(x_k),\ldots,B_n(x_n) \right) \Omega).
\end{multline}
Assume that $B_1,\ldots,B_n$ are sub-polynomials of the interaction vertex and $ d\geq 0$. The condition \ref{norm:gc} implies that the distribution in the last line of the above equation has \underline{IR}-index $ d+3$. Since the distribution on the LHS of \eqref{eq:one_A} has \underline{IR}-index $ d+1$, the distribution on the LHS of \eqref{eq:one_A2} has \underline{IR}-index $ d+3$. This implies that $c_\gamma=0$ and \ref{norm:one} holds if the basic generator $A_i$ is the vector potential $A_\mu$. To prove that \ref{norm:one} is satisfied when $A_i$ is the spinor field $\psi_a$ or $\overline{\psi}_a$ it is enough to use the fact that the Feynman propagator of a massive particle is smooth in the vicinity of zero.

Now let us suppose that the time-ordered products with at most $n$ arguments satisfy \ref{norm:ward} and \ref{norm:gc}. Define the time-ordered products with $n+1$ arguments such that \ref{norm:gc} holds. We will show that \ref{norm:ward} is satisfied for $k=n$. By the results of Appendix B of~\cite{dutsch1999local} it is enough to prove that
\begin{multline}\label{eq:proof_ward}
 \partial^x_{\mu}(\Omega|\T(j^\mu(x),B_1(x_1),\ldots,B_{n}(x_{n}))\Omega)
 \\
 -  \ri \sum_{j=1}^{n} \mathbf{q}(B_j) \,\delta(x_j-x)\,(\Omega|\T(B_1(x_1),\ldots,B_{n}(x_{n}))\Omega)
\end{multline}
vanishes if $B_1,\ldots,B_n$ are sub-polynomials of the interaction vertex of QED. We also know that the above distribution is of the form
\begin{equation}
 \sum_{\substack{\gamma\\|\gamma|\leq d+1}} c_\gamma \partial^\gamma \delta(x_1-x)\ldots \delta(x_{n}-x),
\end{equation}
where 
\begin{equation}
  d =3+\sum_{j=1}^n \dim(B_j) -4n.
\end{equation}
If $ d+1\geq 0$, then the constants $c_\gamma$ must be zero since by \ref{norm:gc} the distribution \eqref{eq:proof_ward} has \underline{IR}-index $ d+2$.
 
\subsection{Properties of Wightman and Green functions}\label{sec:properties}

If the time-ordered products satisfy the normalization conditions formulated in Section \ref{sec:std_norm_con}, then the Wightman and Green functions have a~number of important properties. In fact, the Wightman functions fulfill almost all the standard axioms, which are listed in~\cite{streater2000pct}, in the sense of formal power series in the coupling constants. It is only not clear whether the cluster decomposition property holds. Let $C_1,\ldots,C_m\in\Fa$ be arbitrary polynomials. The following conditions are satisfied.
\begin{enumerate}[leftmargin=*,label={(\arabic*)}]
\item Poincar{\'e} covariance:
\begin{equation}
 \Wig(C_1(x_1),\ldots,C_m(x_m))=\Wig((\rho(\Lambda)C_1)(\Lambda x_1+a),\ldots,(\rho(\Lambda)C_m)(\Lambda x_k+a)),
\end{equation}
where $\rho$ is the representation of $SL(2,\C)$ acting on $\Fa$. This property follows from the normalization condition \ref{norm:p}, Part (B) of Theorem \ref{thm:main2} and the fact that if $t\in \cS'(\R^{4m})$, $t(q_1,\ldots,q_m)=c+O^{\textrm{dist}}(|(q_1,\ldots,q_m)|^\delta)$ for some $c\in\C$ and $\delta>0$, then \begin{equation}
t(\Lambda q_1,\ldots,\Lambda q_n) \exp(\ri (\Lambda q_1 + \ldots + \Lambda q_n)\cdot a) = c+O^{\textrm{dist}}(|(q_1,\ldots,q_m)|^\delta).
\end{equation}
\item Spectrum condition: The Fourier transform of the Wightman function
\begin{equation}
 \Wig(\F{C}_1(p_1),\ldots,\F{C}_m(p_m))
\end{equation}
has the support contained in 
\begin{equation}
 \left\{(p_1,\ldots,p_m)\in\R^{4m}\,:\, \sum_{j=1}^m p_j =0, ~~
 \forall_{k}~\sum_{j=1}^k p_j \in \overline{V}^+ \right\}.
\end{equation}
For the proof of this property see \cite{epstein1973role} (p. 267).
\item Hermiticity:
\begin{equation}
 \overline{\Wig(C_1(x_1),\ldots,C_m(x_m))}=\Wig(C^*_m(x_m),\ldots,C^*_1(x_1)).
\end{equation}
The definition of the adjoint in $\Fa$ is given in Section~\ref{sec:ff}. In the proof of the above property we use the identity $C_\adv(g;h)^* =C^*_\adv(\bar g;\bar h)$, which follows from the normalization condition \ref{norm:u}.
\item Local (anti-)commutativity: For $C_k,C_{k+1}\in \Fh$ it holds
\begin{multline}
 \Wig(\ldots,C_k(x_k),C_{k+1}(x_{k+1}),\ldots)
 \\
 = (-1)^{\mathbf{f}(C_k)\mathbf{f}(C_{k+1})} \Wig(\ldots,C_{k+1}(x_{k+1}),C_k(x_k),\ldots)
\end{multline}
if $x_k$ and $x_{k+1}$ are spatially separated, i.e. $(x_k-x_{k+1})^2\leq 0$.
\item Positive definiteness condition: Let $j_0\in\N_+$, $\{1,\ldots,j_0\} \ni j\mapsto f_j \in S(\R^{4n_j})$ and $C_{j,1},\ldots,C_{j,n_j}\in\mathcal{F}$, where $n_j\in\N_+$. Then the formal power series  
\begin{multline}
 \sum_{j,k\in\N_0}\int\rd^4 x_1,\ldots\rd^4 x_{n_j}\rd^4 y_1,\ldots\rd^4 y_{n_k}\, \overline{f_j(x_1,\ldots,x_{n_j})}f_k(y_1,\ldots,y_{n_k})\,
 \\
 \Wig(C^*_{j,1}(x_1),\ldots,C^*_{j,n_j}(x_{n_j}),C_{k,1}(y_1),\ldots,C_{k,n_k}(y_{n_k}))
\end{multline}
is nonnegative in the sense of Definition~\ref{dfn:positive_formal} stated in Section~\ref{sec:poincare_state}. The above condition is satisfied in models which can be defined on the Fock space with a~positive-definite inner product. It follows from Theorem~\ref{thm:state_positive} formulated in Section~\ref{sec:poincare_state}.
\item Field equations: For example in the massless $\varphi^4$ theory it holds
\begin{multline}
 \square_x \Wig(C_1(x_1),\ldots,\varphi(x),\ldots,C_m(x_m)) 
 \\
 + \frac{\lambda}{3!} \Wig(C_1(x_1),\ldots,\varphi^3(x),\ldots,C_m(x_m)) = 0.
\end{multline}
Similar equations are satisfied in other models. This condition is a direct consequence of the normalization condition \ref{norm:one}. Note that it is not included in the list of the standard axioms~\cite{streater2000pct}.
\end{enumerate}

\noindent Now let us state the properties of the Green functions.
\begin{enumerate}[leftmargin=*,label={(\arabic*)}]
\item Poincar{\'e} covariance:
\begin{equation}
 \Gre(C_1(x_1),\ldots,C_m(x_m))=\Gre((\rho(\Lambda)C_1)(\Lambda x_1+a),\ldots,(\rho(\Lambda)C_m)(\Lambda x_m+a)),
\end{equation}
where $\rho$ is the representation of $SL(2,\C)$ acting on $\Fa$.

\item Graded-symmetry:
\begin{equation}
 \Gre(C_1(x_1),\ldots,C_m(x_m)) = (-1)^{\mathbf{f}(\pi)}\Gre(C_{\pi(1)}(x_{\pi(1)}),\ldots,C_{\pi(m)}(x_{\pi(m)})).
\end{equation}
for all $C_1,\ldots,C_m\in\Fh$, where $\mathbf{f}(\pi)\in\Z/2\Z$ is the number of transpositions in the permutation $\pi\in\mathcal{P}_m$ that involve a~pair of fields with odd fermion number.
\item Causality:
\begin{equation}
 \Gre(C_1(x_1),\ldots,C_m(x_m)) = \Wig(C_1(x_1),\ldots,C_m(x_m))
\end{equation}
if for all $j\in\{1,\ldots,m-1\}$ the point $x_{j}$ is not in the causal past of any of the points $x_{j+1},\ldots,x_m$.
\end{enumerate}

\section{Vacuum state}\label{sec:vacuum_state}

In this section we consider the framework of perturbative algebraic quantum field theory --- the formulation of perturbative QFT in terms of local abstract algebras \cite{fredenhagen2015perturbative,rejzner2016perturbative}. In the first subsection we introduce the notation and remind the reader the construction of the net of local abstract algebras of interacting fields \cite{brunetti2000microlocal}, known in the literature under the name of the algebraic adiabatic limit. In Section \ref{sec:poincare_state} we define a Poincar{\'e}-invariant functional on the algebra of interacting fields. In the case of models which are constructed on the Fock space with a positive-definite covariant inner product this functional is positive. As a result, it is a state which can be interpreted as an interacting vacuum state.

\subsection{Algebraic adiabatic limit}\label{sec:algebraic_adiabatic}

The construction of the local abstract algebras of interacting operators is based on the following observation made in \cite{brunetti2000microlocal}. For simplicity, we assume that there is only one interaction vertex $\mathcal{L}$. The switching function is denoted by $g$, the coupling constant -- by $e$ and the interacting operators by $C_\adv(g;h)=C_\adv(\mathbf{g};h)$, where $\mathbf{g}=eg\otimes \mathcal{L}$. Let $\mathcal{O}$ be an open bounded subset of $\R^4$ and $g,g'\in\mathcal{D}(\R^4)$ coincide in a neighborhood of $J^+(\mathcal{O})\cap J^-(\mathcal{O})$, where $J^\pm(\mathcal{O}):=\mathcal{O}+\overline{V}^\pm$ is the causal future/past of the set $\mathcal{O}$. Then, as a result of the causal factorization formula \eqref{eq:causal_fact}, there exists a unitary operator $V(g',g)$ such that
\begin{equation}
 V(g',g) C_\adv(g;h) V(g',g)^{-1} = C_\adv(g';h)
\end{equation}
for any $C\in\Fa$ and $h\in\mathcal{D}(\R^4)$, $\supp\,h\subset\mathcal{O}$. Let us introduce the following denotation
\begin{equation}
 \mathcal{D}_\mathcal{O} := \{ g\in \mathcal{D}(\R^4) \,:\, g \equiv 1 \textrm{ on a neighborhood of } J^+(\mathcal{O})\cap J^-(\mathcal{O})\}
\end{equation}
for bounded open sets $\mathcal{O}\subset\R^4$ and define for $h\in\mathcal{D}(\R^4)$ such that $\supp\, h\subset \mathcal{O}$ the function
\begin{equation}\label{eq:algebraic_adiabatic_int_field}
 C_\adv(\cdot;h): \mathcal{D}_\mathcal{O} \ni g \mapsto C_\adv(g;h)\in L(\mathcal{D}_0)\llbracket e\rrbracket,
\end{equation}
where $L(\mathcal{D}_0)\llbracket e\rrbracket$ is the space of formal power series in $e$ with coefficients in $L(\mathcal{D}_0)$. We define the addition, multiplication and conjugation of these maps by the pointwise operations. The local algebra of interacting fields localized in $\mathcal{O}$ denoted by $\mathfrak{F}(\mathcal{O})$ is by definition the ${}^*$-algebra over $\C\llbracket e\rrbracket$ with unity $\id$ generated by $C_\adv(\cdot;h)$ with $h\in\mathcal{D}(\R^4)$, $\supp \,h\subset \mathcal{O}$ and $C\in\Fa$. For $\mathcal{O}'\subset\mathcal{O}$ we define the embedding
\begin{equation}
 \iota_{\mathcal{O}\mathcal{O}'}:~\mathfrak{F}(\mathcal{O}')\rightarrow \mathfrak{F}(\mathcal{O})
\end{equation}
by the restriction of the function \eqref{eq:algebraic_adiabatic_int_field} to $\mathcal{D}_{\mathcal{O}}\subset\mathcal{D}_{\mathcal{O}'}$. We have 
\begin{equation}
 \iota_{\mathcal{O}\mathcal{O}'}\circ\iota_{\mathcal{O}'\mathcal{O}''}= \iota_{\mathcal{O}\mathcal{O}''}.
\end{equation}
The global algebra of interacting fields $\mathfrak{F}$ is the inductive limit of the net $\mathcal{O}\to\mathfrak{F}(\mathcal{O})$ with canonical embeddings
\begin{equation}
 \iota_\mathcal{O}: \mathfrak{F}(\mathcal{O}) \to \mathfrak{F}.
\end{equation}
The action of the Poincar{\'e} transformations on $\mathfrak{F}$ is given by the automorphisms defined in the following way on the generators of $\mathfrak{F}$
\begin{equation}
 \alpha_{a,\Lambda}(C_\adv(\cdot;h)) = (\rho(\Lambda^{-1})C)_\adv(\cdot;h_{a,\Lambda}),
\end{equation}
where $h_{a,\Lambda}(x) = h(\Lambda^{-1}(x-a))$ and $\rho$ is the representation of the Lorentz group acting on $\Fa$ introduced in Section~\ref{sec:ff}.

As shown in \cite{dutsch2001algebraic,fredenhagen2015perturbative} the net $\mathcal{O}\to\mathfrak{F}(\mathcal{O})$ constructed above fulfills the axioms: (1)~isotony, (2) Poincar{\'e} covariance and (3) Einstein causality (see \cite{haag2012local} for their definition) in the sense of formal power series. The construction of the algebra $\mathfrak{F}$ and the net $\mathcal{O}\to\mathfrak{F}(\mathcal{O})$ outlined above is called in the literature the algebraic adiabatic limit.

In the case of theories with local symmetries the algebra $\mathfrak{F}$ plays only an auxiliary role. From physical point of view more important is the algebra of observables, which is a certain quotient algebra consisting of equivalence classes of gauge-invariant fields. In the case of QED the construction of the algebra of observables was carried out by D{\"u}tsch and Fredenhagen in \cite{dutsch1999local} and in the case of non-abelian Yang--Mills theories without matter -- by Hollands in \cite{hollands2008renormalized}.

\subsection{Poincar{\'e}-invariant state}\label{sec:poincare_state}

Elements of $\mathfrak{F}$ will be denoted by ${\bf B}(\cdot)$. They are linear combinations of products of $C_\adv(\cdot;h)$ for some $C\in\Fa$ and $h\in\cD(\R^4)$. 
\begin{dfn}
A linear functional $\sigma:\,\mathfrak{F}\to\C\llbracket e\rrbracket$ is called a state if it is normalized $\sigma(\id)=1$, real $\sigma(\mathbf{B}(\cdot)^*)=\overline{\sigma(\mathbf{B}(\cdot))}$ and positive $\sigma(\mathbf{B}(\cdot)^*\, \mathbf{B}(\cdot))\geq 0$ for all $\mathbf{B}(\cdot)\in\mathfrak{F}$.
\end{dfn}
\begin{dfn}[\cite{dutsch1999local}]\label{dfn:positive_formal}
A formal power series $a\in\C\llbracket e\rrbracket$ is nonnegative iff there exists $b\in\C\llbracket e\rrbracket$ such that $a=\overline{b} b$. Equivalently, $a\in\R\llbracket e\rrbracket$ and the first non-vanishing coefficient of $a$ (if exists) is of even order and positive.
\end{dfn}
It turns out that the existence of the weak adiabatic limit allows to define a real, normalized and Poincar{\'e} invariant functional on the algebra of interacting fields $\mathfrak{F}$.
\begin{thm}
Let us assume that a model under consideration is purely massive or satisfies Assumption~\ref{asm} formulated in Section~\ref{sec:PROOF}. Suppose that the time-ordered products satisfy the standard normalization conditions listed in Section \ref{sec:std_norm_con} and the normalization condition \ref{norm:wAL} (the last condition is required only in the case of models with massless fields). The unique linear functional $\sigma:\mathfrak{F}\to\C\llbracket e\rrbracket$ given by
\begin{equation}\label{eq:state}
 \sigma(C_{1,\adv}(\cdot;h_1)\ldots C_{n,\adv}(\cdot;h_n)) = \lim_{\epsilon\searrow0} \, (\Omega|C_{1,\adv}(g_\epsilon;h_1)\ldots C_{n,\adv}(g_\epsilon;h_n) \Omega)
\end{equation}
for any $n\in\N_0$, $h_1,\ldots,h_n\in\mathcal{D}(\R^4)$ is real, normalized and Poincar{\'e} invariant.
\end{thm}
\begin{proof}
Let $\mathcal{O}\subset\R^4$ be a~bounded region such that $\supp\,h_1,\ldots,\supp\,h_n\subset\mathcal{O}$.  For any $\mathcal{O}$ and $g\in\mathcal{D}(\mathcal{O})$ we have $g_\epsilon\in\mathcal{D}(\mathcal{O})$ for all $\epsilon\in(0,1)$, where $g_\epsilon(x):=g(\epsilon x)$. By the result of Section \ref{sec:weak_massive_proof} or Theorem~\ref{thm:main2} the weak adiabatic limit exists. This, in particular, implies that the limit on the RHS of Equation \eqref{eq:state} exists and is independent of $g$. As a result, the functional $\sigma$ is well defined. It is evident that $\sigma$ is normalized. Hermiticity of the Wightman functions implies that the functional $\sigma$ is real. Poincar{\'e} invariance of $\sigma$,
\begin{equation}
 \sigma(\,\alpha_{a,\Lambda}(\,\mathbf{B}(\cdot)\,)\,) = \sigma(\,\mathbf{B}(\cdot)\,)~~~\textrm{for all}~~ \mathbf{B}\in\mathfrak{F},
\end{equation}
is a consequence of Poincar{\'e} covariance of the Wightman functions. 
\end{proof}

\begin{thm}\label{thm:state_positive}
In the case of models defined on the Fock space with a positive-definite covariant inner product the functional $\sigma$ given by \eqref{eq:state} is positive (hence, it is a state). 
\end{thm}
\begin{proof}
Let $\mathcal{O}$ be a~bounded open set in $\R^4$, ${\bf B}(\cdot)\in\mathfrak{F}(\mathcal{O})$, $g\in\mathcal{D}_\mathcal{O}$ and $g_\epsilon(x):=g(\epsilon x)$. We have to show that the formal power series
\begin{equation}
 \lim_{\epsilon\searrow 0}\,(\Omega|{\bf B}(g_\epsilon)^*\, {\bf B}(g_\epsilon)  \Omega) 
\end{equation}
is nonnegative in the sense of Definition~\ref{dfn:positive_formal}. First note that the coefficients of this series are real-valued. The series is clearly nonnegative if ${\bf B}(\cdot)=0$. Assume that ${\bf B}(\cdot)\neq0$ and denote by $B(\cdot)$ the first non-vanishing coefficient of the expansion of ${\bf B}(\cdot)$ in powers of $e$. It follows from Lemma~\ref{lem:adiabatic_first} that
\begin{equation}
 \lim_{\epsilon\searrow 0}\,B(g_\epsilon) = B(g).
\end{equation}
The operator $B(g) \in L(\mathcal{D}_0)$ is a nonzero local operator, i.e. a Wick polynomial (but not necessarily a Wick polynomial at a single point) smeared with some test function of compact support. Since nonzero local operators do not annihilate the vacuum and the covariant inner product on the Fock space is positive definite \cite{streater2000pct}, we have 
\begin{equation}
 (\Omega|B(g)^* B(g)  \Omega) > 0,
\end{equation}
which finishes the proof.
\end{proof}
\begin{lem}\label{lem:adiabatic_first}
Let $\mathcal{O}$ be a~bounded open set in $\R^4$ and ${\bf B}(\cdot)\in\mathfrak{F}(\mathcal{O})$. 
\begin{enumerate}[leftmargin=*,label={(\Alph*)}]
 \item If ${\bf B}(g)=O(e^n)$ for some $g\in\mathcal{D}_\mathcal{O}$, then ${\bf B}(g')=O(e^n)$ for all $g'\in\mathcal{D}_\mathcal{O}$.
 \item Assume that ${\bf B}(\cdot)\neq 0$ and let $B(\cdot)$ be the first non-vanishing coefficient in the expansion of ${\bf B}(\cdot)$ in powers of $e$. Then $B(g)=B(g')$ for all $g,g'\in\mathcal{D}_\mathcal{O}$.
\end{enumerate}
\end{lem}
\noindent The above lemma follows immediately from the following fact.
\begin{lem}
Let $\mathcal{O}$ be a~bounded region in $\R^4$ and ${\bf B}(\cdot)\in\mathfrak{F}(\mathcal{O})$. For any $g,g'\in\mathcal{D}_\mathcal{O}$ it holds
\begin{equation}
 V(g',g) {\bf B}(g) V(g',g)^{-1} = {\bf B}(g'),
\end{equation} 
where 
\begin{equation}
 V(g',g) = \id + O(e).
\end{equation}
\end{lem}
In the case of models with vector fields such as QED the covariant inner product on $\mathcal{D}_0$ is indefinite and the functional $\sigma$ given by \eqref{eq:state} is not positive. Nevertheless, it is expected that it can be used to define a Poincar\'e-invariant state on the algebra of observables.

\section{Summary and outlook}

We proved the existence of the weak adiabatic limit in the Epstein--Glaser approach to perturbative QFT in a large class of models. The result implies the existence of the Wightman and Green functions, which have most of the standard properties following from the Wightman axioms. It can also be used to construct a real, normalized and Poincar\'e-invariant functional on the algebra of interacting fields obtained by means of the algebraic adiabatic limit. In the case of models which are defined on the Fock space with a~positive-definite covariant inner product we proved that this functional is positive in the sense of formal power series. Consequently, it is a state which we interpret as an interacting vacuum state.

In the construction of models containing vector fields one uses a covariant inner product which is not positive definite \cite{scharf2016gauge}. Consequently, in such models the above-mentioned functional is indefinite and in particular is not a state. However, in the case of models with vector fields besides the algebra of interacting fields one considers also the algebra of observables. In fact, only the latter algebra has a direct physical interpretation. It is defined as a quotient algebra consisting of equivalence classes of gauge-invariant fields. The algebras of observables in quantum electrodynamics and non-abelian Yang--Mills theories without matter fields were constructed in \cite{dutsch1999local} and \cite{hollands2008renormalized}, respectively. Moreover, it was shown that each of the sub-algebras of these algebras consisting of observables localized in a~bounded open subset of the Minkowski spacetime can be represented on a pre-Hilbert space with a positive-definite inner product (in fact, in~\cite{hollands2008renormalized} this was proved even in the case of curved spacetime). The definition of a state on the global algebra of observables, in particular the definition of a Poincar{\'e}-invariant state, is an open problem. An interesting generalization of our results would be a construction of a Poincar\'e-invariant functional on the algebra of observables in the case of quantum electrodynamics and non-abelian Yang--Mills theories without matter fields, and a proof that resulting functional is positive.

The techniques developed in the proof of the existence of the weak adiabatic limit were also used to establish that one can define the time-ordered products in such a way that they satisfy the central normalization condition. This result implies, in particular, the existence of the central splitting solution in quantum electrodynamics. The time-ordered products of sub-polynomials of the interaction vertex of quantum electrodynamics normalized in this way are fixed uniquely and satisfy the standard normalization conditions, including the Ward identities.

\section*{Acknowledgements}
I would like to express my gratitude to Andrzej Herdegen for discussions and many helpful comments. I am grateful to the referees for careful reading of the manuscript and a number of insightful suggestions. I would also like to thank Michael D{\"u}tsch for his kind and useful remarks. This work was supported by the Polish Ministry of Science and Higher Education, grant number 7150/E-338/M/2017.
\appendix

\section{Second order of perturbation theory}\label{sec:mass}

In this this appendix we show that in models with massless particles, in contrast to purely massive models, the existence of the weak adiabatic limit requires appropriate normalization of some of the time-ordered products. More specifically, we prove that the weak adiabatic limit does not exist in the second order of the perturbation theory unless the self-energies of all massless fields (which are used in the definition of a given model) are normalized in such a way that the physical masses of these fields vanish. Note that this condition is a~part of the normalization condition \ref{norm:wAL}, which is needed in our proof of the existence of the weak adiabatic limit.

For simplicity we consider the model defined in terms of two real scalar fields: a~massive field $\psi$ and a~massless field $\varphi$ with the interaction vertex $\mathcal{L}=\frac{1}{2}\psi^2\varphi$ and the coupling constant $e$. By the result of Section \ref{sec:PROOF} it is possible to normalize the time-ordered products such that the weak adiabatic limit exists in the model. In particular, the correction of order $e^2$ to the Green function
\begin{equation}\label{eq:app_green}
 \int\rd^4 x_1\rd^4 x_2 \,f(x_1,x_2)\,\Gre(\varphi(x_1),\varphi(x_2))
\end{equation}
is certainly well defined for any $f\in\cS(\R^8)$ if the Fourier transform of the distribution
\begin{equation}\label{eq:app1_vev}
 (\Omega|\T(\psi^2(x_1),\psi^2(x_2))\Omega),
\end{equation}
which (up to a multiplicative constant) is the second order correction to the self-energy of the field $\varphi$, is of the form
\begin{equation}\label{eq:app1_sigma}
 (2\pi)^4\delta(q_1+q_2) \Sigma(q_1),
\end{equation}
where $\Sigma(0)=0$ and $\partial_\mu \Sigma(0)=0$. The distribution \eqref{eq:app1_sigma} coincides in some neighborhood of zero with the Fourier transform of $(\Omega|\Adv(\psi^2(x_1);\psi^2(x_2))\Omega)$. Since this distribution involves only massive fields, the function $\Sigma:\,\R^4\to\C$ is analytic in some neighborhood of the origin \cite{epstein1973role}. The second-order correction to \eqref{eq:app_green}, obtained by taking the weak adiabatic limit, is proportional to
\begin{equation}
 \int\mP{k}\,\F{f}(k,-k)\frac{1}{k^2+\ri\zerop}\Sigma(k)\frac{1}{k^2+\ri\zerop}.
\end{equation}
The above expression is well defined for all $f\in\cS(\R^8)$ since $(k^2+\ri\zerop)^{-2}$ is well defined as a distribution for all test functions vanishing at the origin. The distribution $(k^2+\ri\zerop)^{-2}$ can be extended to the space of all test functions, but the extension is not unique. This indicates that the weak adiabatic limit does not exist if $\Sigma(0)\neq 0$.

In what follows, we shall show that the condition $\Sigma(0)=0$, which expresses the correct mass normalization of the field $\varphi$, is in fact necessary for the existence of the weak adiabatic limit in the model in the second order of the perturbation theory. To this end, let us redefine the VEV of the time-ordered product \eqref{eq:app1_vev} by adding to it $c\,\delta(x_1-x_2)$ with $c\neq 0$. Note that after this redefinition $\Sigma(0)=c\neq 0$. In the EG approach the second-order contribution to \eqref{eq:app_green} is defined as the value at $q_1=q_2=0$ in the sense of {\L}ojasiewicz of the Schwartz distributions
\begin{equation}\label{eq:app_adv}
 \int\rd^4 x_1\rd^4 x_2\,f(x_1,x_2)\,(\Omega|\Adv(\F{\cL}(q_1),\F{\cL}(q_2);\varphi(x_1),\varphi(x_2))\Omega)
\end{equation}
or
\begin{equation}\label{eq:app_ret}
 \int\rd^4 x_1\rd^4 x_2\,f(x_1,x_2)\,(\Omega|\Ret(\F{\cL}(q_1),\F{\cL}(q_2);\varphi(x_1),\varphi(x_2))\Omega).
\end{equation}
The differences between the contributions to \eqref{eq:app_adv} and \eqref{eq:app_ret} before and after the redefinition of \eqref{eq:app1_vev} are given by
\begin{multline}
 d^\pm(q_1,q_2) = 2c \int\mP{k_1}\mP{k_2}\,\F{f}(k_1,k_2)
 \\ 
 \left[\frac{\ri}{k_1^2+\ri\zerop}\frac{\ri}{k_2^2+\ri\zerop}
 -(2\pi)\theta(\pm k_1^0)\delta(k_1^2)(2\pi)\theta(\pm k_2^0)\delta(k_2^2)\right](2\pi)^4\delta(k_1+k_2+q_1+q_2),
\end{multline}
where $+$ and $-$ correspond to \eqref{eq:app_adv} and \eqref{eq:app_ret}, respectively. The difference between the contributions to 
\begin{equation}
 \int\rd^4 x_1\rd^4 x_2\,f(x_1,x_2)\,(\Omega|\Dif(\F{\cL}(q_1),\F{\cL}(q_2);\varphi(x_1),\varphi(x_2))\Omega)
\end{equation}
before and after the redefinition of \eqref{eq:app1_vev} is given by
\begin{multline}
 d(q_1,q_2)
 =2c\int\mH{0}{k_1}\mH{0}{k_2}\,\F{f}(k_1,k_2)\,(2\pi)^4\delta(k_1+k_2+q_1+q_2)
 \\
 - 2c\int\mH{0}{k_1}\mH{0}{k_2}\,\overline{\F{f}(k_1,k_2)}\,(2\pi)^4\delta(k_1+k_2-q_1-q_2),
\end{multline}
where $\rd\mu_0(k) = \frac{1}{(2\pi)^3} \rd^4 k\,\theta(k^0) \delta(k^2)$. If $\F{f}(0,0)=0$, then 
\begin{equation}
 |d(q_1,q_2)|\leq\const\, |(q_1,q_2)|.
\end{equation}
Since $d$ is continuous we have $d(q_1,q_2)=O^\mathrm{dist}(|(q_1,q_2)|)$. Consequently, using Theorem \ref{thm:math_splitting} and the decomposition \eqref{eq:decomposition_weak_massive} one shows that the distributions $d^\pm(q_1,q_2)$ have values at $q_1=q_2=0$ in the sense of {\L}ojasiewicz and these values coincide. As a~result, it is enough to consider $f\in\cS(\R^8)$ such that $\F{f}=\const\neq0$ in some neighborhood of $0$. In this case we have 
\begin{equation}
 d(q_1,q_2) =\const\, \sgn(q_1^0+q_2^0) \theta((q_1+q_2)^2)
\end{equation}
in a neighborhood of $0$ (the constant in the above expression is nonzero). It is evident that the distribution $d(q_1,q_2)$ does not have value at $0$ in the sense of {\L}ojasiewicz. Thus, if $\F{f}(0,0)\neq0$, then at most one of the distributions \eqref{eq:app_adv} and \eqref{eq:app_ret} has value at $0$. Actually, none of these distributions has value at $0$ if $\F{f}(0,0)\neq0$. Indeed, since we know that the distributions \eqref{eq:app_adv} and \eqref{eq:app_ret} have value at $0$ if $\F{f}(0,0)=0$ it is enough to consider $f$ such that $\F{f}(0,0)\neq0$ and $\F{f}(k_1,k_2)=\F{f}(-k_1,-k_2)$. For such $f$ it holds $d^+(q_1,q_2)=d^-(-q_1,-q_2)$ and it is not possible that exactly one of the distributions $d^\pm$ has value at $0$, hence none of them has value at $0$. This implies that the weak adiabatic limit does not exist after the above-mentioned redefinition of the distribution~\eqref{eq:app1_vev}.

\section{Interaction vertices of different dimensions}\label{sec:general}

In this appendix we investigate models with the interaction vertices of different canonical dimensions. We describe modifications in the basic assumptions about the time-ordered products which facilitate the proof of the existence of the weak adiabatic limit in the case of such models. 

As an example, let us first consider the model defined in terms of two real scalar fields: a~massive field $\psi$ and a~massless field $\varphi$ with the interaction vertices: $\cL_1=\frac{1}{4!}\varphi^4$, $\cL_2=\frac{1}{2}\varphi^2\psi$. We have $\dim(\cL_1)=4$, $\dim(\cL_2)=3$. Note that models which contain at least one vertex of dimension~$4$ are renormalizable only if $\CC=0$ in the axiom \ref{axiom7}. However, if there are vertices of dimension strictly less than~$4$, then the axiom \ref{axiom7} with $\CC=0$ is too restrictive and cannot be satisfied simultaneously with the normalization condition \ref{norm:wAL}, which is needed in our proof of the existence of the weak adiabatic limit. It is possible to show (using the method of the proof of Theorem \ref{thm:main1}) that \ref{norm:wAL} can be fulfilled if the axiom \ref{axiom7} is modified in such a way that it acquires the form \eqref{eq:sd_bound2} in the case of the VEVs of the sub-polynomials of the interaction vertices for any super-quadri-indices $u_1,\ldots,u_n$ involving only massless fields. In the present model this may be achieved by replacing \ref{axiom7} with the following bound on the scaling degree
\begin{equation}
\sd(\, (\Omega|\T(B_1(x_1),\ldots,B_{k-1}(x_{k-1}),B_k(0))\Omega)\,) \leq \sum_{j=1}^k \mathrm{d}(B_j),
\end{equation}
where $\mathrm{d}(B)=\dim(B) + 1$ if $B$ contains the field $\psi$ and $\mathrm{d}(B)=\dim(B)$ otherwise. However, if we added $\cL_3=\frac{1}{4!}\psi^4$ to the list of the interaction vertices of the model, then the model would not be renormalizable under present assumptions since $\mathrm{d}(\cL_3)=5$. As we will indicate the above problem may be avoided by weakening the axiom \ref{axiom3}.

Now let the interaction vertices $\cL_1,\ldots,\cL_{\mathrm{q}}$ of a model under consideration be such that $\dim(\cL_l)\leq 4$ for any $l\in\{1,\ldots,\mathrm{q}\}$. If $\dim(\cL_l)\leq 3$ we demand in addition that $\cL_l$ contains at least one massive field. We will argue that after appropriate modification of the axioms \ref{axiom3} and \ref{axiom7} it is possible to normalize the time-ordered products such that they satisfy the normalization condition \ref{norm:wAL}. We first replace the condition \eqref{eq:T_expansion} in the axiom \ref{axiom3} by
\begin{multline}
 \T(B_1(x_1),\ldots,B_k(x_k)) 
 =
 \\
 \sum_{s_1,\ldots,s_k} t^{s_1,\ldots,s_k}(B_1(x_1),\ldots,B_k(x_k))\, 
 ~\frac{\normord{A^{s_1}(x_1)\ldots A^{s_k}(x_k)}}{s_1!\ldots s_k!},
\end{multline}
where $t^{s_1,\ldots,s_n}(B_1(x_1),\ldots,B_n(x_n))$ is some translationally invariant Schwartz distribution for every super-quadri-indices $s_1,\ldots,s_n$ and $B_1,\ldots,B_n\in\Fa$. Next, we impose the following condition as a substitute of the axiom \ref{axiom7} 
\begin{multline}
 \sd(\, t^{s_1,\ldots,s_k}(B_1(x_1),\ldots,B_{k-1}(x_{k-1}),B_k(0))\,) 
 \\
 \leq \sum_{j=1}^k \textbf{d}(B_j) - \sum_{i=1}^{\mathrm{p}} (\ext_{\mathbf{s}}(A_i) \dim(A_i) + \der_{\mathbf{s}}(A_i)),
\end{multline}
where $\mathbf{d}:\Fh\to \frac{1}{2}\Z$ is a~function such that $\mathbf{d}(B)\geq\dim(B)$ and $\mathbf{d}(cB)=\mathbf{d}(B)$ for any $c\in\C$ and $B\in\Fh$. It can be proved that if the time-ordered products with at most $n$ arguments satisfy the modified axioms then it is possible to define the time-ordered products with $n+1$ arguments which again satisfy them. If we chose the map $\mathbf{d}$ such that $\mathbf{d}(\cL_l)=4$ for all $l\in\{1,\ldots,\mathrm{q}\}$, then one can prove, along the lines of Theorem \ref{thm:main1}, that the normalization freedom allows to define the time-ordered products such that the condition \ref{norm:wAL} holds. The existence of the weak adiabatic limit follows then from Theorem \ref{thm:main2}. 

\section{Gell-Mann and Low formula}\label{sec:magic_formula}

In the standard approach to the quantum field theory the Green functions are defined by the following formula due to Gell-Mann and Low \cite{gell1951bound}
\begin{multline}
 \Gre(C_1(x_1),\ldots,C_m(x_m)) =
 \\
 \frac{\sum_{n=0}^\infty \frac{\ri^ne^n}{n!} \int\rd^4 y_1\ldots\rd^4 y_n\,(\Omega|\T(I_n,C_1(x_1),\ldots,C_m(x_m))\Omega)}{\sum_{n=0}^\infty \frac{\ri^ne^n}{n!} \int\rd^4 y_1\ldots\rd^4 y_n\,(\Omega|\T(I_n)\Omega)},
\end{multline}
where $I_n:=(\mathcal{L}(y_1),\ldots,\mathcal{L}(y_n))$. For simplicity, we assumed that there is only one interaction vertex. The Gell-Mann and Low formula is not well defined in its standard form stated above because of the presence of the formal integrals. For this reason, we consider instead the IR-improved Gell-Mann and Low formula, already introduced in \cite{dutsch1997slavnov}, which states that
\begin{equation}\label{eq:app_G_GL_lim}
 \Gre^{\textrm{GL}}(C_1(x_1),\ldots,C_m(x_m)) := \lim_{\epsilon\searrow0} \Gre^{\textrm{GL}}_\epsilon(C_1(x_1),\ldots,C_m(x_m)),
\end{equation}
where
\begin{multline}\label{eq:app_G_GL_def}
 \Gre^{\textrm{GL}}_\epsilon(C_1(x_1),\ldots,C_m(x_m)) :=
 \\
 \frac{\sum_{n=0}^\infty\! \frac{\ri^ne^n}{n!}\! \int\rd^4 y_1\ldots\rd^4 y_n g_\epsilon(y_1)\ldots g_\epsilon(y_n)(\Omega|\T(I_n,C_1(x_1),\ldots,C_m(x_m))\Omega)}{\sum_{n=0}^\infty \frac{\ri^ne^n}{n!} \int\rd^4 y_1\ldots\rd^4 y_n\,g_\epsilon(y_1)\ldots g_\epsilon(y_n)\,(\Omega|\T(I_n)\Omega)}.
\end{multline}
The above distribution is well defined for any $\epsilon>0$, but the existence of the limit on the RHS of \eqref{eq:app_G_GL_lim} is non-trivial. 

We will show that for models satisfying Assumption \ref{asm} the limit \eqref{eq:app_G_GL_lim} exists and coincides with the EG definition of the corresponding Green function. This is a generalization of the results of  \cite{dutsch1997slavnov}, where a similar statement was proved for purely massive models. Let us first note that
\begin{multline}\label{eq:app_G_GL}
 \Gre^{\textrm{GL}}_\epsilon(C_1(x_1),\ldots,C_m(x_m))=
 \\
 (-\ri)^{m} \frac{\delta}{\delta h_m(x_m)}\ldots\frac{\delta}{\delta h_1(x_1)} 
  \frac{(\Omega|S(\mathbf{g}_\epsilon+\mathbf{h})\Omega)}{(\Omega|S(\mathbf{g}_\epsilon)\Omega)}  \bigg|_{\mathbf{h}=0}, 
\end{multline}
where $S(\mathbf{g})$ is given by \eqref{eq:def_Texp}, $\mathbf{g}_\epsilon=e g_\epsilon \otimes \cL$ and $\mathbf{h}=\sum_{j=1}^m h_j\otimes C_j$. On the other hand, the Green functions in the EG approach are given by
\begin{equation}
 \Gre^\rEG(C_1(x_1),\ldots,C_m(x_m)) := \lim_{\epsilon\searrow0} \Gre^\rEG_\epsilon(C_1(x_1),\ldots,C_m(x_m)),
\end{equation}
where
\begin{multline}\label{eq:app_G_EG}
 \Gre^\rEG_\epsilon(C_1(x_1),\ldots,C_m(x_m)):=
 \\
 (-\ri)^{m} \frac{\delta}{\delta h_m(x_m)}\ldots\frac{\delta}{\delta h_1(x_1)} 
 (\Omega|S(\mathbf{g}_\epsilon+\mathbf{h}) S(\mathbf{g}_\epsilon)^{-1}\Omega)  \bigg|_{\mathbf{h}=0}. 
\end{multline}

\begin{thm}
Suppose that a model under consideration satisfies Assumption \ref{asm} formulated in Section~\ref{sec:PROOF}. If the time-ordered products fulfill the normalization condition~\ref{norm:wAL}, then it holds
\begin{equation}\label{eq:app_thm}
 \lim_{\epsilon\searrow0} \Gre^{\mathrm{GL}}_\epsilon(C_1(x_1),\ldots,C_m(x_m))
 =
 \lim_{\epsilon\searrow0}\Gre^\rEG_\epsilon(C_1(x_1),\ldots,C_m(x_m)) ,
\end{equation}
where the limit on the RHS of the above equation exists by Theorem \ref{thm:main2} and defines the Green function in the EG approach.
\begin{proof}
Consider the expression
\begin{multline}\label{eq:GL_diff}
  (-\ri)^{m} \frac{\delta}{\delta h_m(x_m)}\ldots\frac{\delta}{\delta h_1(x_1)}
  \\
  \left[ (\Omega|S(\mathbf{g}_\epsilon+\mathbf{h})\Omega)  - (\Omega|S(\mathbf{g}_\epsilon+\mathbf{h})S(\mathbf{g}_\epsilon)^{-1}\Omega)\,(\Omega|S(\mathbf{g}_\epsilon)\Omega) \right] \bigg|_{\mathbf{h}=0}.
\end{multline}
It follows from \eqref{eq:time_ordered_adv} and \eqref{eq:def_Texp} and 
\begin{equation}
 (\Omega|S(\mathbf{g}_\epsilon+\mathbf{h})\Omega) = (\Omega|S(\mathbf{g}_\epsilon+\mathbf{h})S(\mathbf{g}_\epsilon)^{-1}S(\mathbf{g}_\epsilon)\Omega)
\end{equation}
that the coefficient of order $e^n$ of the formal power series \eqref{eq:GL_diff} is a linear combination of terms of the form
\begin{multline}
 \int\rd^4 y_1\ldots\rd^4 y_n\,g_\epsilon(y_1)\ldots g_\epsilon(y_n)  
 \\
 [ (\Omega|\Adv(I'_n;J)\T(I''_n)\Omega) -
 (\Omega|\Adv(I'_n;J)\Omega)\,(\Omega|\T(I''_n)\Omega)],
\end{multline}
where the concatenation of the sequences $I'_n$ and $I''_n$ is some permutation of the sequence $I_n=(\mathcal{L}(y_1),\ldots,\mathcal{L}(y_n))$ and $J=(C_1(x_1),\ldots,C_m(x_m))$. 

Let us show that 
\begin{equation}\label{eq:GL_dist}
 (\Omega|\Adv(I'_n;J)\T(I''_n)\Omega) -
 (\Omega|\Adv(I'_n;J)\Omega)\,(\Omega|\T(I''_n)\Omega)
\end{equation}
has IR-index $d=1$ with respect to all variables $y_1,\ldots,y_n$. To this end, we use the following variant of Theorem \ref{thm:product_F}: Assume that the $\mathcal{F}$ products \eqref{eq:product_families} satisfy the assumptions (1)--(3) of Theorem \ref{thm:product_F} for all super-quadri-indices $s_1, \ldots, s_n$, $s'_1, \ldots, s'_{n'}$ which involve only massless fields such that at least one of them is nonzero, then the distribution
\begin{multline}
 (\Omega|F(B_1(x_1),\ldots,B_n(x_n))F'(B_1^{\prime}(x'_1),\ldots,B_{n'}^{\prime}(x'_{n'}))\Omega) 
 \\
 - (\Omega|F(B_1(x_1),\ldots,B_n(x_n))\Omega)\,(\Omega|F'(B_1^{\prime}(x'_1),\ldots,B_{n'}^{\prime}(x'_{n'}))\Omega)
\end{multline}
fulfills the conditions (1')-(3') stated in this theorem. The above statement about \eqref{eq:GL_dist} follows from Part~(A) of Theorem \ref{thm:main2}, the normalization condition \ref{norm:wAL2} with $F$ being the time-ordered product and the above reformulation of Theorem \ref{thm:product_F} (more precisely we use the statement (3')).

Using the above result, Definition \ref{def:IR2} of the IR-index and Part (A) of Theorem~\ref{thm:math_adiabatic_limit} we obtain that all coefficients in the formal expansion in powers of the coupling constant $e$ of \eqref{eq:GL_diff} after integrating with arbitrary Schwartz function $f(x_1,\ldots,x_m)$ are of order $O(\epsilon^{1-\varepsilon})$ for any $\varepsilon>0$. It follows from the normalization condition \ref{norm:wAL} that
\begin{equation}
 (\Omega|S(\mathbf{g}_\epsilon)\Omega) =  O(\epsilon^{-\varepsilon})
\end{equation}
for any $\varepsilon>0$. The above equation means that all coefficients in the formal expansion in powers of the coupling constant $e$ of the LHS are of order $O(\epsilon^{-\varepsilon})$ for any $\varepsilon>0$. Since the first term of this expansion is equal $1$ we have
\begin{equation}
 (\Omega|S(\mathbf{g}_\epsilon)\Omega)^{-1} =  O(\epsilon^{-\varepsilon})
\end{equation}
for any $\varepsilon>0$. Consequently, each term of the $e$ expansion of the difference between the RHS of Equations \eqref{eq:app_G_GL} and \eqref{eq:app_G_EG} integrated with arbitrary Schwartz function $f(x_1,\ldots,x_m)$ is of the order $O(\epsilon^{1-2\varepsilon})$ for any $\varepsilon>0$. This finishes the proof.
\end{proof}
\end{thm}

If Assumption \ref{asm} is satisfied and the interaction vertex contains at least one massive field, then according to Part~(2) of the normalization condition \ref{norm:gc} in Section~\ref{sec:central} the distribution
\begin{equation}\label{eq:app_T}
 (\Omega|\T(\mathcal{L}(x_1),\ldots,\mathcal{L}(x_n))\Omega)
\end{equation}
can be normalized in such a way that its \underline{IR}-index equals $d=5$. As a result 
\begin{equation}
 (\Omega|S(\mathbf{g}_\epsilon)\Omega) -1 =  O(\epsilon^{1-\varepsilon})
\end{equation}
for any $\varepsilon >0$ and both sides of \eqref{eq:app_thm} are equal to
\begin{equation}
 \lim_{\epsilon\searrow0} 
 ~(-\ri)^{m} \frac{\delta}{\delta h_m(x_m)}\ldots\frac{\delta}{\delta h_1(x_1)} 
 (\Omega|S(\mathbf{g}_\epsilon+\mathbf{h})\Omega)  \bigg|_{\mathbf{h}=0}.
\end{equation}
This shows that for the above-mentioned normalization of \eqref{eq:app_T} the denominators in \eqref{eq:app_G_GL_def} and \eqref{eq:app_G_GL} may be omitted.

\sloppy

\printbibliography

\end{document}